\documentclass[12pt]{report}

\usepackage{layouts}

\usepackage{amssymb}
\usepackage{amsmath}
\usepackage{amsfonts} 
\usepackage{array}
\usepackage{url}
\usepackage[monochrome]{color}
\usepackage{soul}
\usepackage{caption}
\usepackage{tabularx}                       
\usepackage{multirow}
\usepackage{multicol}
\usepackage{graphicx}                       

\usepackage[square,sort,compress,comma,numbers]{natbib}
\usepackage{indentfirst}                    
\usepackage{afterpage}                      
\usepackage{algorithm}
\usepackage{algorithmic}

\usepackage{xspace}



\usepackage[para,online,flushleft]{threeparttable}
\usepackage{booktabs}
\usepackage{dcolumn}
\usepackage{longtable}
\usepackage{threeparttablex}

\let\OLDitemize\itemize
\let\OLDenumerate\enumerate
\let\OLDdescription\description

\renewcommand{\itemize}[1][\relax]{\OLDitemize}
\renewcommand{\enumerate}[1][\relax]{\OLDenumerate}
\renewcommand{\description}[1][\relax]{\OLDdescription}

\usepackage{tikz}
\usepackage{tkz-tab}
\usetikzlibrary{shapes}
\usetikzlibrary{positioning}

\newtheorem{proof}{Proof}
\newtheorem{lemma}{Lemma}
\newtheorem{theorem}{Theorem}

\usepackage{framed}

\usepackage{footnote}

\usepackage{subfigure}  
\usepackage{epstopdf}
\usepackage[left=1.25in, right=1.25in, top=1in, bottom=1in]{geometry}
\newcommand{\setlinespacing}[1]%
           {\setlength{\baselineskip}{#1 \defbaselineskip}}

\newcommand{\singlespacing}{\setlength{\baselineskip}{\defbaselineskip}}

\usepackage{etoolbox} 
\newcommand{\circled}[2][]{%
        \tikz[baseline=(char.base)]{%
                \node[shape = circle, draw, inner sep = 1pt]
                (char) {\phantom{\ifblank{#1}{#2}{#1}}};%
                \node at (char.center) {\makebox[0pt][c]{#2}};}}
\robustify{\circled}

\newcommand{\dacircled}[2][]{%
        \tikz[baseline=(char.base)]{%
                \node[shape = circle, draw, dashed, inner sep = 1pt]
                (char) {\phantom{\ifblank{#1}{#2}{#1}}};%
                \node at (char.center) {\makebox[0pt][c]{#2}};}}
\robustify{\dacircled}

\renewcommand{\baselinestretch}{1}

\renewcommand{\contentsname}{\normalsize Table of Contents}
\renewcommand{\listfigurename}{\normalsize List of Figures}
\renewcommand{\listtablename}{\normalsize List of Tables}
\renewcommand{\abstractname}{\large Abstract}
\renewcommand{\appendixname}{\normalsize Appendix}

\renewcommand{\appendixname}{\large Appendix}

\newlength{\defbaselineskip}
\usepackage{titlesec}

\usepackage[subfigure]{tocloft}%

\titlespacing*{\chapter}{0pt}{-0.32in}{25pt}

\titleformat{\chapter}[hang]{\centering\normalsize\bfseries}{\chaptertitlename\ \thechapter}{1em}{} 
\titleformat{\section}{\normalsize\bfseries}{\thesection}{1em}{}
\titleformat{\subsection}{\normalsize\bfseries}{\thesubsection}{1em}{}
\titleformat{\subsubsection}{\normalsize\it}{\thesubsubsection}{1em}{}

\def\XXint#1#2#3{{\setbox0=\hbox{$#1{#2#3}{\int}$}
		\vcenter{\hbox{$#2#3$}}\kern-.5\wd0}}

\usepackage{graphics}
\usepackage{subfig}                         
\usepackage{epsfig}
\usepackage{float}
\usepackage{slashbox}

\usepackage{texlogos}
\usepackage{mathtools}
\usepackage{breqn}
\DeclarePairedDelimiter\floor{\lfloor}{\rfloor}

\usepackage[utf8]{inputenc}
\usepackage[T1]{fontenc}
\usepackage{lmodern}
\usepackage{microtype}
\usepackage{bm}

\usepackage{lscape}

\begin{document}

	\pagenumbering{roman} \vspace{-0.2in}
	\thispagestyle{empty}
 \hbox{\vspace{0.5in}}
	\begin{center}

	\begin{singlespacing}
		{\normalsize\bfseries Resource Allocation in Multigranular Optical Networks}\\[4\baselineskip]
 
	by Jingxin Wu \\[3\baselineskip]
	B.S. in School of Information Science and Engineering, July 2010, Harbin Institute of Technology \\
	M.S. in School of Information Science and Engineering, July 2012, Harbin Institute of Technology \\[2\baselineskip]
	A Dissertation submitted to \\[3\baselineskip] %
	The Faculty of \\
	The School of Engineering and Applied Science \\
	of The George Washington University \\
	in partial satisfaction of the requirements \\
	for the degree of Doctor of Philosophy \\[3\baselineskip]
	August 31, 2018 \\[4\baselineskip]
	Dissertation directed by \\[2ex]
	Suresh Subramaniam \\
	Professor of Engineering and Applied Science
	
\end{singlespacing}
\end{center}

	\newpage

\hbox{\ }
\vspace{-0.3in}
\renewcommand{\baselinestretch}{1.5}
\noindent
The School of Engineering and Applied Science of The George Washington University certifies that Jingxin Wu has passed the Final Examination for the degree of Doctor of Philosophy as of August 13, 2018. This is the final and approved form of the dissertation.\\
\normalsize

\begin{center} 
	\singlespacing 
	\ \\
	Resource Allocation in Multigranular Optical Networks \\[1\baselineskip]
	Jingxin Wu\\[3\baselineskip]
\end{center}

\singlespacing \noindent Dissertation Research Committee: \\

\singlespacing{
	\setlength{\leftskip}{0.5in}
	\noindent Suresh Subramaniam, Professor of Engineering and Applied Science, Dissertation Director\\
	
	\noindent Tian Lan, Associate Professor of Engineering and Applied Science, Committee Member\\
	
	\noindent Milo{\v{s}} Doroslova{\v{c}}ki, Associate Professor of Engineering and Applied Science, Committee Member\\
	
	\noindent Guru Venkataramani, Associate Professor of Engineering and Applied Science, Committee Member\\
	
	
	\noindent Hiroshi Hasegawa, Associate Professor of The Department of Electrical Engineering and Computer Science, Nagoya University, Committee Member\\
	
}

\newpage
	\addcontentsline{toc}{chapter}{Dedication}
\hbox{\ }
\renewcommand{\baselinestretch}{1.5}
\small\normalsize\begin{center}
	\vspace{-0.5in}
	\textbf{Dedication} \\[1\baselineskip]
	\normalsize
\end{center}

In memory of the departed loved ones.\\

To my family, Lab634 and SEH 5390. \\
	
\newpage
	\addcontentsline{toc}{chapter}{Acknowledgements}
\hbox{\ }
\renewcommand{\baselinestretch}{1.5}
\normalsize\begin{center}
	\vspace{-0.5in}
	\textbf{Acknowledgements}  \\[1\baselineskip]
	\normalsize
\end{center}

I would first like to thank my advisor, Professor Suresh Subramaniam, for his continuous support throughout the past six years. His expertise and immense knowledge on the networking area have always led me to the appropriate direction of my research work. His patience, motivation and advice have been essential for my research papers. He has also provided extensive professional and personal guidance and helped me a great deal. I feel privileged to have him as my advisor.

I wish to thank Professors Hiroshi Hasegawa, Tian Lan, Milos Doroslovacki and Guru Venkataramani for generously offering their time and guidance as my Ph.D. dissertation committee members. I would especially like to thank Professor Hasegawa and Professor Venkataramani, coauthors of my research papers, who are always willing to share their insightful comments.

I gratefully acknowledge the support of NSF (through grant CNS-1406971), which sponsored my research.

I sincerely thank my friends, Yu Xiang, Juzi Zhao, Yifan Chen, Fan Yao, Chong Liu, Yongbo Li, Maotong Xu, Hang Liu, Xiaolu Song, Ming Yao, Shijing Li, Jie Chen, Bingqian Lu, Sultan Alamro and all others, for their friendship and for supporting me through the hard times of my life and making my stay in DC wonderful.

Last, but not the least, I would like to thank my family to whom I owe a lot. Thanks to my beloved parents, grandmother for the continuous and selfless love. I am also grateful to my brother for the moral support and constant encouragement.
	
	\newpage
	
	\addcontentsline{toc}{chapter}{Abstract}

\hbox{\ }
\renewcommand{\baselinestretch}{1.5}
\normalsize\begin{center}
	\vspace{-0.435in}
	\textbf{Abstract}  \\[1\baselineskip]
	\normalsize
	\singlespacing Resource Allocation in Multigranular Optical Networks
	\ \\
\end{center}

\begin{center}
	\noindent\fbox{%
		\parbox{\linewidth}{%
			\textbf{Thesis Statement: Cost-effective switching and spectrum utilization efficiency have become critical design considerations in optical networks. This dissertation provides in-depth exploration of these important aspects, and proposes effective techniques for low-cost switching architectures and resource allocation algorithms to facilitate the adoption of optical networks in the near future.}
		}%
	}
\end{center}

The dramatic growth of Internet traffic brings challenges for optical network designers. The increasing traffic and bandwidth requirements mean that various resource allocation schemes to achieve different network design goals assume great importance. The general problem of resource allocation to lightpath requests is a challenging problem.

\indent An emerging technology of flexible and more fine-grained grid through the use of Optical Orthogonal Frequency Division Multiplexing (OOFDM) allows fiber bandwidth to be more suitably matched up with application requirements, thereby making the network more elastic than the conventional Wavelength Division Multiplexing (WDM) optical networks. Despite the advances of employing OOFDM technology in elastic optical networks (EONs), imminent fiber capacity exhaustion due to the ever-increasing demands means that multiple fibers per link will be inevitable. While increasing the number of fibers boosts the capacity of networks, there is a price to pay for it in the form of increased number of switch ports / complexity of switches. The huge amount of traffic demands and thus high hardware requirements motivate multigranularity (such as wavebanding) to save costs in optical networks. This dissertation aims to tackle several types of resource allocation challenges in multi-granular optical networks to either improve the spectrum utilization or provide cost-effective switching techniques. 


\indent The first part of this dissertation presents our work on an important joint scheduling (or {\em co-scheduling} of computational and network resources) problem in today's applications such as cloud computing. Such applications involve the processing of complex jobs consisting of several inter-dependent tasks executing on heterogeneous clusters of computing resources, which are interconnected by high-speed optical networks. Our co-scheduling algorithms take the constraints introduced by elastic optical networks into account and aim at either minimizing the makespan of a set of jobs in the static case, or minimizing the job blocking when jobs arrive dynamically.

\indent The second part of the dissertation addresses problems on waveband switching in WDM networks. Grouping together a set of consecutive wavelengths in a WDM network and switching them together as a single waveband could achieve savings in switching costs of an optical cross-connect (OXC). A specific problem that we consider is the optimal Band Minimization Problem in WDM mesh networks to minimize the required number of switching elements while accommodating a set of traffic demands. 
Then we evaluate OXC node architectures in WDM networks with multiple parallel fibers. A hierarchical architecture that utilizes wavebanding and has lower complexity is compared with the conventional architecture in terms of the cost of the OXC node and the power consumption. Analytical models for computing the blocking probability of connection requests are proposed and validated.

\indent In an effort to reduce the complexity of optical OXCs, a flexible wavebanding OXC architecture has been proposed recently. Elastic networking and flexible wavebanding introduce a new problem, namely, the routing, fiber, waveband, and spectrum assignment (RFBSA) problem. In the third part of the dissertation, a framework for solving the RFBSA problem in networks with multi-fiber links that can accommodate non-contiguous and non-uniform wavebanding is proposed. Then a joint banding-node placement and RFBSA problem to meet the network budget while maintaining good network performance is addressed.

\indent The final part of the dissertation addresses dynamic Routing and Spectrum Assignment problem in multi-fiber EONs to further improve spectrum efficiency. We propose and evaluate novel schemes to minimize the demand blocking ratio of requests.
\newpage
	
	\setcounter{tocdepth}{3}
	\setcounter{secnumdepth}{3}
	\tableofcontents
	
	\addcontentsline{toc}{chapter}{List of Figures}
	\clearpage
	\listoffigures
	\addcontentsline{toc}{chapter}{List of Tables}
	\clearpage
	\listoftables 
	\chapter{Introduction}
\pagenumbering{arabic}
\label{intro}

\indent The continuous growth of Internet traffic demands brings challenges for optical network designers. As traffic demands continue to grow and optical networks become even more pervasive, cost-efficient networking techniques become even more essential. Developing cost-efficient switching architectures and algorithms is necessary to reduce switching costs in the network and achieve faster adoption of optical infrastructures so as to facilitate the availability of high-speed networking. Given the limited network bandwidth capacity and ever-increasing traffic demands, the other important problem is to provision the requests with high spectrum efficiency. Thus the resource management problems in optical networks are nontrivial. Different resource allocation schemes would lead to different network performance. How to efficiently utilize various resources (such as wavelengths/subcarriers) in the system as well as meet different constraints (such as demand completion time, hardware requirements, blocking probability of connections) is indispensable in network design.

This dissertation aims at tackling these important aspects, and developing effective techniques for low-cost switching architectures and resource allocation algorithms for the next generation optical networks. 

In this chapter, we first briefly review the technologies used in this work, and outline our contributions.

\section{Optical Networks}\label{sec.label11}

\indent Optical networking technologies have been widely introduced and deployed at the Internet's core networks. In the foreseeable future, optical networks will form the underlying physical infrastructure. There is a dramatic growth of Internet traffic, which results from emerging applications, such as livestreams and social networking. According to \cite{website:ciscoInternet}, the global Internet traffic is expected to be more than 60 TBps in 2020. The need for optical-electrical-optical conversions can be eliminated by reconfigurable optical add/drop multiplexers (ROADMs) or optical cross-connects (OXCs). The enormous bandwidth demands can be carried by such all-optical paths between their sources nodes and destination nodes. 

\subsection{WDM Optical Networks}\label{sec.label12}

\indent Wavelength Division Multiplexing (WDM) technology multiplexes information to be transmitted onto a large number of wavelengths. By utilizing the dense wavelength division multiplexing, multiple wavelengths/optical paths can be accommodated within a fiber. Those wavelengths are located on a fixed spacing defined by ITU-T. OXCs are used as the switching nodes in current optical networks \cite{tanaka2014performance}. As the traffic grows, more wavelengths are needed in today's WDM networks. To cope with the growing number of wavelengths, large-scale OXCs are needed. In a wavelength-switched optical node, each wavelength is switched individually, that is, one switching element corresponds to one wavelength. In this case, the switching cost can be very high.

\subsection{OOFDM-based Optical Networks}\label{sec.label13}

\indent Optical orthogonal frequency division multiplexing (OOFDM) technology is proposed as an approach in \cite{OFDM1} to achieve high spectral efficiency and a flexible rate. A novel programmable mechanism for OOFDM which utilizes advanced digital signal processing, parallel signal detection, and flexible resource management schemes for subwavelength level multiplexing and grooming is presented in \cite{OFDM}. In OOFDM, the fiber bandwidth is carved up into subcarriers spaced only a few GHz apart with subcarrier bit rates of a few Gbps. This is to be contrasted with WDM where the typical wavelength spacing is 25 GHz, and wavelength bit-rates are 10, 40, or 100 Gbps. An attractive feature of OOFDM is that flexible bands of subcarriers may be assigned to a service as needed. In this way, allocated network resources can be matched up with service requirements in a much more flexible manner than in WDM-based networks. Such optical networks have therefore been called {\em elastic optical networks} \cite{JKTWSTH10}.

\subsection{Multi-fiber Links}
\indent The explosive use of video-based services, huge data transmission between data centers and so on causes continuous Internet traffic expansion. Imminent fiber capacity exhaustion means that multiple fibers per link will be inevitable. To accommodate increasing traffic demands, deploying multiple fibers on a physical link is a promising approach. Each link contains multiple fibers, and the number of fibers on each link may be different. All fibers consist of the same number of wavelengths or subcarriers (also called frequency slots or FSs) in WDM networks / EONs. At each OXC, a wavelength / FS on an input fiber can be switched to the same wavelength / FS on an output fiber.

\subsection{Multi-granular Optical Networks}\label{sec.label14}

\indent Grouping a bundle of wavelengths and switching them together as a single waveband (\cite{izmailov2002hybrid, wang2012multi}) will relieve the hardware requirements. It has been verified that such grouping of wavelengths can be realized by PLC technology (\cite{chandrasekhar2005flexible, kakehashi2007performance}) and the reduced number of switches enables us to monolithically integrate all the necessary de-multiplexing and switching functions on a single chip \cite{ishii2009ultra}. Savings in switching costs can be achieved by this waveband switching. For example, each mirror in a 3-D micro-electro-mechanical systems (MEMS) based wavelength selective switch (WSS) \cite{yuan2008fully} is dedicated to switching a wavelength path. If multiple wavelength paths are grouped and switched together, they will share the same mirror so that the number of mirrors needed in one OXC node will be decreased. 

\indent Such kind of multi-granularity could also be utilized in optical networks with multi-fiber links. Even though increasing the number of fibers improves the capacity of the network, it also increases the number of ports in OXCs, and thus increasing the complexity and hardware requirements of OXCs. By utilizing the wavebanding feature, the complexity of OXC architecture can be reduced with less hardware requirement. In EONs, the waveband switching corresponds to aggregating {\em a set} of optical spectral ranges and routing them with a single waveband port \cite{patel2012hierarchical}. Along with the benefits of wavebanding technique, the switching constraints are also introduced. The resource allocation problems with different objectives in such multi-granular optical networks need to be carefully addressed.

\section{Resource Allocation Problems}\label{sec.label15}

\indent This dissertation aims to tackle multiple types of resource allocation challenges in multi-granular optical networks to either improve the spectrum utilization or provide cost-effective switching techniques.

\subsection{Co-Scheduling Problem}\label{sec.label16}

\indent Current cloud computing \cite{alamro2016cred, limrungsi2012providing, xu2017optimizing, xu2017laser, xu2017cred, aggarwal2017optimality, zhao2017elastic, sultan2018Shed, xu2018Chronos, lu2012novel}, e-science, and data center applications \cite{yao2014comparative, li2018pushing, yao2017statsym, xue2017simber, wang2017capitalizing, li2015poster} consist of job requests that usually have multiple inter-dependent tasks, with each task having certain computational needs \cite{liu2016reconfigurable, li2018multichoice, yao2015dual,li2017mobiqor, yao2017cloud,li2016sarre,ts-bat,popcorns}. At the same time, the inter-task communication may not be negligible. Therefore, an important challenge is to efficiently schedule both the computational and networking resources, which is called {\em co-scheduling}. Co-scheduling in a dynamic, heterogeneous, distributed computing environment has been extensively studied for about two decades, e.g., \cite{SM94, RS01, SS05, VKSCKSS09, CCK12}. While co-scheduling problem in WDM networks has been widely explored, it hasn't been studied in the elastic optical networks. 
 
\subsection{Routing and Spectrum Assignment Problem}\label{sec.label17}

\indent One important problem in EONs is the Routing and Spectrum Assignment / Allocation (RSA) problem, which is evolved from the routing and wavelength assignment (RWA) problem in WDM networks. The objective of RSA is to find a number of unoccupied frequency slots (FSs) to meet traffic demands and establish lightpaths \cite{wang2011study,chatterjee2015routing, wu2015comparison}. The main constraints of the RSA problem are spectrum contiguity, spectrum continuity, and spectrum non-overlapping. The spectrum contiguity constraint ensures that the allocated FSs to a lightpath are contiguous. The spectrum continuity constraint ensures that the same FSs are allocated on every fiber along the route. The spectrum non-overlapping constraint ensures that any FS in any fiber is allocated to at most one lightpath. The RSA problem can be divided into two subproblems: routing problem and spectrum allocation (SA). For static RSA problem, the objective is to minimize the maximum slot index (MS) while provisioning all traffic requests. For dynamic RSA problem, the objective is to minimize the blocking ratio and demand blocking ratio of traffic requests.

\subsection{Routing, Fiber, Band, and Spectrum Assignment (RFBSA) Problem}\label{sec.label18}

\indent In a multi-granular elastic optical network with multi-fiber links, the RSA problem evolves into a routing, fiber, band, and spectrum assignment problem. The goal of RFBSA problem is to assign the route, fibers, wavebands and the contiguous slot set to accommodate each request. RSA is known to be an NP-complete problem in EONs. Adding the fiber and flexible waveband selection further increases the problem complexity in another two dimensions \cite{wu2017routing}. 

\section{Contributions}\label{sec.label19}

\indent In this dissertation, we first address the co-scheduling problem of computing and networking resources in a single-granular elastic optical networks for multi-task jobs. Both static and dynamic versions of the multi-job multi-task co-scheduling problem in OOFDM-based elastic optical networks are considered. Each job request includes a set of tasks, each task has a workload requirement, and there are communication requirements between tasks. The objective is to minimize the makespan (finishing time of all jobs) of the set of jobs. This scheduling problem is clearly NP-hard, as even very restricted versions of the problem have been proved to be so~\cite{HVV94}. We present an integer linear programming (ILP) formulation suitable for only very small instances, and two heuristic algorithms (applicable to larger problem instances) to allocate computational resources (such as virtual machines) and networking resources (such as subcarriers) to each job. 

\indent In the second part of the dissertation, we study multiple resource allocation problems in multi-granular optical networks. We first consider the optimal waveband design problem in WDM mesh networks, without the restriction of uniform wavebanding. The band minimization problem is considered from the entire network's point of view, as opposed to a single node's. An Integer Linear Programming (ILP) formulation and effective heuristics are proposed to solve the problem. Our framework helps to achieve a great amount of hardware requirements. Then we compare two architectures for an OXC - a conventional architecture and the hierarchical architecture - with multiple fibers per link in WDM networks, and present heuristics for resource (i.e., fiber and wavelength) assignment and analytical models to compute their blocking performance. After that, we address the RFBSA problem in EONs with multi-fiber links that can accommodate non-contiguous and non-uniform wavebanding. We propose an auxiliary layered-graph framework with pluggable cost functions and develop cost functions to minimize the maximum spectral usage for a given set of traffic demands. To achieve good network performance, as well as saving considerable hardware costs, we propose heuristics to solve the joint wavebanding node placement and RFBSA problem, given the total number of available WSSs for the network as a budget.

\indent In the last part of the dissertation, we try to further improve the spectrum efficiency in multi-fiber EONs. We first propose an ILP model for the network planning based on topology and traffic pattern information. For each arriving request, one of the candidate paths for each source-destination node pair will be selected according to probabilities precomputed from the ILP. Then we utilize a dedicated partition scheme which provides a particular spectrum range for each request size. A next state aware spectrum assignment algorithm with resource sharing among partitions is proposed. Each scheme performs well in improving the spectrum efficiency and shows a huge performance enhancement jointly. 

\chapter{Co-scheduling Computational and Networking Resources in Elastic Optical Networks}
\label{chap_2}

\indent In this chapter, we address the problem of {\em co-scheduling} in elastic optical networks. We consider the multi-job multi-task {\em co-scheduling} problem to minimize the makespan of traffic demands in OOFDM-based elastic optical networks \cite{wu2014co}.

\section{Related Work}\label{sec.label21}

\indent To the best of our knowledge, this is the first study on this topic. Work exists in the literature on {\em co-scheduling} in WDM-based optical networks. The authors of \cite{DWLJ11} study the feasibility of virtualized optical network (VON) service for workflow applications. They propose a computation and communication delay-aware rescheduling ($C^2$DAR) scheme based on a new defined Scheduled Result Graph (SRG) concept to allocate computing resources and the shared VON resource. An availability-driven scheduling scheme is proposed in \cite{ZGXWJHG10}, which improves the directed acyclic graph (DAG) applications' availability iteratively by allocating two copies of one communication task to two disjoint lightpaths for data transfer while satisfying application deadline requirements. In \cite{AED10}, a Genetic Algorithm (GA) based approach is presented to co-schedule resources with the objective of minimizing application completion time. The authors of \cite{AZNS11} propose a service-oriented energy-efficient Internet architecture. It introduces an energy-aware analytical model and algorithms that cooperatively optimize the selection and scheduling of resources such that the overall power consumption by both the network and IT resources is minimized. In \cite{KGV13}, the {\em co-scheduling} problem in lambda-grids for advance reservation requests is studied, with the aim to minimize the job blocking probability. Our problem is different from these existing work because of the multiple subcarrier allocation feature in elastic optical networks. Specifically,  the multiple subcarriers (a subcarrier {\em band}) for a service are typically allocated in a contiguous manner, as blocks allocated to two different services must be separated by a guardband in order to avoid interference \cite{JKTWSTH10}.   

\section{Network Model and Problem Statement}\label{sec.label22}

\indent We model the system as follows. Informally speaking, there is a network of nodes (with collocated and computing (VMs) and switching resources) interconnected by optical links. A job consists of multiple inter-dependent tasks and each task is assigned a number of VMs on a node. After a task finishes processing, it communicates with its ``child'' tasks and transfers its results over the network's links. A task cannot start executing until all of its ``parent'' tasks have finished executing and their results have been transferred to the node assigned to the task. A job is considered to be finished if all of its tasks finish execution. The problem is to allocate resources (VMs and subcarriers) so that an objective is met. We formalize this with some notation below.

\indent The physical network topology is $G(V, E)$, where $V$ is the set of nodes, $E$ is the set of links, and $N=|V|$ and $L=|E|$. Each physical node $v \in V$ has $H_{v}$ virtual machines (VMs). Each link has two fibers with opposite directions, and each fiber has multiple subcarriers ($F$), with each subcarrier having bit rate $C$ Gbps (a fixed modulation scheme is assumed). A guardband consisting of multiple subcarriers ($G$) is placed between two adjacent subcarrier bands assigned to different connections. The shortest path $p_{s,d}$ (based on hop) for each pair of nodes $(s, d)$ is precomputed. We assume a time-slotted system, i.e., each computation or communication takes up an integer number of time slots (which is variable, depending on the number of allocated VMs and subcarriers).\footnote{We use a time-slotted system to facilitate the development of the ILP in Section III.}

\indent Each job $j$ is represented as a directed acyclic graph (DAG) $G^j (I^j, Q^j)$, where $I^j$ is the set of tasks of job $j$. If there is a link from task $i$ to task $r$, it means the following: task $i^j$ must finish and transfer its results to task $r^j$ before the task $r^j$ can start processing. Each task $i^j \in I^j$ has a workload requirement ($w^{j} _{i}$), which is represented as the needed processing time slots of the task when assigned {\em one} VM. There is a data transfer size $D_q^j$ associated with each link of $Q^j$; this represents the number of time slots needed to transfer the result over link $q$ (from the task at the head of the DAG link to the task on the tail of the link) when assigned {\em one} subcarrier. Fig.~\ref{fig:coscheduling_networkjob} shows a network example and a job example, respectively.

\begin{figure}
\includegraphics[scale=0.6]{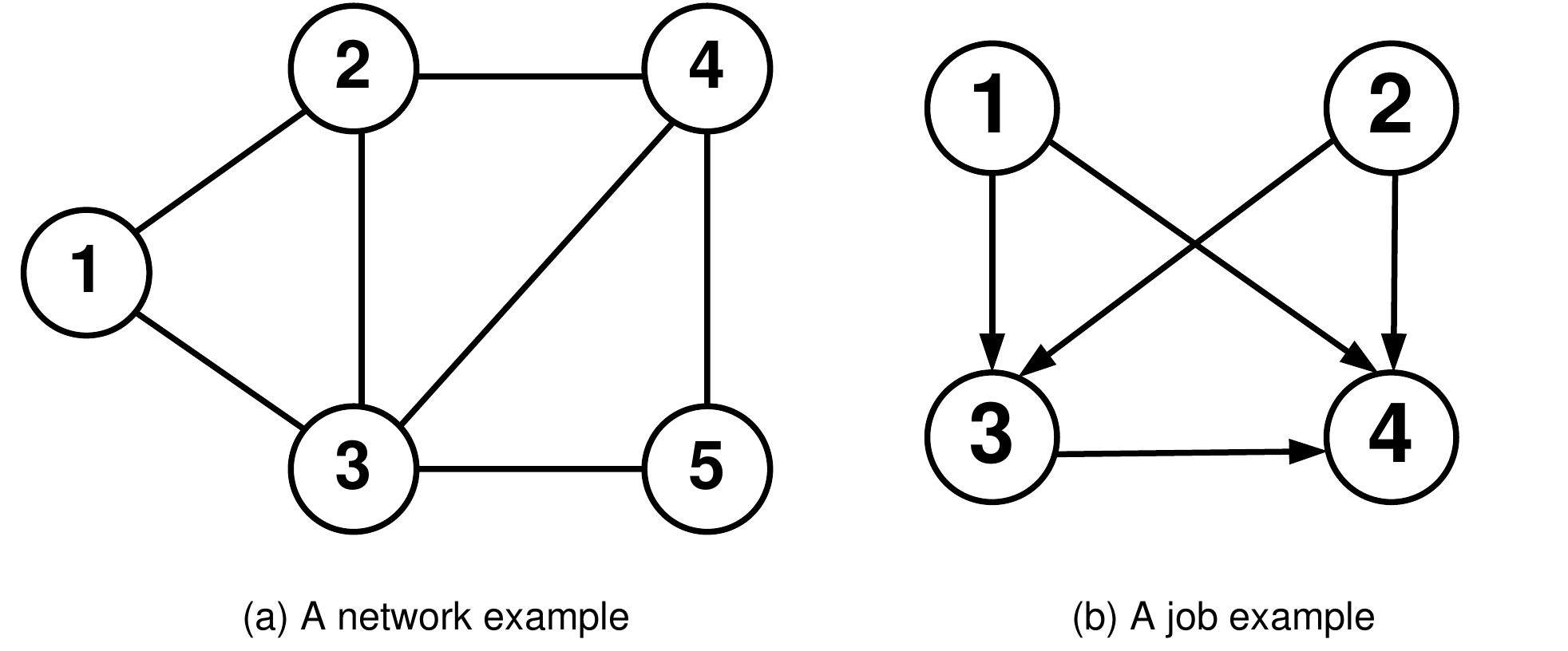}
\centering
\caption{\label{fig:coscheduling_networkjob} Examples of  a network graph and job DAG.}
\end{figure}

\indent We consider the problem of allocating computational resources (VMs) and communicational resources (subcarriers) for each task. Consider a job; for each of its tasks, a physical node $v$ must be assigned to it, and a number of VMs (say, $k$) on that physical node should be allocated to the task during its execution. The processing time is then calculated as $\left\lceil \frac{w^{j} _{i}}{k} \right\rceil$ slots (since $k$ VMs are allocated). Further, a physical path $p$ should be assigned to each link in $Q^j$,  and a band (of size $B$, say) of contiguous subcarriers on each physical link of the path should be allocated to transfer the results; the transfer time is then calculated as $\left\lceil \frac{D_q^j}{B} \right\rceil$. For two dependent tasks $i^j, r^j \in I^j$ (i.e., there is a link in $Q_j$ between these tasks in the DAG), if the two tasks are assigned to the same physical node, there is no need to allocate subcarriers to transfer the result of task $i^j$ to $r^j$, and the communication time is $0$. Note that after a task finishes processing at a physical node, the results can be stored temporarily at the node for some time before being transmitted to its successor tasks, i.e., data transfer does not have to immediately follow the completion of a task's execution. Similarly, after the results from parent tasks (of a given task) have arrived at the assigned physical node of the task, they can be stored temporarily for some time before the current task starts processing. In this work, we do not take into account the setup delay for each transmission, which could be viewed as a constant.

\indent For the static case, a set of jobs is given, and the objective is to minimize the makespan (completion time of the job that finishes last) \cite{xiang2013self,xiang2016joint,aggarwal2017sprout,xiang2015multi, xiang2017differentiated, xiang2017optimizing,xiang2015taming,xiang2013optimizing, xiang2014optimizing}. For the dynamic case, the jobs arrive in a random manner, and each job has a deadline requirement. If a job cannot finish before the deadline, it is blocked. The objective is to minimize the {\em job blocking}.

\section{Scheduling Algorithms}\label{sec.label23}

\indent In this section we introduce our algorithms for job {\em co-scheduling}.  We first present an ILP formulation for static requests that can be used to solve small problem instances. We then present two heuristics for the static case, and later adapt them to the dynamic case.

\subsection{ILP Formulation}

\indent In the ILP formulation, the input parameters are shown in Table~\ref{table:notations}.

\begin{table}
\centering
\caption{\label{table:notations} Notation for the {\em co-scheduling} problem}
\begin{tabular}{|>{\centering}m{1.8cm}| m{12.6cm}<{\centering}|}
\hline
\bf{Symbol}& \bf{Meaning} \\ \hline
$N$ &  number of nodes in the network \\ \hline
$L$ &  number of links in the network\\ \hline
$\Omega$ &  estimated upper bound of make span\\ \hline
$J$ &  number of jobs \\ \hline
$F$ &  number of subcarriers per fiber \\ \hline
$j$ &  an arbitrary job\\ \hline
$I^j$ &  the set of tasks in job $j$\\ \hline
$i^j$ &  an arbitrary task in job $j$, $\in I^j$\\ \hline
$v$ &  an arbitrary network node\\ \hline
$e$ &  an arbitrary network link\\ \hline
$p$ &  an arbitrary path ($p_{s, d}$ is the path from node $s$ to node $d$)\\ \hline
$\xi_{p_{s,d}}^{e}$ & = 1 if link $e$ is on path $p_{s, d}$; = 0, otherwise\\ \hline
$H_{v}$ &  the maximum number of VMs at node $v$\\ \hline
$H$ &  $\max_{v} H_v$\\ \hline
$w_i^j$ &  workload of task $i$ of job $j$\\ \hline
$S_i^j$ &  the start processing time slot of task $i$ of job $j$\\ \hline
$T_j^j$ &  the finish processing time slot of task $i$ of job $j$\\ \hline
$D_{i, r}^j$ &  the data transfer requirement from task $i$ to task $r$ of job $j$, = 0 if there is no data transfer from $i$ to $r$\\ \hline
$\bar{D}_{i, r}^j$ &  = 1, if the results of task $i$ need to be sent to task $r$ of job $j$\\ \hline
$B_{i, r}^j$ &  the number of subcarriers allocated to transfer results from task $i$ to task $r$ of job $j$\\ \hline
$E_{i, r}^j$ &  start time slot for transferring results from task $i$ to task $r$ of job $j$\\ \hline
$X_{i, r}^j$ &  finish time slot for transferring results from task $i$ to task $r$ of job $j$\\ \hline
\end{tabular}

\end{table}

\begin{center}Objective: Minimize $\max_{i, j} \sum_{x=1}^{\Omega}xg_{i, x}^j +\sum_{v=1}^N\sum_{k=1}^{H}(\left\lceil \frac{w_i^j}{k} \right\rceil-1)a^j_{i, v, k}$\end{center}

Variables:

a)
\begin{equation*}
a^{j}_{i, v, k}=
\begin{cases}
                1, \quad\makebox{if task } i \makebox{ of job } j \makebox{ is assigned to node } v \makebox{ and allocated }  k  \makebox{ VMs;} \\
                0, \quad\makebox{otherwise }
\end{cases}
 \end{equation*}

b)
\begin{equation*}
g^j_{i, x}=
\begin{cases}
                1, \quad\makebox{if task } i \makebox{ of job } j \makebox{ start processing time slot is }x;\\
                0, \quad\makebox{otherwise }
\end{cases}
 \end{equation*}

c)
\begin{equation*}
A^{j} _{i, v, t, k}=
\begin{cases}
                1, \quad\makebox{if task } i \makebox{ of job } j \makebox{ use } k \makebox{ VMs on node } v \makebox { at time slot } t;\\
                0, \quad\makebox{otherwise }
\end{cases}
 \end{equation*}

d)
\begin{equation*}
\phi^{j} _{i, r, y, m}=
\begin{cases}
                1, \quad\makebox{if the first subcarrier index allocated for communication } \\ \quad\quad\makebox{from task } i \makebox{ to task } r \makebox{ of job } j \makebox{ is }y \makebox{ and the number of }\\ \quad\quad  \makebox{allocated subcarriers is }m; \\
                0, \quad\makebox{otherwise }
\end{cases}
 \end{equation*}

e)
\begin{equation*}
u^{j} _{i, r, l}=
\begin{cases}
                1, \quad\makebox{if the start transmission time slot for communication } \\\quad\quad \makebox{from task } i \makebox{ to task } r \makebox{ of job } j \makebox{ is }l;\\
                0, \quad\makebox{otherwise }
\end{cases}
 \end{equation*}

f)
\begin{equation*}
o^{j} _{i, r, c, t, e}=
\begin{cases}
                1, \quad\makebox{if the result transmission from task } i \makebox{ to task } r\makebox{ of job } j\\\quad\quad  \makebox{use subcarrier }c \makebox{ on link } e \makebox{at time slot } t; \\
                0,\quad\makebox{otherwise }
\end{cases}
 \end{equation*}

g)
\begin{equation*}
z^{j} _{i, r, v}=
\begin{cases}
                1, \quad\makebox{if the task } i \makebox{ and its child task } r\makebox{ of job } j\\\quad\quad\makebox{are both allocated to node }v; \\
                0,\quad\makebox{otherwise }
\end{cases}
 \end{equation*}

Constraints:

a) Each task of each job can be assigned to only one node and with fixed VM allocation. For all $j, i$
\begin{equation*}
\sum_{v=1}^{N} \sum_{k=1}^{H_v} a^{j} _{i, v, k} = 1
\end{equation*}

b) Constraint to find the value of $z^{j}_{i, r, v}$. For all $j, i, r, v$,

\begin{equation*}
\sum_{k=1}^{H_v} a^{j} _{i, v, k}+ \sum_{k=1}^{H_v} a^{j} _{r, v, k} \leq 1+z^{j} _{i, r, v}
\end{equation*}
\begin{equation*}
\sum_{k=1}^{H_v} a^{j} _{i, v, k}+ \sum_{k=1}^{H_v} a^{j} _{r, v, k} \geq 2z^{j} _{i, r, v}
\end{equation*}

c) Each task of each job can start processing at only one time slot. For all $i, j$,
\begin{equation*}
\sum_{x=1}^{\Omega}g_{i, x}^j=1
\end{equation*}

d) The constraint to get value of $A^{j} _{i, v, t, k}$ based on the value of $g^{j} _{i, x}$ and $a^{j}_{i, v, k}$. For all $j, i, v, k, t, t-\lceil \frac{w_i^j}{k}\rceil+1 \leq x \leq t$,
\begin{equation*}
g^{j} _{i, x}+ a^{j}_{i, v, k} \leq 1+ A^{j} _{i, v, t, k}
\end{equation*}
\begin{equation*}
g^{j} _{i, x}+ a^{j}_{i, v, k} \geq 2A^{j} _{i, v, t, k}
\end{equation*}

e) Each node's VM capacity cannot be exceeded for any time slot. For all $v, t$,
\begin{equation*}
\sum_{j=1}^J\sum_{i\in I^j} \sum_{k=1}^{H} kA^{j} _{i, v, t, k}\leq H_v
\end{equation*}

f) Each transmission from task $i$ to $r$ can start at only one time slot. For all $i, r, j$,
\begin{equation*}
\sum_{l=1}^{\Omega}u_{i, r, l}^j=1
\end{equation*}

g) Each transmission has only one start subcarrier index and has a certain number of allocated subcarriers. For all $i, r, j$,
\begin{equation*}
\sum_{y=1}^{F}\sum_{m=0}^{F} \phi_{i, r, y, m}^j=1
\end{equation*}

h) For each task, its communication start time slot cannot be earlier than its finish processing time slot. For all $i, r, j$,
\begin{equation*}
\sum_{x=1}^{\Omega}xg_{i, x}^j +\sum_{v=1}^N\sum_{k=1}^{H}(\lceil \frac{w_i^j}{k} \rceil-1)a^j_{i, v, k}<\sum_{l=1}^{\Omega} lu_{i, r, l}^j
\end{equation*}

i) A task should start after the data transfer from its preceding/parent tasks is completed. For all $ i, r, j$,
\begin{equation*}
\begin{multlined}
\sum_{l=1}^{\Omega}lu_{i, r, l}^j+\sum_{m=1}^{F}\sum_{y=1}^{F} (\lceil \frac{D^{j} _{ir}}{m} \rceil-1)\phi^{j} _{i, r, y, m}-\sum_{x=1}^{\Omega}xg^{j}_{r, x} < \Lambda(1-\bar{D}^{j} _{ir}+\sum_{v=1}^N z_{i, r, v}^j)
\end{multlined}
\end{equation*}

\noindent
where $\Lambda$ is a very large number ($> \max_{i, r, j} D^{j} _{ir}$).

j) Constraint to get the value of  $o_{i, r, c, t, e}^j$. For all $j, i, r, e, c, t, c-m+1 \leq y \leq c, t-\lceil \frac{D^{j} _{ir}}{m} \rceil+1 \leq l \leq t$,

\begin{equation*}
\begin{multlined}
\xi_{p_{n_1, n_2}}^{e}+ u^{j}_{i, r, l}+ \phi^{j}_{i, r, y, m}+\sum_{k=1}^{H}a_{i, n_1, k}^j+\sum_{k=1}^{H}a_{r, n_2, k}^j \leq 4+ o_{i, r, c, t, e}^j
\end{multlined}
\end{equation*}
\begin{equation*}
\begin{multlined}
\xi_{p_{n_1, n_2}}^{e}+ u^{j}_{i, r, l}+ \phi^{j}_{i, r, y, m}+\sum_{k=1}^{H}a_{i, n_1, k}^j+\sum_{k=1}^{H}a_{r, n_2, k}^j \geq 5 o_{i, r, c, t, e}^j
\end{multlined}
\end{equation*}

k) Each link's subcarrier $c$ can be allocated to at most one data transfer. For all $e, t, c$,
\begin{equation*}
\sum_{j=1}^J \sum_{i\in I^{j}}\sum_{r \in I^{j}} o_{i, r, c, t, e}^j \leq 1.
\end{equation*}

\indent Current formulation needs the collaboration between networks and data centers. Exact information must be interchanged between networks and data centers. This could be justified for some large cloud service companies while more constraints should be imposed for small distributed data centers connected by different network operators. To simplify the problem, we just assume that there is a global scheduler which collects information from both the network and data center side.

\subsection {Static Case}

\indent Two heuristics based on List Scheduling are proposed. Both of them allocate computational and communication resources to jobs one by one. The first algorithm (First Fit) schedules the jobs and tasks in each job by the order of their index. The second algorithm (Children Aware) schedules them by decreasing order of their weight (the method to find a job's weight and a task's weight will be shown later). Two global states are maintained: one is $h_v(t)$, the number of free VMs on physical node $v$ at time slot $t$; and the other is $s_e^c(t)$, which equals 1 if subcarrier $c$ on physical link $e$ is available at time slot $t$ and is 0 otherwise.

\indent Before describing the algorithm details, we introduce the notion of {\em layer}. Each job's tasks are categorized into several layers -- the first layer consists of tasks that have no preceding tasks. The tasks of layer two are those not in layer one and all of their preceding tasks are in layer one. In general, the tasks of layer $\eta$ are the ones not in layers $\eta^{\prime} < \eta$, and at least one of their preceding tasks is in layer $\eta-1$.

\subsubsection{First Fit (FF)}

\indent In the First Fit algorithm, the jobs are numbered (in the order they are generated) and considered from smallest index to largest index. For each layer of a job, from the lowest layer to the highest layer, the tasks are also considered by increasing order of index. Let $I_{\eta}$ denote the set of tasks in Layer $\eta$.

\indent {\em Layer One:} For each start processing time slot $S_{i}^j$ of task $i$, the network nodes with $h(S_{i}^j)>0$ are considered one by one. For each node $n$, the number of potentially allocatable VMs $k_i^j$ are checked one by one, from $\min(h_{n}(S_i), w_i^j)$ to $1$. After this scheduling, the corresponding $h_v(t)$ states are updated, as shown in Algorithm 1.

\indent {\em Higher Layers:} For each task, the algorithm considers the task's start processing time slot $S_{i}^j$, and the nodes with $h(S_{i}^j)>0$ are considered one by one. Let $U_i^j$ be the set of preceding tasks of task $i^j$, then $S_{i}^j \geq \max_{i' \in U_i^j}T_{i'}^j$;  they are equal when task $i$ is assigned to the same node as that preceding task (and there is no waiting time before task $i$'s processing). For each possible $S_{i}^j$ and $n_i^j$, we check the condition that all of its preceding tasks (note that all tasks in $U_i^j$ have been assigned already) can finish their data transfer to task $i$ no later than $S_{i}^j$. For $q_{i', i}^j$, the number of allocated subcarriers is selected as the one that achieves the minimum transfer finishing time $X_{i', i}^j$. The allocation method we use is First-Fit, i.e., the first available contiguous subcarrier band on the corresponding path is allocated. If the condition is satisfied, then consider the number of allocated VMs $k_i^j$. After this scheduling, the corresponding $h_v(t)$ and $s_e^c(t)$ states are updated, as shown in Algorithm 2.

\begin{algorithm}
\caption{First Fit: Allocate Layer One}
\begin{algorithmic}[1]
\FOR {Task $i=1,2, \ldots, |I_1|$}
\STATE $T'_i=\infty$
\FOR {Start processing time $S_i=1,2, \ldots, \Omega$}
\IF {$T'_i<\infty$}
\STATE Break
\ENDIF
\STATE Keep the physical nodes with $h_v(S_i)>0$, get set $V'$
\FOR {Node $n_i=1,2, \ldots, |V'|$}
\IF {$T'_i<\infty$}
\STATE Break
\ENDIF
\FOR {VM $k_i=\min(h_{n}(S_i), w_i^j), \ldots, 1$}
\STATE Calculate $T_i=S_i+\lceil \frac{w_i}{k_i} \rceil-1$
\IF {$T_i<T'_i$, and $h_n(t) \geq k_i$ for $S_i \leq t \leq T_i$}
\STATE $T'_i=T_i$
\STATE Break
\ENDIF
\ENDFOR
\ENDFOR
\ENDFOR
\STATE Assign task $i$ to the $n_i$, $k_i$, $S_i$ that achieves $T'_i$ and update $h_n(t)=h_n(t)-k_i$ for $S_i \leq t \leq T'_i$
\ENDFOR
\end{algorithmic}
\end{algorithm}

\begin{algorithm}
\caption{First Fit: Allocate Layer $\eta$}
\begin{algorithmic}[1]
\FOR {Task $i=1,2, \ldots, |I_{\eta}|$}
\vspace{-4pt}
\STATE $T'_i=\infty$
\vspace{-4pt}
\FOR {Start processing time $S_i=\max_{i' \in U_i}T_{i'}, \ldots, \Omega$}
\vspace{-4pt}
\IF {$T'_i<\infty$}
\vspace{-5pt}
\STATE Break
\vspace{-5pt}
\ENDIF
\vspace{-4pt}
\STATE Keep the physical nodes with $h_v(S_i)>0$, get set $V'$
\vspace{-4pt}
\FOR {Node $n_i=1,2, \ldots, |V'|$}
\vspace{-4pt}
\IF {$T'_i<\infty$}
\vspace{-5pt}
\STATE Break
\vspace{-5pt}
\ENDIF
\vspace{-4pt}
\STATE $U'_i=U_i$
\vspace{-4pt}
\FOR {$i'=1,2, \ldots, |U'_i|$}
\vspace{-4pt}
\STATE Minimum finish transferring time $\bar{X}_{i', i}= \infty$
\vspace{-4pt}
\FOR {Start transferring time $E_{i', i}=T_{i'}, \ldots, \Omega$}
\vspace{-4pt}
\FOR {Number of subcarriers $B_{i', i}=1, 2, \ldots, F$}
\vspace{-4pt}
\STATE Finish transferring time $X_{i', i}=E_{i', i}+ \lceil \frac{D_{i', i}}{B_{i', i}}\rceil-1$
\vspace{-4pt}
\IF {$X_{i', i}<\bar{X}_{i', i}$, $X_{i', i} \leq S_i$, and there is contiguous available subcarrier band with size $\geq B_{i', i}$ on path from $n_{i'}$ to $n_i$}
\vspace{-4pt}
\STATE $\bar{X}_{i', i} =X_{i', i}$
\vspace{-4pt}
\ENDIF
\vspace{-4pt}
\ENDFOR
\vspace{-4pt}
\ENDFOR
\vspace{-4pt}
\ENDFOR
\vspace{-4pt}
\IF {For all $i'$, $\bar{X}_{i', i} < \infty$}
\vspace{-4pt}
\FOR {VM $k_i=\min(h_{n}(S_i), w_i^j), \ldots, 1$}
\vspace{-4pt}
\STATE Calculate $T_i=S_i+\lceil \frac{w_i}{k_i} \rceil-1$
\vspace{-4pt}
\IF {$T_i<T'_i$, and $h_n(t) \geq k_i$ for $S_i \leq t \leq T_i$}
\vspace{-4pt}
\STATE $T'_i=T_i$, Break
\vspace{-4pt}
\ENDIF
\vspace{-5pt}
\ENDFOR
\vspace{-5pt}
\ENDIF
\vspace{-5pt}
\ENDFOR
\vspace{-5pt}
\ENDFOR
\vspace{-4pt}
\STATE Assign task $i$ to the $n_i$, $k_i$, $S_i$ that achieves $T'_i$, update $h_n(t)=h_n(t)-k_i$ for $S_i \leq t \leq T'_i$
\vspace{-4pt}
\STATE Allocate First-Fit subcarriers that achieves $\bar{X}_{i', i}$ and update corresponding $s_e^c(t)$
\ENDFOR
\end{algorithmic}
\end{algorithm}
\setlength{\textfloatsep}{1pt}

\subsubsection{Children-Aware (CA)}

\indent In the Children-Aware algorithm, each task and each job is assigned a weight. To find the weight of each task $i \in I^j$, a directed subgraph of the job $j$ is considered; the top node of the subgraph is task $i$, and all tasks of higher layers which are children of $i$ are included the subgraph. For example, the subgraph with task $1$ of the job in Fig.~\ref{fig:coscheduling_networkjob} (b) is shown in Fig.~\ref{fig:subgraph}. Let $R_{i}$ denote the set of tasks of the subgraph which are directly connected to task $i$ (in the example, task $3$ is the only task of $R_{1}$). The weight of task $i$ is defined as, \begin{equation}
\theta_i=\max_{r \in R_i} (\alpha w_i+\beta D_{i, r}+\theta_r).
\label{eq:theta}
\end{equation}

\noindent
where $\alpha$ and $\beta$ are suitable constants, and the other variables are as defined in Table I. Parameters $\alpha$ and $\beta$ are chosen according to the ratio of computation resources to communication resources. The weight of job $j$ is set as $\theta^j=\max_{i \in I_0} \theta_i$ (recall that $I_0$ is set of tasks of layer $0$).

\begin{figure}
\includegraphics[scale=0.45]{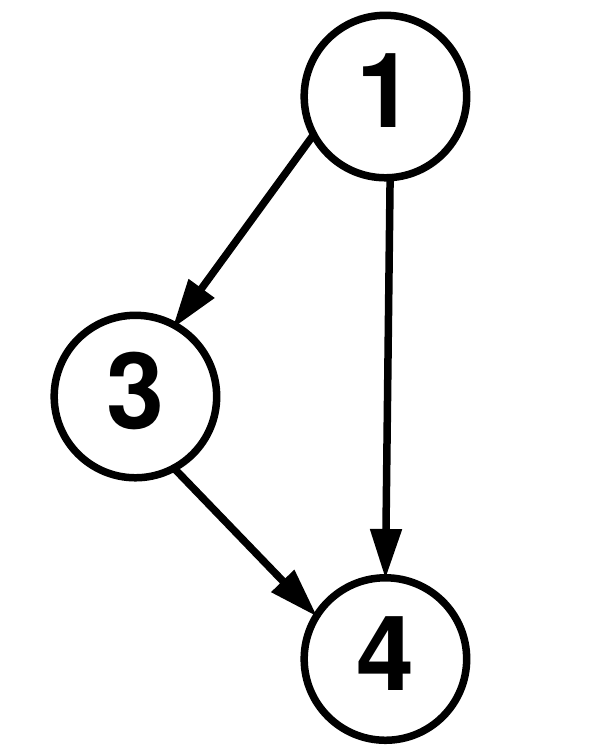}
\centering
\caption{\label{fig:subgraph} Example of subgraph.}
\end{figure}

\indent The first step of the CA algorithm is to sort the jobs by decreasing order of their weights $\theta^j$. Then the jobs are considered one by one by this order. For each layer of a job, the tasks of the current layer are sorted by decreasing order of weight $\theta_i$, and the allocation is done in this order.

\indent Two changes are made to Algorithm 1 and Algorithm 2: (a) The Break part in lines $4$-$6$ and $9$-$11$ (both Algorithm 1 and Algorithm 2) are removed, and (b) In line 12 of Algorithm 2, $U_i$ is first sorted in decreasing order of data transfer size to task $i$; and this ordered set is $U'_i$.

\indent The worst case time complexity for both algorithms are the same: $J(IwNH_v+I^2D^2NF)$. Here, $I$ denotes the maximum number of tasks in one job, $w$ denotes the maximum task workload, and $D$ denotes the maximum data transfer requirement.

\subsection{Dynamic Requests}

\indent For the dynamic case, we assume the following model. At each time slot there is a probability $p_{new}$ that a new job will arrive during this time slot. Thus the tasks of the job can start processing from the next time slot. Besides the workload and transfer size requirements, each job has a deadline requirement. If the algorithm determines that some of the tasks of the job cannot finish before the deadline, the job is blocked and no resources are allocated to it. We adapt the previous two heuristics to the dynamic case.

\indent At each time slot, the algorithms check if a new job arrives during this time slot; if so, schedule it; otherwise, go on to consider the next time slot. Note that for a job that arrives at time slot $\tau$, its tasks may start execution at a later time slot $\tau'$, as long as all of its tasks can finish before the deadline. 

\section{Simulation Results}\label{sec.label24}

\indent We present results for two network topologies, a small 5-node network shown in Fig.~\ref{fig:coscheduling_networkjob} (a) and the larger NSFNET network, shown in Fig.~\ref{fig:NSF}. Each subcarrier being $12.5$ GHz (according to standardization work related to the
flexible grid specification is on-going in ITU-T~\cite{ref125}), the data rate for QPSK is 12.5 Gbps per subcarrier. In addition, the guardband is assumed to be $G=2$ subcarriers. For each pair of network nodes, the shortest path (in terms of hops) is used.

\begin{figure}
\includegraphics[scale=0.4]{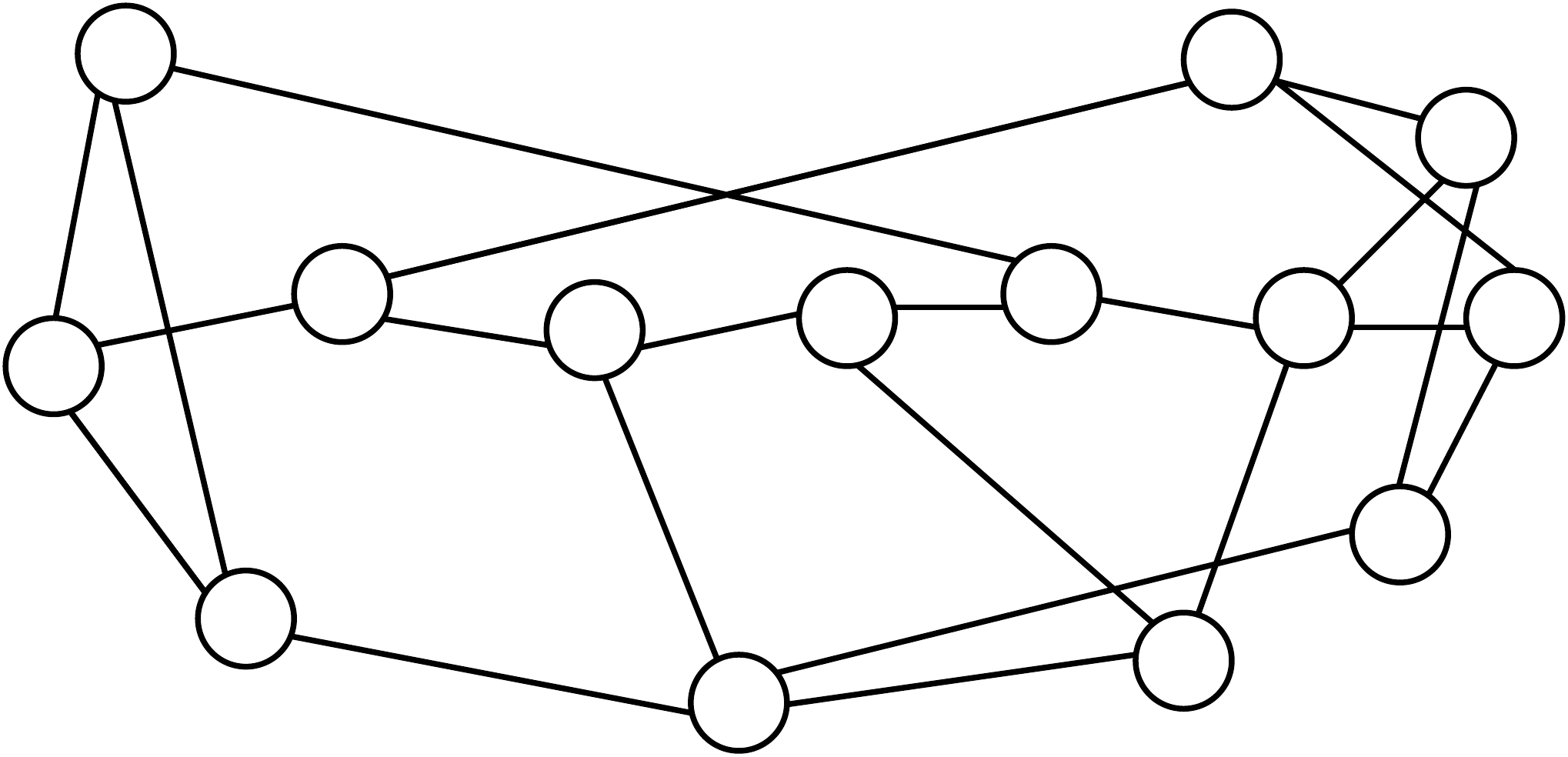}
\centering
\caption{\label{fig:NSF} NSF network.}
\end{figure}


\subsection {Static Requests}

\indent We first present results for the static case. Due to the high complexity of the ILP, we are only able to obtain results for the small network. There are $5$-$10$ VMs per network node. Each fiber has $12$ subcarriers. Job requests are generated as follows. Each job has between $2$ and $4$ tasks (uniformly random), and a result transfer requirement exists from task $i$ to another task $r$, $r>i$ (of the same job) with probability $0.8$ (we discard any task graph that is not connected). Each task's workload is a random number between $5$ and $15$. Each transfer size is uniformly distributed over the range $15$-$20$. In our case, since subcarrier resources are more abundant than VM resources, we choose $\alpha=0.5$ and $\beta=1$ for the Children-Aware algorithm. Results for $2$-$7$ job requests are shown in Table~\ref{coscheduling_ilp}.

\begin{table}
\caption{\label{coscheduling_ilp} Performance of ILP and algorithms for 5-node network: Static requests}
\centering
\begin{tabular}{|c|c|c|c|}
\hline
\bf{No. of Jobs}& \bf{First Fit (FF)}  &\bf{Children Aware (CA)} & \bf{ILP}\\ \hline
\bf{2} & 5  & 4  & 4\\ \hline
\bf{3}   & 7 & 7 & 6 \\ \hline
\bf{4} & 9   & 8 & 7  \\ \hline
\bf{5}   & 11 & 10  & 7 \\ \hline
\bf{6}  & 11  & 8  & 7 \\ \hline
\bf{7}   & 11  & 10 & 8 \\ \hline
\end{tabular}
\vspace{10pt}
\end{table}

\indent From these results, we observe that, in general, the CA algorithm performs better than FF. In addition, their results are quite close to those achieved by the ILP. For the NSFNET network, there are $50$-$100$ VMs per network node \cite{SIGCOMM08}, and each fiber has 320 subcarriers. Job requests are generated as follows. Each job has between $5$ and $10$ tasks, and a result transfer requirement exists from task $i$ to another task $r$ with probability $0.5$. Each task's workload is a random number between $50$ and $150$. Each transfer size is uniformly distributed over the range $250$-$350$. Results for the two heuristics are presented in Table~\ref{coscheduling_largestatic}. These results confirm our earlier observations. The execution times (ET) in minutes for the two heuristics are also given in the same table. We see that CA performs better for large workloads with reasonable time complexity.

\begin{table}
\caption{\label{coscheduling_largestatic} Performance of algorithms for NSFNET: Static requests}
\centering
\begin{tabular}{|c|c|c|c|c|}
\hline
\bf{No. of Jobs}& \bf{FF}  &\bf{ET (FF)}  &\bf{CA} &\bf{ET (CA)} \\ \hline
\bf{50} & 159  & 8 & 42 & 20 \\ \hline
\bf{100}   & 159 & 18 & 82 & 53 \\ \hline
\bf{150} & 204 & 34  & 123 & 93 \\ \hline
\bf{200}   & 204 &50 & 162 & 209 \\ \hline
\bf{250}  & 268 & 63  & 202 & 400 \\ \hline
\bf{300}   & 324 & 77  & 244 & 427 \\ \hline
\end{tabular}
\vspace{10pt}
\end{table}

\subsection{Dynamic Requests}

\indent For the dynamic case, at each time slot, there is a new job arrival with probability $0.5$. Each job's deadline $\Gamma$ is set as $\Gamma=\max_{i \in I^j}\lceil \theta_i \rceil$, where  $\theta_i$ is calculated by Equation~(\ref{eq:theta}), but with $\alpha=0.01$, $\beta=0.01$ (different from the values used to calculate the task's weight), which estimates the job makespan by assuming 100 VMs to each task and 100 subcarriers to transfer results between tasks. The parameters of the network and jobs are the same as in the static case. Results for the heuristics are shown in Fig.~\ref{fig:coscheduling_dynamic1}. The execution times for the heuristics are shown in Table \ref{coscheduling_largedynamic}. Once again, the superior performance of CA can be seen.

\begin{figure}
\includegraphics[scale=0.55]{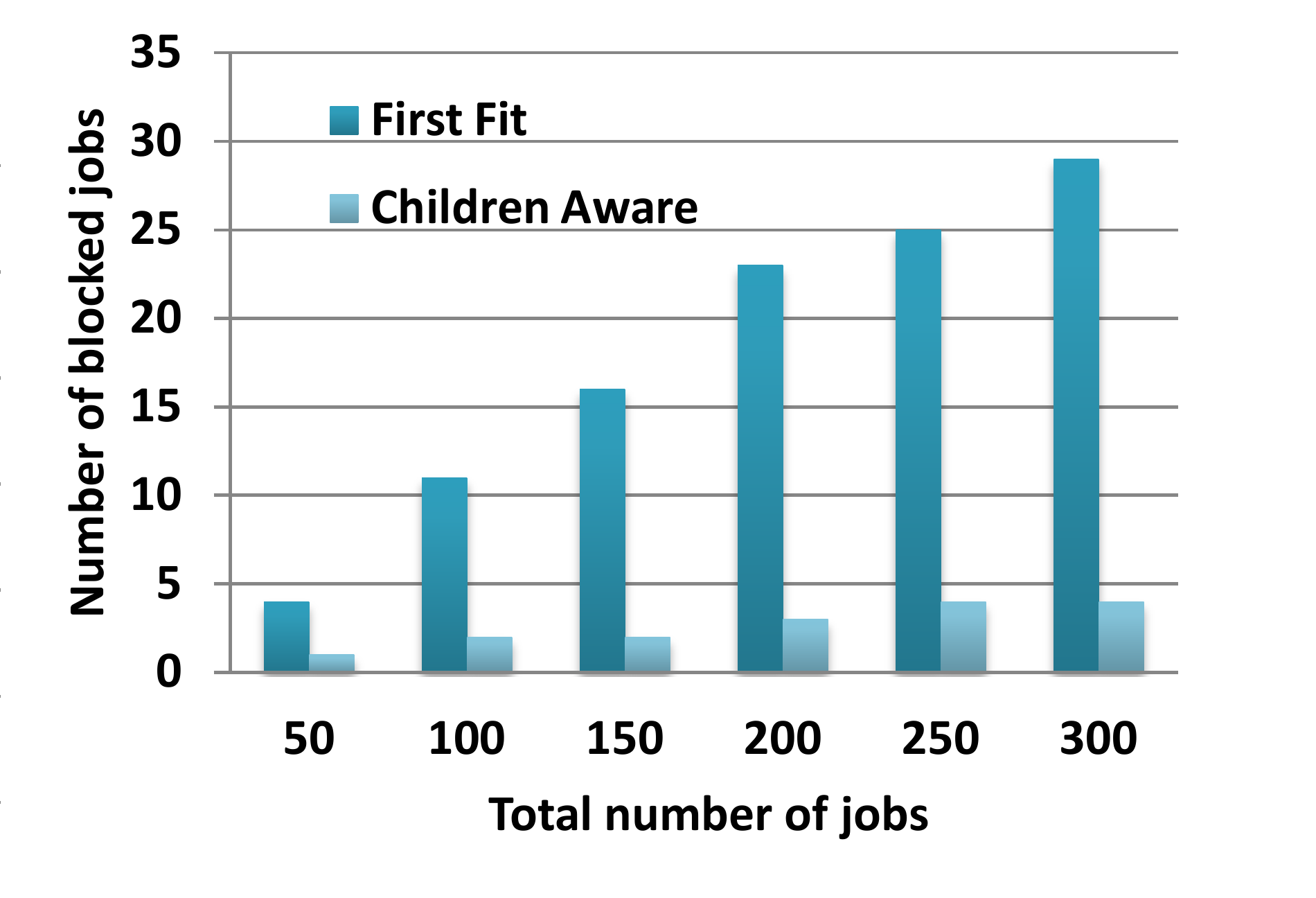}
\centering
\caption{\label{fig:coscheduling_dynamic1} Number of blocked jobs for dynamic traffic.}
\end{figure}

\begin{table}
\caption{\label{coscheduling_largedynamic} Execution time in minutes for NSFNET: Dynamic requests}
\centering
\begin{tabular}{|c|c|c|}
\hline
\bf{No. of Jobs} &\bf{Execution time (FF)}  &\bf{Execution time (CA)} \\ \hline
\bf{50}  &7  & 25 \\ \hline
\bf{100}   & 14  & 56 \\ \hline
\bf{150}  & 23   & 78 \\ \hline
\bf{200}   &32  & 115 \\ \hline
\bf{250}   & 34   & 142 \\ \hline
\bf{300}    & 49   & 210 \\ \hline
\end{tabular}
\vspace{10pt}
\end{table}

\section{Conclusions}\label{sec.label25}
\indent In this chapter, we investigated the problem of {\em co-scheduling} computational and networking resources for both static and dynamic multi-task jobs in elastic optical networks. We proposed two heuristics, and presented representative simulation results. According to the results, for the static case, in general, the Children-Aware algorithm has better makespan performance than the First Fit algorithm. For the dynamic case, the adapted Children-Aware algorithm has good blocking performance. Since we give weights to tasks according to the task dependencies, tasks which have more computation requirement and initiate more transmission requirements are expected to schedule first.

\chapter{Optimal Nonuniform Wavebanding in WDM Mesh Networks}
\label{chap_3}

\indent This chapter investigates the multi-granularity in WDM networks. The scheduling of wavelengths into wavebands without uniform wavebanding constraint to save switching costs in WDM mesh networks is studied. 

\section{Related Work}\label{sec.label31}

\indent There are two types of waveband switching -- uniform waveband switching and nonuniform waveband switching \cite{izmailov2003nonuniform}. For uniform waveband switching, the sizes of wavebands at a node are all the same and this common waveband size is used for all network nodes. Optical port-saving under uniform waveband switching is studied in \cite{torab2006waveband}. A uniform waveband design problem in mesh networks is studied in \cite{alnaimi2010uniform}. In nonuniform waveband switching, the sizes of bands could be different for each node of the network. Utilizing nonuniform waveband switching is more advantageous in saving switching elements for WDM networks \cite{chen2006uniform}. The authors in \cite{chen2006uniform} first consider a star topology with a single source node and then extend their results to ring topologies.

\indent Previous work \cite{izmailov2003nonuniform} which investigated nonuniform wavebanding focusing on a single switching node has been demonstrated to be not suitable for an entire network \cite{turkcu2014optimal}. The authors in \cite{turkcu2014optimal} focus on optimal waveband design in ring networks, and propose a novel framework for band minimization. In this chapter, we study a more complicated case, which is the optimal waveband design problem in mesh networks, without the restriction of uniform wavebanding.

\section{Network Model and Definitions}\label{sec.label32}
 
\indent Consider a mesh network with $N$ nodes. The nodes are numbered as $1,2,\cdots, N$. Given a traffic as a set of lightpaths, a Routing and Wavelength Assignment (RWA) algorithm is used to provision the lightpaths. Suppose $W$ wavelengths are assigned to the given traffic after routing and wavelength assignment. The node degree (excluding local add/drop) for node $n$ is $d_n$. We can construct a ``code'' for each node $n$ and concatenate the codes for all nodes to generate a code for each wavelength. 

\indent The principles for a code are as follows. There are $\binom{d_n}{2}+d_n =\frac{d_n(d_n+1)}{2} $ bits in the code for node $n$. We assume the connection requests are bidirectional, and the code permutations for I-O switching are not considered here. For every input-output pair in a node, $(I_i,O_j), i,j=1,2,...,d_n$, we have one bit that is 1 if the wavelength needs to be switched between $I_i$ and $O_j$, and 0 otherwise. There are $\binom{d_n}{2}$ such pairs for each node. $d_n$ more bits are appended to the code -- the $i$th bit in the $d_n$ bits is 1 if the wavelength for port $i$ is added or dropped. Thus, a code of length $\frac{d_n(d_n+1)}{2}$ fully encodes the behavior of a single wavelength at node $n$. Concatenating such codes for a wavelength at nodes 1 through $N$ makes up the row for that wavelength in the RWA matrix. Let $R=\sum_{n=1}^{N}\frac{d_n(d_n+1)}{2}$. Thus the RWA matrix has $W$ rows and $R$ columns. Let $r_w$ denote the row vector of the RWA matrix corresponding to wavelength $w$, and let $c_v$ denote the column vector corresponding to bit $v$ in the concatenated code. 

\indent We observe that the runs of 1's in a specific column, say $c_n$, represent opportunities for wavebanding at a specific input-output pair or at an add/drop port at a node. If the location of $c_n$ corresponds to an input-output pair $(I_i,O_j)$ at a node, a string of consecutive 1's in the column means that all of those consecutive wavelengths are bypassed between port $i$ and port $j$, and hence can be switched as a single band. Similarly, if the location of  $c_n$ corresponds to the $i$th add/drop port at a node, a string of consecutive 1's in this column means that all of those consecutive wavelengths are added/dropped at port $i$. Thus, the number of bands at a node can be obtained by counting the strings of consecutive 1's in each column which corresponds to a bit of the node's code and adding them together. The total number of bands in the network can be obtained by counting the strings of consecutive 1's in each column and summing them up through all columns. 

\indent For each column of the RWA matrix, we observe that when the transition between the first row and the second row is (0$\rightarrow$1), (1$\rightarrow$0) or (1$\rightarrow$1), the number of bands is one; in the case of (0$\rightarrow$0), the number of bands is zero. From the second row on, if the transition between two adjacent rows is (0$\rightarrow$1), the number of bands increases by 1 which means that a new band begins; otherwise, the number of bands remains unchanged. It is obvious that in the case of (0$\rightarrow$0), the number of bands remains unchanged. For (1$\rightarrow$0) and (1$\rightarrow$1) cases, the 1's belong to the previous bands which have already been counted, thus the number of bands remains unchanged. In this case, we can define the distance between two rows of one column as below. For each column $c$, we have
\begin{displaymath}    d_{1,j} ^c =
 \begin{cases}
    1   &  \text{if row 1 to row $j$ is (0$\rightarrow$1), (1$\rightarrow$0) or (1$\rightarrow$1);} \\
    0	 &  \text{if row 1 to row $j$ is (0$\rightarrow$0)},
 \end{cases}                
\end{displaymath}

\noindent
where $j=2,\cdots,W$. 

\begin{displaymath}    d_{i,j} ^c =
 \begin{cases}
    1   &  \text{if row $i$ to row $j$ is (0$\rightarrow$1); }\\
    0	 &  \text{if row $i$ to row $j$ is (0$\rightarrow$0), (1$\rightarrow$0) or (1$\rightarrow$1)},
 \end{cases}                
\end{displaymath}

\noindent
where $i=2,\cdots,W; j=2,\cdots,W; i\not=j$.
Then the number of bands for a column is $\beta_c = \sum_{i=1}^{W-1}d_{i,(i+1)}^c$. For example, the number of bands for the column $\left(
\begin{array}{ccccc}
1 & 0 & 0 & 1 &1
\end{array}
\right)$$^{\rm T}$ is 2. 

\indent The total number of bands in the RWA matrix is equal to the sum of the number of bands over all columns. We use $B$ to denote the total number of bands in the RWA matrix. 
\begin{equation}
B=\sum_{c=1}^{R}\beta_c=\sum_{c=1}^{R}\sum_{i=1}^{W-1}d_{i,(i+1)}^c=\sum_{i=1}^{W-1}\sum_{c=1}^{R} d_{i,(i+1)}^c.
\end{equation}
Since $d_{i,j}=\sum_{c=1}^{R} d_{i,j}^c$ is the distance between two rows $i$ and $j$, we have $B=\sum_{i=1}^{W-1} d_{i,(i+1)}$, which means that the total number of bands is equal to the summation of the distances between adjacent rows (where distances are defined as above). The Band Minimization Problem (BMP) is to minimize this number for a given RWA by re-ordering the rows of the RWA matrix (equivalent to re-assigning the wavelengths to lightpaths).

\indent We can then transform this problem into a Minimum Weight Hamiltonian Path Problem (\cite{berge1973graphs, libura1991sensitivity}) with a constraint that the first row is chosen. We can construct a graph $G = (V,E)$ consisting of $W$ nodes, each corresponding to one row/wavelength of the RWA matrix. Every node of this graph is connected to every other node with an edge weighted as the distance between the corresponding rows. Hence, the edge $(v_i, v_j)$ has the weight $d_{i,j}$ for every $i$ and $j$. $d_{i,j}$ could be different from $d_{j,i}$ according to the definition of distance above. The graph in this case is a directed graph with edge weights being integers. Then, the Hamiltonian path that visits every node in this graph with the minimum total weight gives the minimum total number of bands of the RWA matrix. Due to the definition of distance clarified above, different first rows/nodes give different link weights in the graph. The Hamiltonian path should start from the node corresponding to the chosen first row for that graph. Every row in the RWA matrix could be chosen as the first row, so we need to construct $W$ different graphs corresponding to $W$ different first rows and get solutions for each graph. The minimum solution among the $W$ solutions is the desired one. 

\section{Band Minimization ILP And Algorithms}\label{sec.label33}

\indent In this section, we first present an ILP formulation that can be used to solve small problem instances, and then present several heuristics that can be used to solve larger problem instances.

\indent Given a set of traffic requests, we first use Dijkstra's algorithm to assign shortest paths and the First-Fit algorithm to assign wavelengths to the given traffic requests. For each connection request, the first available wavelength is assigned to it; if no existing wavelength can serve this request, a new wavelength is used. The RWA matrix ($\Delta$) is then generated according to assignment results and the rules of code generation in the above section. Assume that $W$ wavelengths are used. The RWA matrix may include several empty elements ({\em don't cares}). To complete the matrix, each {\em don't care} is assigned the same value as the one immediately above it in the same column. If the {\em don't care} occurs in the first row, it is assigned the value of 0. The new matrix after assigning {\em don't cares} is denoted as $\overline{\Delta}$. Assume that there are two lightpaths between every pair of nodes in the small network shown in Figure \ref{smallnet}. The port ids are shown around each node in the network. There are 4 nodes in the network with degrees 3, 2, 2, 3, respectively. Correspondingly, there are 6, 3, 3, 6 bits for each node resulting in 18 bits for each wavelength. Figure \ref{RWA_lightpath} shows an RWA example for the given traffic. The corresponding RWA matrix is generated as in Table \ref{RWA_generation}. The entry $(i,j)$ for a node denotes the routing behavior between port $i$ and port $j$ of that node, while the entry $k$ for a node denotes the add/drop behavior of a wavelength at port $k$ of that node. The $\times$s in the matrix denote the {\em don't cares}. The generated RWA matrix after {\em don't cares} assignment is shown in Figure \ref{RWAmatrix}.

\begin{figure}
 \begin{tabular}{cc}
    \begin{minipage}[t]{2.6in}
    \centering
    \includegraphics[scale=0.66]{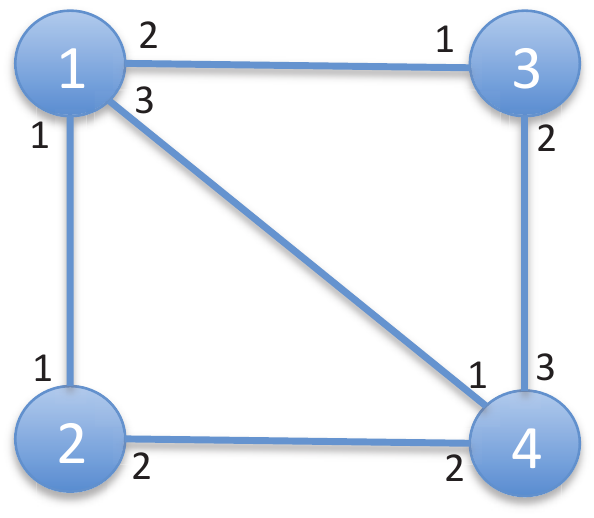}
    \caption{A 4-node mesh network example.}
    \label{smallnet}
    \end{minipage}
    \hspace{0.1in}
    \begin{minipage}[t]{2.8in}
    \centering
    \includegraphics[scale=0.55]{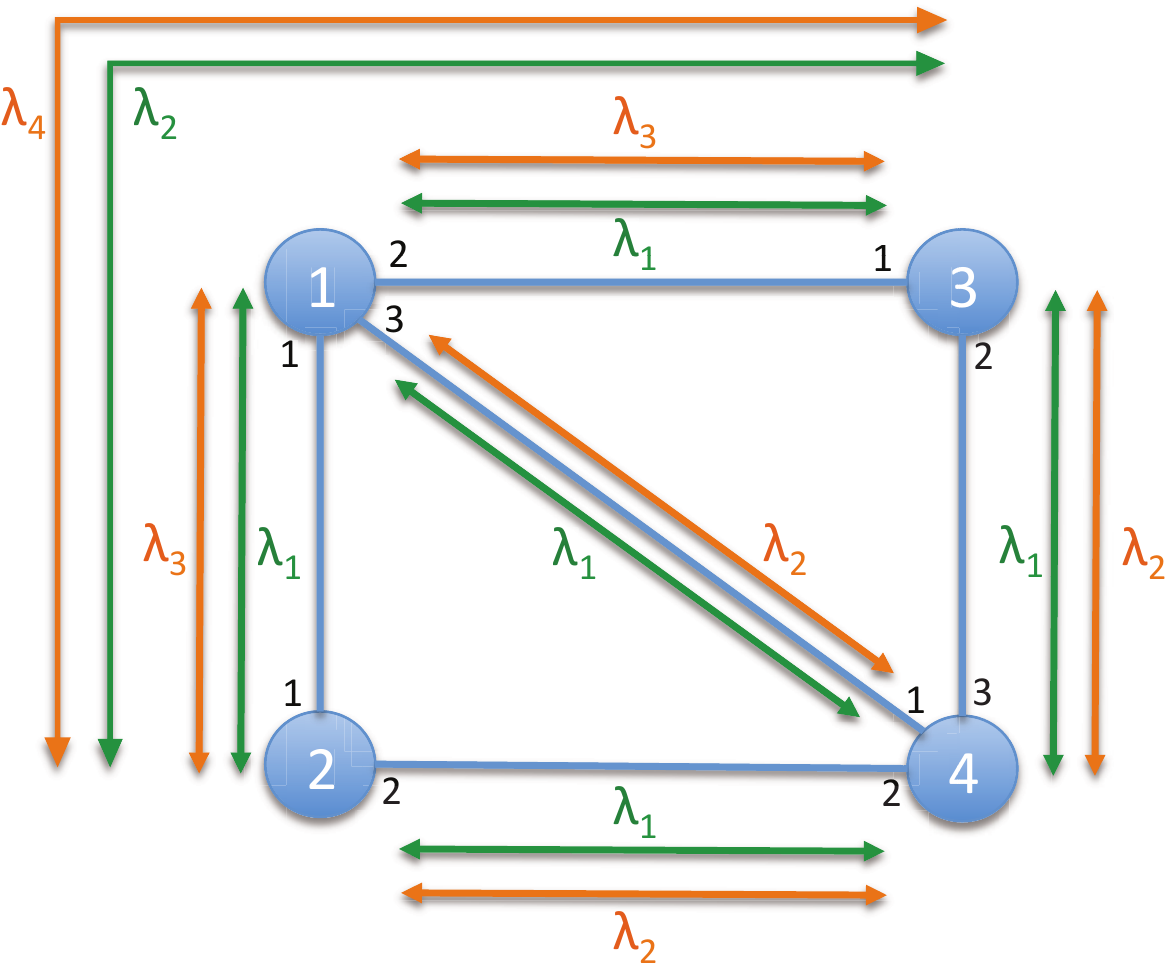}
    \caption{A routing and wavelength assignment example for given traffic.}
    \label{RWA_lightpath}
    \end{minipage}
 \end{tabular}
\end{figure}

\renewcommand\arraystretch{1.2}
\begin{table}\scriptsize
\caption{\label{RWA_generation} RWA Matrix Generation.}
\centering
\begin{tabular}{|c|c|c|c|c|c|c|c|c|c|c|c|c|c|c|c|c|c|c|}
\hline
\multirow{2}{*}{} & \multicolumn{6}{c|}{\bf{Node 1}} & \multicolumn{3}{c|}{\bf{Node 2}} & \multicolumn{3}{c|}{\bf{Node 3}} & \multicolumn{6}{c|}{\bf{Node 4}} \\
\cline{2-7} \cline{8-10} \cline{11-13} \cline{14-19}
& (1,2) & (1,3) & (2,3) & 1 & 2 & 3 & (1,2) & 1 & 2 & (1,2) & 1 & 2 & (1,2) & (1,3) & (2,3) & 1 & 2 & 3\\
\hline
\bf{$\lambda_1$} & 0 & 0 & 0 & 1 & 1 & 1 & 0 & 1 & 1 & 0 & 1 & 1 & 0 & 0 & 0 & 1 & 1 & 1 \\ \hline
\bf{$\lambda_2$} & 1 & 0 & 0 & 0 & 0 & 1 & 0 & 1 & 1 & 0 & 1 & 1 & 0 & 0 & 0 & 1 & 1 & 1 \\ \hline
\bf{$\lambda_3$} & 0 & 0 & 0 & 1 & 1 & $\times$ & 0 & 1 & $\times$ & 0 & 1 & $\times$ & $\times$ & $\times$ & $\times$ & $\times$ & $\times$ & $\times$ \\ \hline
\bf{$\lambda_4$} & 1 & 0 & 0 & 0 & 0 & $\times$ & 0 & 1 & $\times$ & 0 & 1 & $\times$ & $\times$ & $\times$ & $\times$ & $\times$ & $\times$ & $\times$ \\
\hline
\end{tabular}
\vspace{10pt}
\end{table}

\begin{figure}
\begin{center}
\includegraphics[scale=0.55]{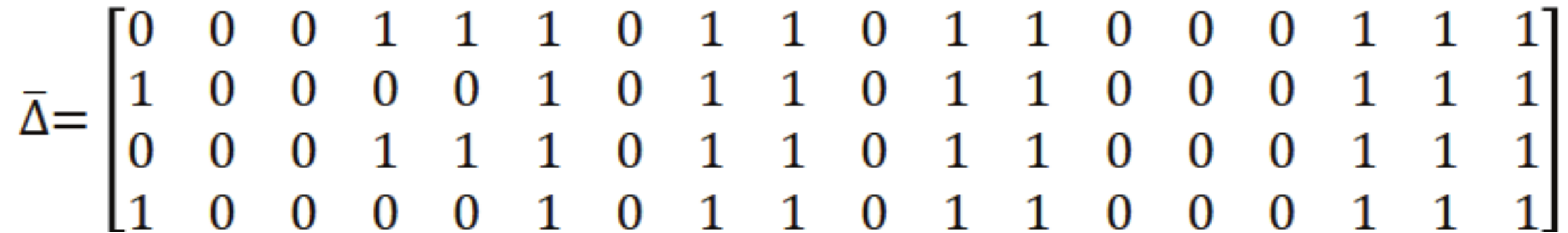}
\caption{RWA matrix example.}
\label{RWAmatrix}
\end{center}
\end{figure}

\indent The second step is to construct $W$ directed graphs according to the RWA matrix with each row being the first node to be visited in the Hamiltonian Path. For each graph, a Minimum Weight Hamiltonian Path Problem \cite{yao2013jop, moon2017understanding, yao-ijpp, liu17graphene, xia2018uksm,yao_hpca18, yan2018parallel, liu2016ibfs, liu2015enterprise} with given first node is to be solved and the desired result is the minimum result among those $W$ solutions. Figure~\ref{HMP1} gives an example of the constructed graph when row 1 in $\overline{\Delta}$ of Figure~\ref{RWAmatrix} is chosen as the first node. Each node corresponds to a row in $\overline{\Delta}$. The number on each link denotes the link weight, which is the distance between the corresponding two rows. For this graph, a Hamiltonian Path with total weight 11 is chosen as the minimum weight Hamiltonian Path. The path is shown as red lines in the figure. The constructed graph when row 2 in $\overline{\Delta}$ is chosen as the first node and the resulting Hamiltonian Path is shown in Figure~\ref{HMP2}. For each of the other two graphs generated by choosing other rows in $\overline{\Delta}$ as the first node, a Hamiltonian Path can also be chosen. Then the minimum among the four weights is the final result for this example.

\begin{figure}
 \begin{tabular}{cc}
    \begin{minipage}[t]{2.8in}
    \centering
    \includegraphics[scale=0.53]{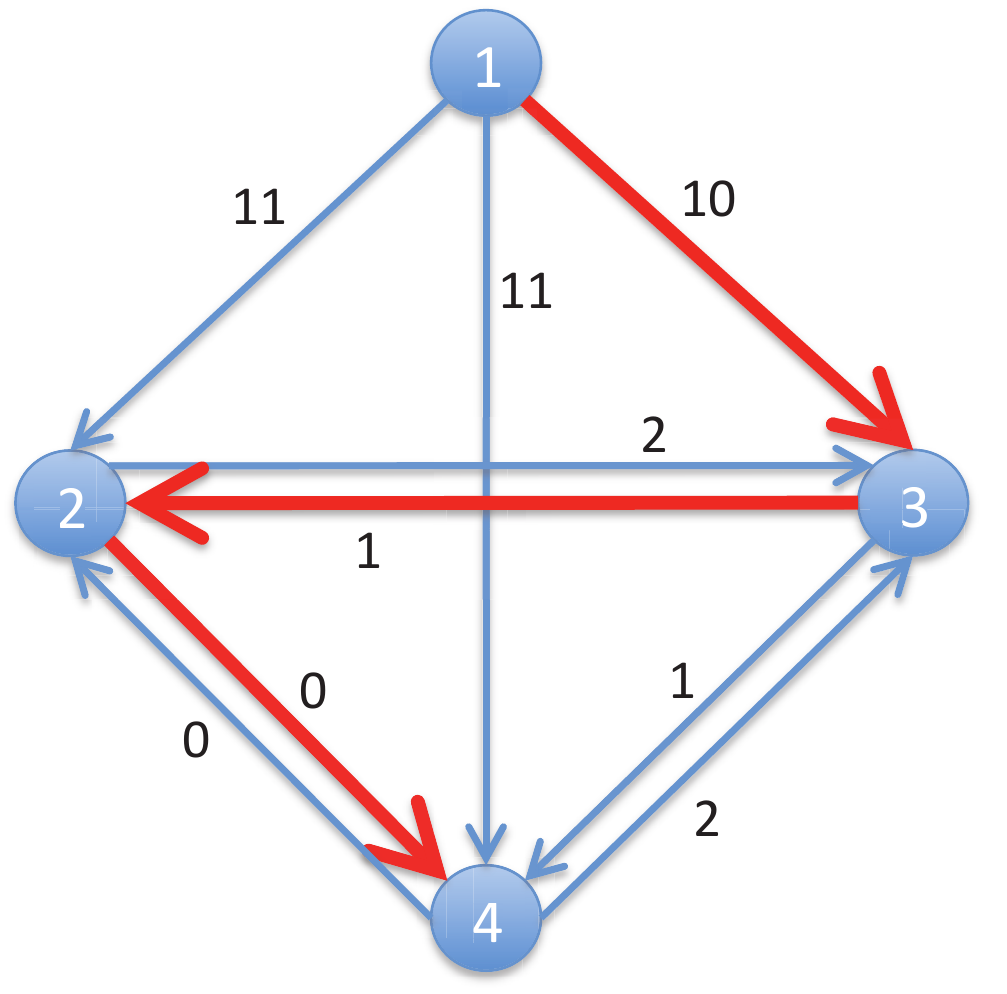}
    \caption{Graph generated from $\overline{\Delta}$ with first row chosen as the first node.}
    \label{HMP1}
    \end{minipage}
    \hspace{0.2in}
    \begin{minipage}[t]{2.8in}
    \centering
    \includegraphics[scale=0.53]{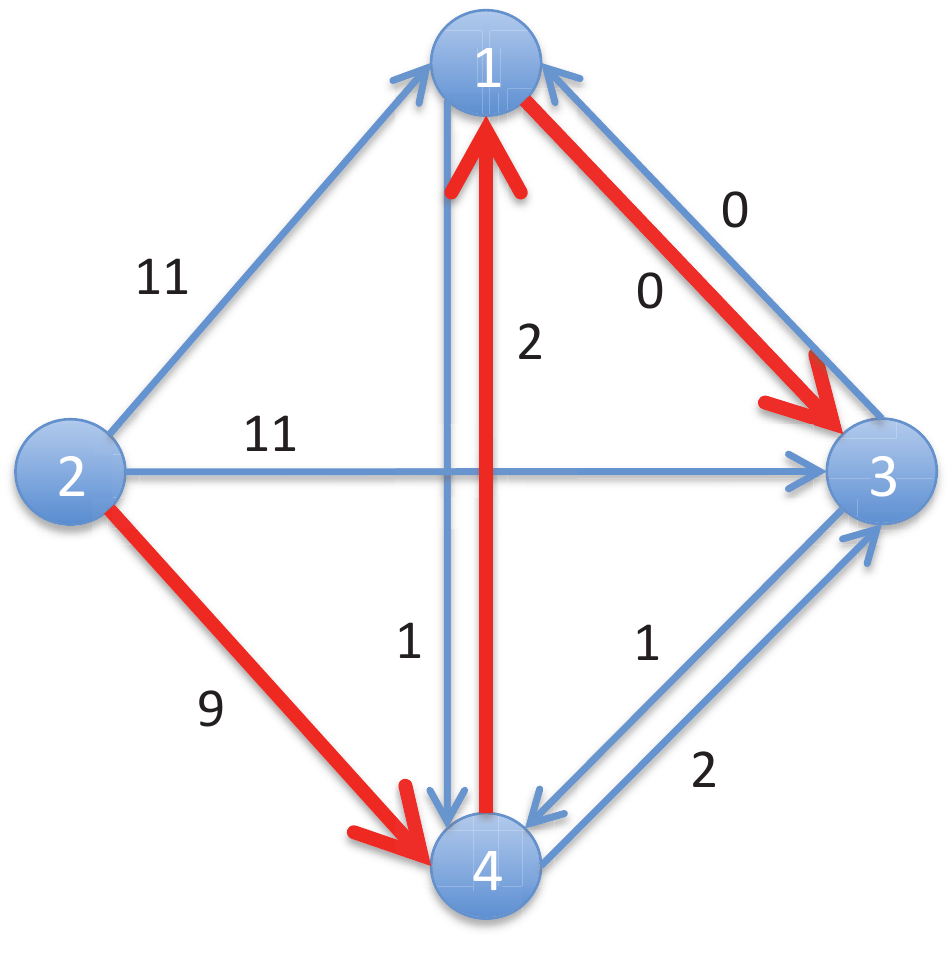}
    \caption{Graph generated from $\overline{\Delta}$ with second row chosen as the first node.}
    \label{HMP2}
    \end{minipage}
 \end{tabular}
\end{figure}

\indent We now present an ILP to solve the MWHPP for each of the $W$ constructed graphs.

\subsection{ILP Formulation}

\indent Given a constructed graph with chosen first node and corresponding link weights calculated from the complete RWA matrix, the goal is to find a Minimum Weight Hamiltonian Path Problem \cite{ liu2013gpu, prefetchguard18, mittal2016computational, prodact-ijpp, huang2014big, yao2017glsvlsi, zhao2017network, 2018asilomar} for this graph. The ILP formulation is as follows.

\begin{center}Objective: Minimize $ \sum_{i=1}^{W} \sum_{j=1,j\not=i}^{W}x_{i, j}\cdot d_{i,j},$\end{center}

\noindent
where $d_{i,j}$ is the link weight between node $i$ and $j$.

Variables:

a)
\begin{equation*}
x_{i, j}=
\begin{cases}
                1, \quad\makebox{if node } j \makebox{ is visited immediately following node } i; \\
                0, \quad\makebox{otherwise }
\end{cases}
 \end{equation*}
 
b) 
\begin{equation*}
a_{k, i}=
\begin{cases}
                1, \quad\makebox{if node } i \makebox{ is visited at step } k; \\
                0, \quad\makebox{otherwise }
\end{cases}
 \end{equation*}
 
 Constraints:
 
a) Each node must be visited exactly once.
\begin{equation*}
\sum_{k=1}^{W} a_{k,i} = 1, \forall i=1,2,\dots,W.
\end{equation*}

b) Only one node is visited at each step.
\begin{equation*}
\sum_{i=1}^{W} a_{k,i} = 1, \forall k=1,2,\dots,W.
\end{equation*}

c) If node $j$ is visited immediately following node $i$, then node $j$ is one step further than node $i$.
\begin{equation*}
\makebox{if} \mspace{5mu}  x_{i,j} = 1, \makebox{then} \mspace{5mu} \sum_{k=1}^{W} k\cdot a_{k,j} - \sum_{k=1}^{W} k\cdot a_{k,i}=1, \forall i,j=1,2,\dots,W.
\end{equation*}

d) When a node is visited at the first step, which means that this node is the starting node of the Hamiltonian Path, there should be no node visited before this node. Otherwise, for each node visited at later steps, there should be exactly one node visited just before this node.
\begin{equation*}
\sum_{j=1,j\not=i}^{W} x_{j,i} + a_{1,i} = 1, \forall i=1,2,\dots,W.
\end{equation*}

e) When a node is visited at the last step, which is the $W$th step, there should be no node visited after this node. Otherwise, there should be exactly one node visited immediately following this node.
\begin{equation*}
\sum_{j=1,j\not=i}^{W} x_{i,j} + a_{W,i} = 1, \forall i=1,2,\dots,W.
\end{equation*}

\subsection {Heuristic Solutions}

\indent We now present three heuristic algorithms to solve the MWHPP. 

\subsubsection{Nearest Neighbor (NN)}
\indent The RWA matrix with all {\em don't cares} assigned ($\overline{\Delta}$) is the input to this algorithm. Each row in $\overline{\Delta}$ can be visited as the first node in the path. For each chosen first row, a distance matrix $D$ can be generated according to the definitions in the previous section. The first step is to find the minimum element in the row corresponding to the chosen first row id in the distance matrix. This means that the row having minimum distance from the first row is chosen. For example, if $D_{1,i}$ is the minimum element in the first row of $D$, then row $i$ is chosen to be the second node in the Hamiltonian path and row $i$ is marked as visited. In the next step, find the minimum element $D_{i,j}$ for unvisited row $j$'s in the $i$th row of $D$, then row $j$ is chosen as the next node to be visited in the Hamiltonian path and marked as visited. Repeat the steps until all rows are visited. For each graph generated by each chosen first row, a Hamiltonian path is created. The path with minimum total weight gives the desired order of rows. The outline of this heuristic is shown in Algorithm~\ref{NN}. In Algorithm~\ref{NN}, $BandsCal$ is a function to calculate the total number of bands in a given matrix. $DistanceMatrixGen$ is a function to generate the distance matrix $D$ according to the first row id and the distance calculation rule described in the previous section. The algorithm $NearestNeighbor$ is used to determine new order of the matrix given the $DistanceMatrixGen$ and the first row id.

\begin{algorithm}
\caption{Pseudocode for the Nearest Neighbor heuristic}
\label{NN}
\begin{algorithmic}[1]
\REQUIRE $\overline{\Delta}$
\ENSURE $\Delta^\prime$ with reordered rows of $\overline{\Delta}$\\
\STATE {\em Initialization}: $min\_total\_bands$ = BandsCal($\overline{\Delta}$);\
\FOR {chosen first row id $f=1,2,\ldots, W$}
\STATE $D$ = DistanceMatrixGen($\overline{\Delta}$, $f$);\\
\STATE $[RowOrder, minTotalWeight]$=NearestNeighbour($D$, $f$);\\
	\IF {$minTotalWeight < min\_total\_bands$}
		\STATE $min\_total\_bands=minTotalWeight$;\\
		\STATE $NewRowOrder = RowOrder$;\\
	\ENDIF
\ENDFOR
\end{algorithmic}
\end{algorithm}

\begin{algorithm}
\caption{NearestNeighbor($D$, $f$)}
\begin{algorithmic}[1]
\STATE Initialize $OrigRows$ as a set of all row ids;
\STATE Insert $f$ into $RowOrder$ as the first row and remove $f$ from $OrigRows$;
\STATE $currentRow=f$;
\WHILE	{$OrigRows$ not empty}
	\STATE $mindis=\infty$;
	\FORALL { $j$ in $OrigRows$ }
		\IF {$D[currentRow][j]<mindis$}
			\STATE $mindis=D[currentRow][j]$;
			\STATE $nextRow=j$;
		\ENDIF
	\ENDFOR
	\STATE remove $nextRow$ from $OrigRows$ and insert it into $RowOrder$;
	\STATE $minTotalWeight=minTotalWeight+mindis$;
	\STATE $currentRow=nextRow$;
\ENDWHILE
\RETURN $RowOrder, minTotalWeight$;
\end{algorithmic}
\end{algorithm}

\subsubsection{ReassignDC (RDC)}
\indent This heuristic is based on the previous Nearest Neighbor heuristic. It is an improvement on NN in terms of reassigning {\em don't cares} after getting the new order of rows. In this case, the original RWA ($\Delta$) is reordered according to the $NewRowOrder$ and then each {\em don't care} is assigned the same value as its adjacent row above it in the same column to form $\Delta^\prime$. Then $BandsCal$ is applied to $\Delta^\prime$ to give the total number of bands in this new matrix. The Pseudocode for ReassignDC is in Algorithm~\ref{RDC}.
\begin{algorithm}
\caption{Pseudocode for the ReassignDC heuristic}
\label{RDC}
\begin{algorithmic}[1]
\REQUIRE $\overline{\Delta}$
\ENSURE $\Delta^\prime$ with reordered rows of $\Delta$ and reassigned {\em don't cares}\\
\STATE {\em Initialization}: $min\_total\_bands$ = BandsCal($\overline{\Delta}$);\
\FOR {chosen first row id $f=1,2,\ldots, W$}
\STATE $D$ = DistanceMatrixGen($\overline{\Delta}$, $f$);\\
\STATE $[RowOrder, minTotalWeight]$=NearestNeighbor($D$, $f$);\\
	\IF {$minTotalWeight < min\_total\_bands$}
		\STATE $min\_total\_bands=minTotalWeight$;\\
		\STATE $NewRowOrder = RowOrder$;\\
	\ENDIF
\ENDFOR
\STATE $\Delta^\prime$ equals to the reordering of $\Delta$ according to $NewRowOrder$;
\STATE Assign {\em don't cares} in $\Delta^\prime$;
\end{algorithmic}
\end{algorithm}

\subsubsection{SortTraffic (ST)}
\indent In this heuristic, the traffic requests are first sorted according to their path length in descending order. Then the FirstFit wavelength assignment algorithm is applied to generate the RWA matrix. Thus the RWA matrix and the original number of bands are different from those of the previous two heuristics. After sorting traffic and performing wavelength assignment, the ReassignDC heuristic is applied to this RWA matrix.

\indent The worst case time complexity of the $NearestNeighbor$ algorithm is $O(W^2)$, while that of the ILP is $O(2^{W^2})$. Since we need to run the $NearestNeighbor$ algorithm or the ILP $W$ times for $W$ different first rows, the worst-case time complexity is $O(W^3)$ for all heuristics and $O(W \cdot 2^{W^2})$ for the ILP.

\section{Simulation Results}\label{sec.label314}

\indent We present results for a small 5-node network (Figure~\ref{fig:network5}), two real network topologies - the NSFNET network (Figure~\ref{fig:NSF}) and the pan-European network (Figure~\ref{fig:pan}), and a 5$\times$5 regular mesh network. The NSFNET network consists of 14 nodes and 21 links, while the pan-European network consists of 28 nodes and 43 links. For the small 5-node network topology, we compare the total number of nonuniform wavebands calculated by the ILP and heuristics to the total number of switching elements in a wavelength switching architecture. Due to the high complexity of the ILP, we are only able to obtain results for the small network. Therefore, only the results in terms of the total number of nonuniform wavebands achieved from the heuristics are compared to the total number of switching elements in a wavelength switching architecture for the NSF, pan-European and 5$\times$5 regular mesh networks.

\begin{figure}
 \begin{tabular}{cc}
    \begin{minipage}[t]{2.7in}
    \centering
    \includegraphics[scale=0.58]{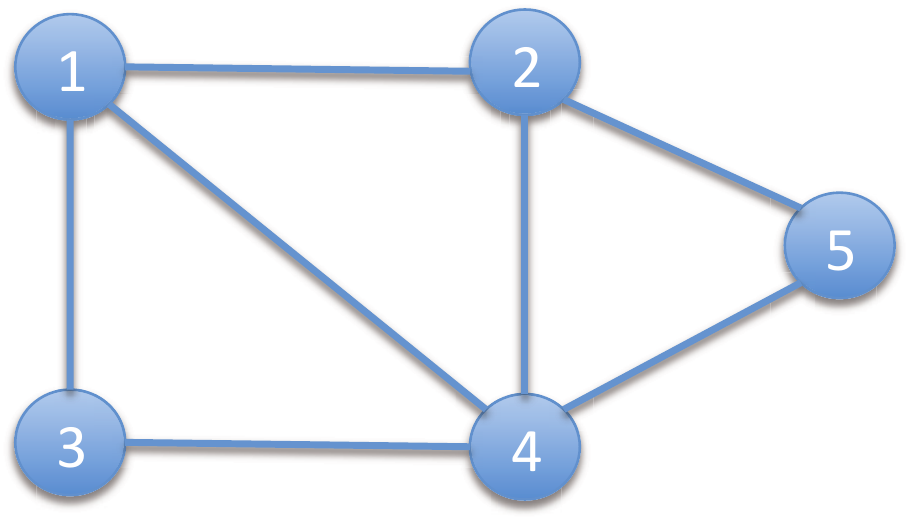}
    \caption{\label{fig:network5} A small 5-node network.}
    \end{minipage}
    \hspace{0.1in}
    \begin{minipage}[t]{3in}
     \centering
    \includegraphics[scale=0.3]{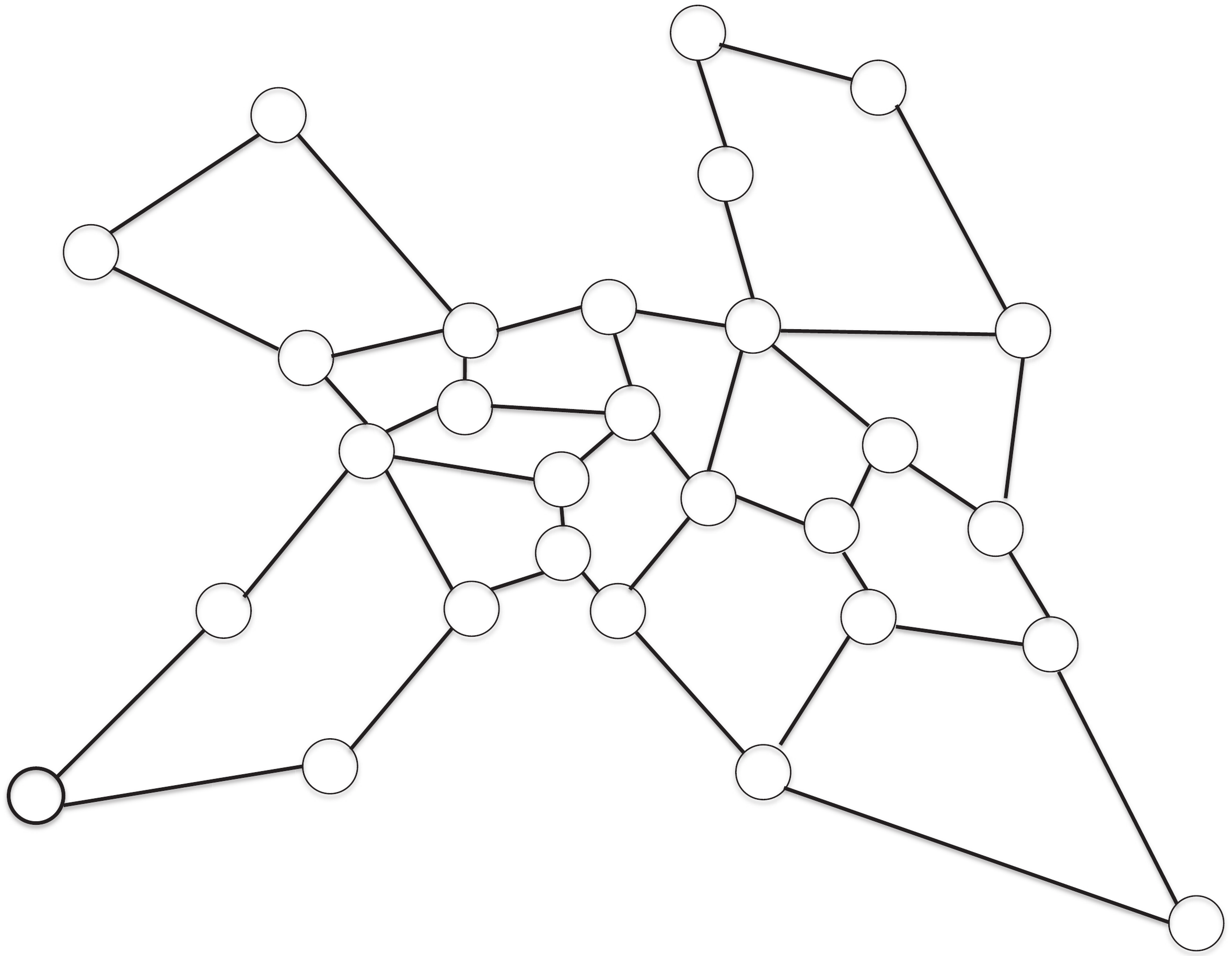}
    \caption{\label{fig:pan} 28-Node Pan-European network.}
    \end{minipage}
 \end{tabular}
\end{figure}
 
\indent The traffic requests are generated as follows. Each request represents a connection between a pair of nodes in the mesh network. The source node and the destination node are uniformly randomly chosen from the node set. Dijkstra's algorithm is applied to find the shortest path between the source and destination nodes. 

\begin{table}
\caption{\label{table:WDMwavebanding_ilp} Performance of ILP and algorithms for 5-node network.}
\centering
\scriptsize
\begin{tabular}{|p{2.1cm}<{\centering}|p{0.8cm}<{\centering}|p{1cm}<{\centering}|p{1.5cm}<{\centering}|p{1cm}<{\centering}|p{1.5cm}<{\centering}|p{1cm}<{\centering}|p{1.5cm}<{\centering}|p{0.5cm}<{\centering}|}
\hline
\multirow{3}{*}{\bf{\# of Requests}} & \multirow{3}{*}{\bf{SWN}} & \multicolumn{2}{c|}{\bf{NN}}& \multicolumn{2}{c|}{\bf{RDC}} & \multicolumn{2}{c|}{\bf{ST}} & \multirow{3}{*}{\bf{ILP}}\\
\cline{3-4} \cline{5-6} \cline{7-8}
&& Orig. \# of bands & Optimized \# of bands & Orig. \# of bands & Optimized \# of bands & Orig. \# of bands & Optimized \# of bands & \\
\hline
\bf{20} &43& 20 & 16 & 20 & 16 & 20 & 16 & 16\\ \hline
\bf{30} &64& 27 & 17 & 27 & 17 & 22 & 18 & 16\\ \hline
\bf{40} & 87& 34 & 17 & 34 & 17 & 28 & 18 & 16\\ \hline
\end{tabular}
\vspace{10pt}
\end{table}

\indent Results for the small 5-node mesh network are shown in Table~\ref{table:WDMwavebanding_ilp}. $SWN$ in the table refers to the number of switching elements used in the original wavelength switching network without banding. For each heuristic, the original number of bands (before applying the MWHPP algorithm) and the new number of bands achieved by the heuristic are presented. Since NN and RDC apply the routing and wavelength assignment to the requests in the order that they are generated, the original number of bands are the same for these two heuristics. ST first sorts the requests in descending order of their path length, thus the RWA matrix, and hence the original number of bands, for ST differs from others. From the table, we can see that the results of heuristics are very close to the ILP results, thus demonstrating the effectiveness of the heuristics. Besides, the reduction in bands is significant when comparing $SWN$ with the number of bands achieved by the heuristics. 

\begin{table}
\caption{\label{table:WDMwavebanding_nsf} Performance of algorithms for NSF network.}
\centering
\scriptsize
\begin{tabular}{|p{2.2cm}<{\centering}|p{0.8cm}<{\centering}|p{1.1cm}<{\centering}|p{1.6cm}<{\centering}|p{1.1cm}<{\centering}|p{1.6cm}<{\centering}|p{1.1cm}<{\centering}|p{1.6cm}<{\centering}|}
\hline
\multirow{3}{*}{\bf{\# of Requests}} & \multirow{3}{*}{\bf{SWN}} & \multicolumn{2}{c|}{\bf{NN}}& \multicolumn{2}{c|}{\bf{RDC}} & \multicolumn{2}{c|}{\bf{ST}} \\
\cline{3-4} \cline{5-6} \cline{7-8}
&& Orig. \# of bands & Optimized \# of bands & Orig. \# of bands & Optimized \# of bands & Orig. \# of bands & Optimized \# of bands\\
\hline
\bf{50} &170& 117 & 115 & 117 & 103 & 120 & 101 \\ \hline
\bf{500} &1579& 865 & 655 & 865 & 596 & 844 & 452 \\ \hline
\bf{5000} & 15606 & 8360 & 4607 & 8360 & 3924 & 7884 & 2538 \\ \hline
\end{tabular}
\vspace{10pt}
\end{table}

\begin{table}
\caption{\label{table:WDMwavebanding_pan} Performance of algorithms for pan-European network.}
\centering
\scriptsize
\begin{tabular}{|p{2.2cm}<{\centering}|p{0.8cm}<{\centering}|p{1.1cm}<{\centering}|p{1.6cm}<{\centering}|p{1.1cm}<{\centering}|p{1.6cm}<{\centering}|p{1.1cm}<{\centering}|p{1.6cm}<{\centering}|}
\hline
\multirow{3}{*}{\bf{\# of Requests}} & \multirow{3}{*}{\bf{SWN}} & \multicolumn{2}{c|}{\bf{NN}}& \multicolumn{2}{c|}{\bf{RDC}} & \multicolumn{2}{c|}{\bf{ST}} \\
\cline{3-4} \cline{5-6} \cline{7-8}
&& Orig. \# of bands & Optimized \# of bands & Orig. \# of bands & Optimized \# of bands & Orig. \# of bands & Optimized \# of bands\\
\hline
\bf{50} &186 & 131 & 130 & 131 & 129 & 123 & 120 \\ \hline
\bf{500} &2061& 1141 & 1083 & 1141 & 940 & 1060 & 823 \\ \hline
\bf{5000} & 20384& 10734 & 8980 & 10734 & 7419 & 9635 & 5070 \\ \hline
\end{tabular}
\vspace{10pt}
\end{table}

\begin{table}
	\caption{\label{table:WDMwavebanding_mesh} Performance of algorithms for 5$\times$5 regular mesh network.}
	\centering
	\scriptsize
	\begin{tabular}{|p{2.2cm}<{\centering}|p{0.8cm}<{\centering}|p{1.1cm}<{\centering}|p{1.6cm}<{\centering}|p{1.1cm}<{\centering}|p{1.6cm}<{\centering}|p{1.1cm}<{\centering}|p{1.6cm}<{\centering}|}
		\hline
		\multirow{3}{*}{\bf{\# of Requests}} & \multirow{3}{*}{\bf{SWN}} & \multicolumn{2}{c|}{\bf{NN}}& \multicolumn{2}{c|}{\bf{RDC}} & \multicolumn{2}{c|}{\bf{ST}} \\
		\cline{3-4} \cline{5-6} \cline{7-8}
		&& Orig. \# of bands & Optimized \# of bands & Orig. \# of bands & Optimized \# of bands & Orig. \# of bands & Optimized \# of bands\\
		\hline
		\bf{50} &211& 145 & 138 & 145 & 136 & 143 & 142 \\ \hline
		\bf{500} &2156 & 1211 & 1070 & 1211 & 957 & 1041 & 737 \\ \hline
		\bf{5000} & 21706 & 11367 & 8328 & 11367 & 7028 & 9498 & 3943 \\ \hline
	\end{tabular}
\vspace{10pt}
\end{table}

\indent Table \ref{table:WDMwavebanding_nsf}, \ref{table:WDMwavebanding_pan} and \ref{table:WDMwavebanding_mesh} show the numerical results for the NSF, pan-European and 5$\times$5 regular mesh network, respectively. We can see that the $SWN$ and the number of bands for pan-European and 5$\times$5 regular mesh network are larger than that of the NSF network under the same traffic demand. This is caused by the longer average path length and more traversing switches for the connection requests in larger networks. From all the tables, we can see that heuristic ST works well in the large networks and for a large number of requests. For a better visualization, a comparison among the heuristics for NSF network is shown in Figure \ref{WDMwavebanding_comp}. Sorting the requests in descending order of the path length and then assigning wavelengths first to longer paths gives more opportunity for wavebanding. The heuristic itself achieves a large further reduction in the number of bands. There are also significant reductions in the number of bands when compared with $SWN$. This demonstrates that the number of switching elements can be reduced by a large amount using waveband switching compared to wavelength switching. 

\begin{figure}
	\centering
	\includegraphics[scale=0.5]{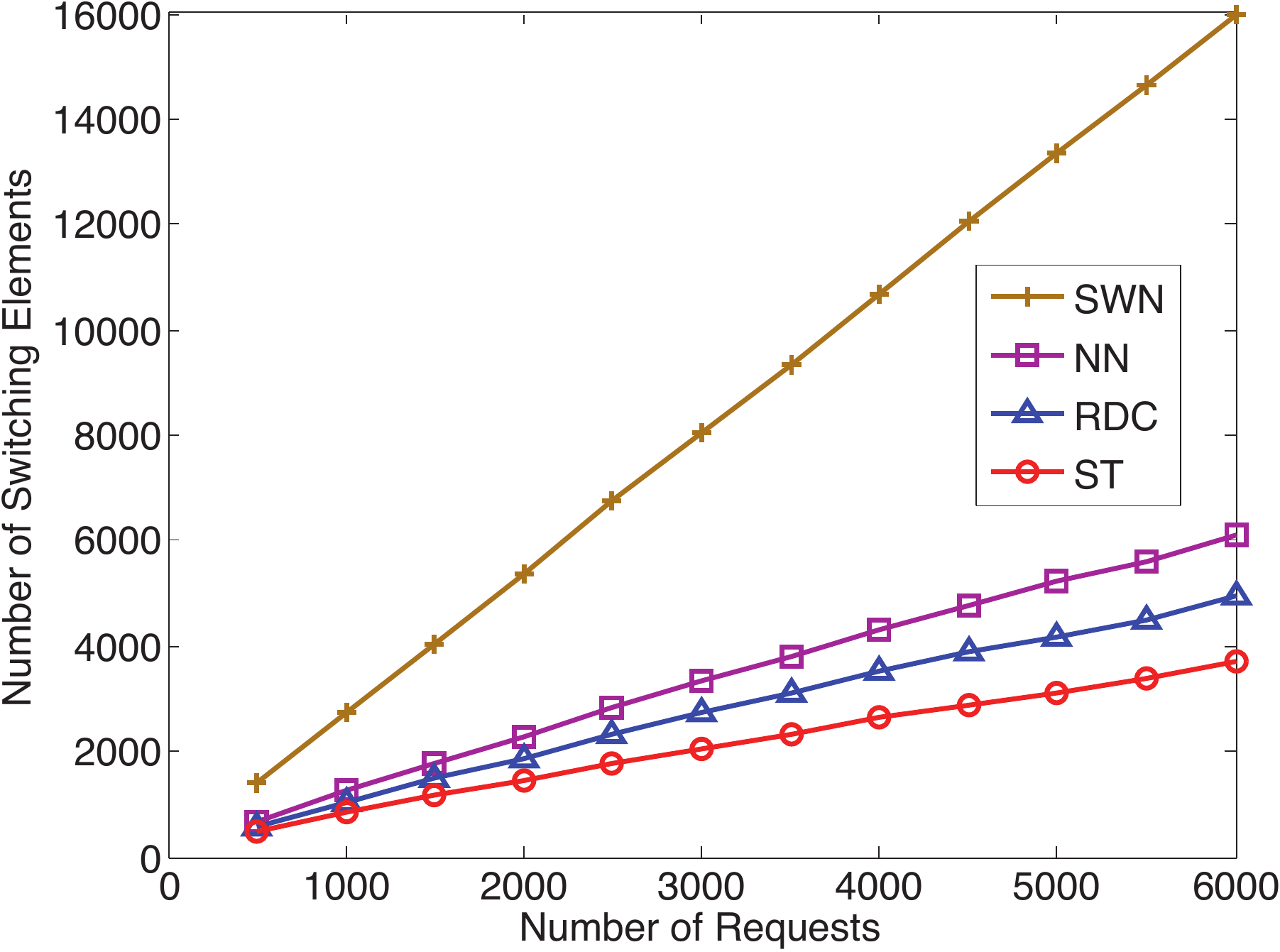}
	\caption{\label{WDMwavebanding_comp} A comparison of heuristic results for NSF network.}
\end{figure}

\section{Application to Dynamic Stochastic Traffic}\label{sec.label34}

\indent In this section, we apply the proposed band framework to the case of dynamic stochastic traffic. We first calculate the nonuniform bands for static all-to-all traffic (one lightpath between every pair of nodes) for the given network, and then assume that the wavelengths are banded accordingly. The minimum number of wavelengths required by the all-to-all traffic is denoted by $W_m$. The total number of available wavelengths in the system, which is denoted by $W_s$, could be different from $W_m$. We assume that $W_s>W_m$. In this case, the bands calculated for the deterministic traffic can be proportionally expanded according to the ratio of $W_s$ to $W_m$. Assume $r = W_s$mod$W_m$ and $n = \floor*{W_s/W_m}$. Given the reordered RWA matrix $\Delta^\prime$ for all-to-all traffic, $r$ rows of $\Delta^\prime$ are duplicated $n$ times, while the remaining $W_m - r$ rows need to be duplicated $n-1$ times.

\indent The dynamic traffic requests arriving to the network are generated according to a Poisson process in our study. Each request has exponential holding time. The source and destination nodes for each request are uniformly randomly chosen from the node set. To accommodate a connection request, a wavelength with available waveband switches on all nodes along the path is needed. 

\indent The First-Fit algorithm is adapted to make wavelength assignment decisions for incoming connection requests. For each source-destination node pair, a wavelength set is pre-calculated based on the designated wavelength in $\Delta^\prime$ for the all-to-all traffic and the expansion metrics $r$ and $n$. For each incoming connection request, we first check the designated wavelength set for this node pair starting from the first wavelength in the set. If no wavelength in this set is available for the request, the availabilities of other remaining wavelengths are examined one by one and the first available wavelength is chosen to accommodate this request. If no wavelength is available at the time the request arrives, the connection request is blocked.

\indent We conduct performance evaluation through simulations, and compare the blocking probability of dynamic traffic demands for the proposed heuristics. In our simulation, each request has a mean holding time of 1 (arbitrary time unit), and the arrival rate of traffic requests is varied in order to examine the performance under varying offered loads. Wavelength conversion is not considered here. For the networks, we assume that each link has two fibers, each with 80 wavelengths. Thus $W_s=160$. 

\begin{figure}
	\centering
	\includegraphics[scale=0.52]{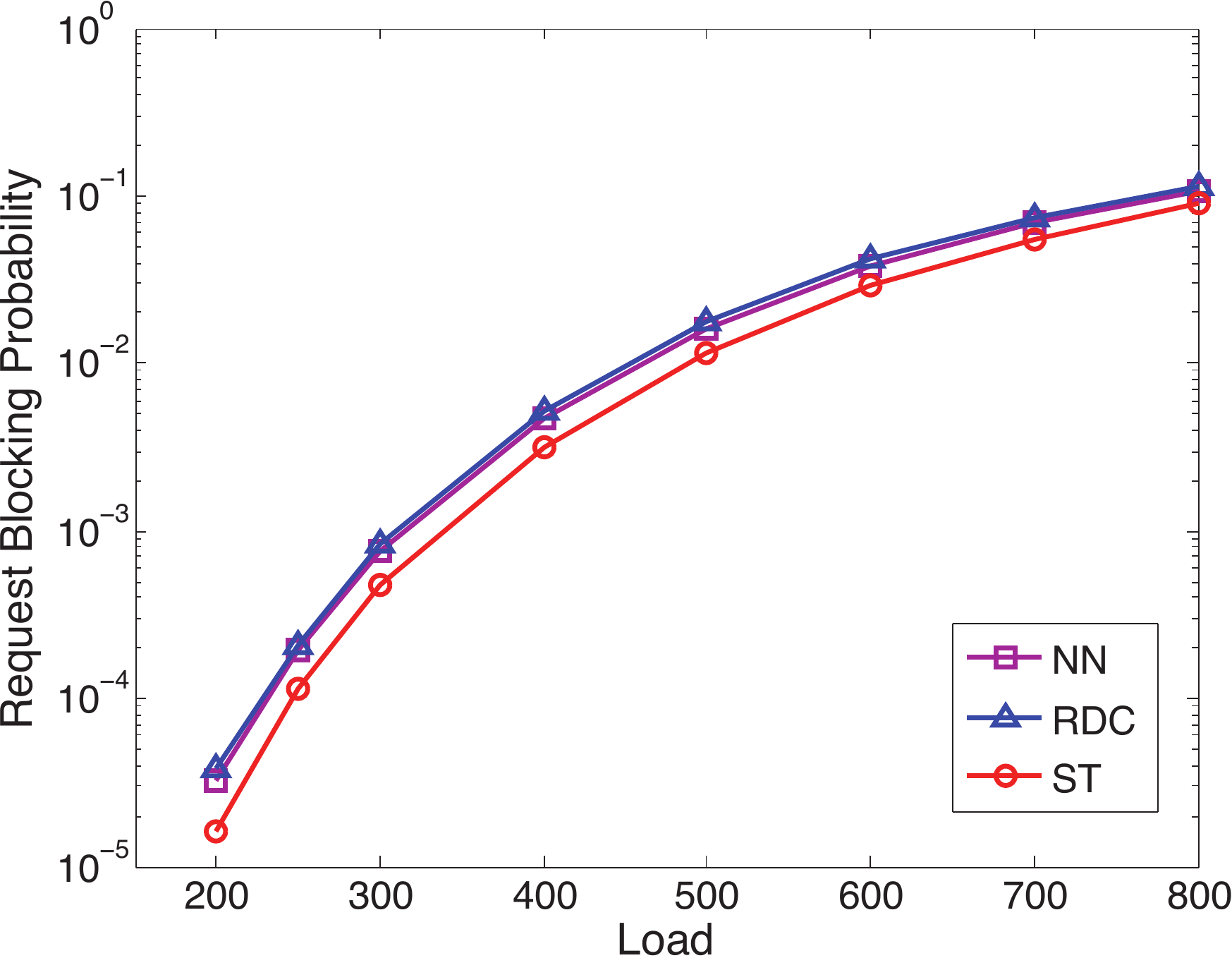}
	\caption{\label{WDMwavebanding_dynamic_NSF} Request blocking probability versus load in the NSF network for $W_s=160$.}
\end{figure}

\begin{figure}	
	\centering
	\includegraphics[scale=0.52]{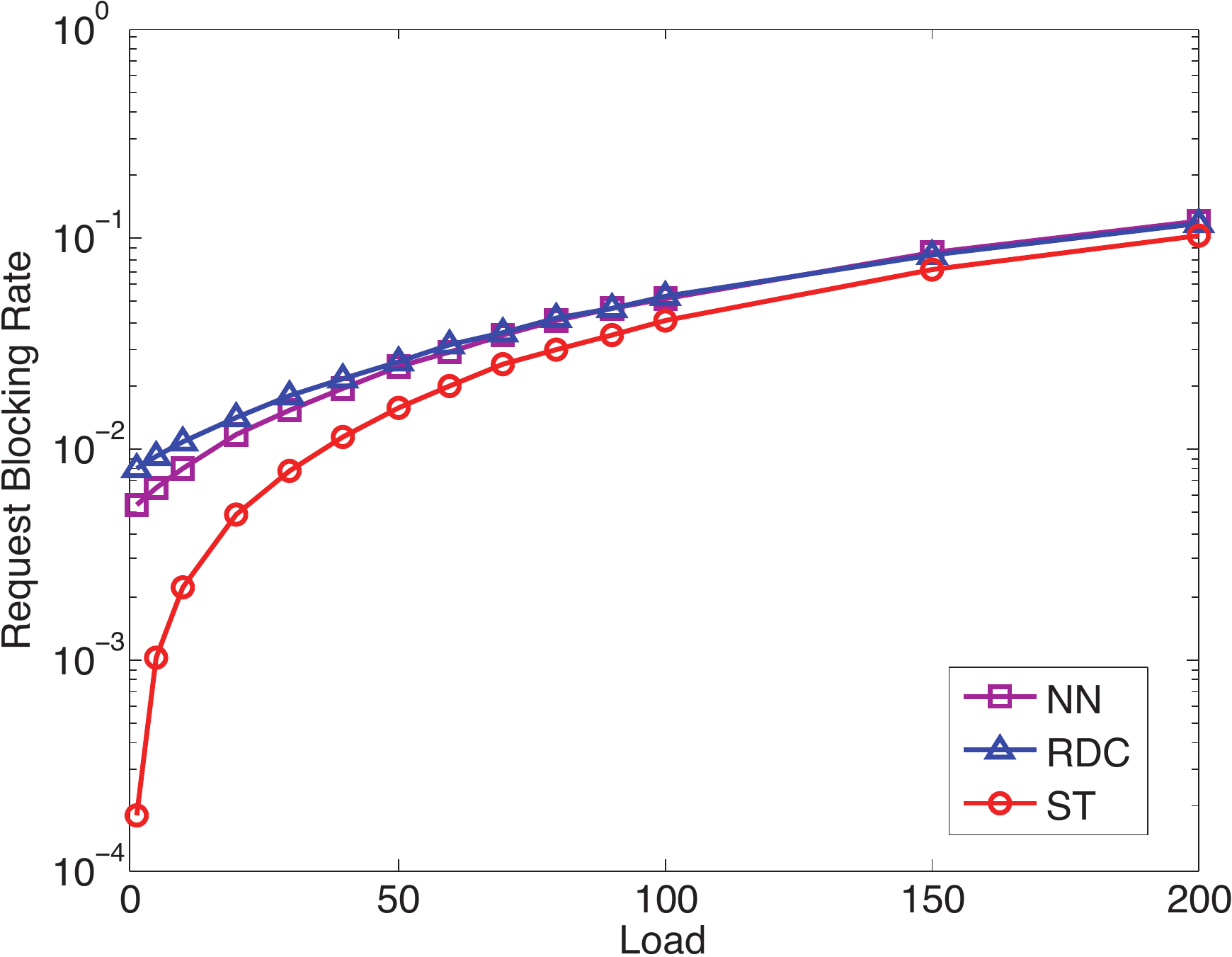}
	\caption{\label{WDMwavebanding_dynamic_euro} Request blocking probability versus load in the pan-European network for $W_s=160$.}
\end{figure}

\indent Figures \ref{WDMwavebanding_dynamic_NSF} and \ref{WDMwavebanding_dynamic_euro} show the blocking probability versus load for the NSF network and the pan-European network, respectively. Each point in the figures is obtained by simulating one million requests after 10000 network warm-up requests in order to obtain steady-state blocking rates. In Figure \ref{WDMwavebanding_dynamic_NSF}, the request arrival rate varies from 200 to 800, while in Figure \ref{WDMwavebanding_dynamic_euro}, the request arrival rate varies from 5 to 200. These ranges were chosen to obtain blocking rates in the desired range of approximately $0.1$ to $0.0001$. From the figures, we can see that ST works better than NN and RDC in terms of blocking rate of traffic demands. By using ST in the NSF network when $W_s=160$, the pre-calculated $W_m=28$, the $SWN=568$ and the number of required wavebands is 257. After the expansion described above, the number of required wavebands for this network is still 257, which is much less than the number of required switches if using wavelength switching. For the pan-European network, the pre-calculated $W_m=160$, the $SWN=3066$, and the number of required non-uniform wavebands for this network is 1253. We see that the performance for the pan-European network is worse than that of the NSF network when $W_s=160$. This is because $W_s=W_m$, and the set of available wavelengths for any connection request is limited in the pan-European network. Figure \ref{WDMwavebanding_BvsW_NSF} and Figure \ref{WDMwavebanding_BvsW_Euro}, respectively, show the request blocking probability versus number of wavelengths in the system ($W_s$) when using ST for NSF network at a load of 800, and the pan-European network at a load of 200.

\begin{figure}
	\begin{tabular}{cc}
		\begin{minipage}[t]{2.8in}
			\centering
			\includegraphics[scale=0.38]{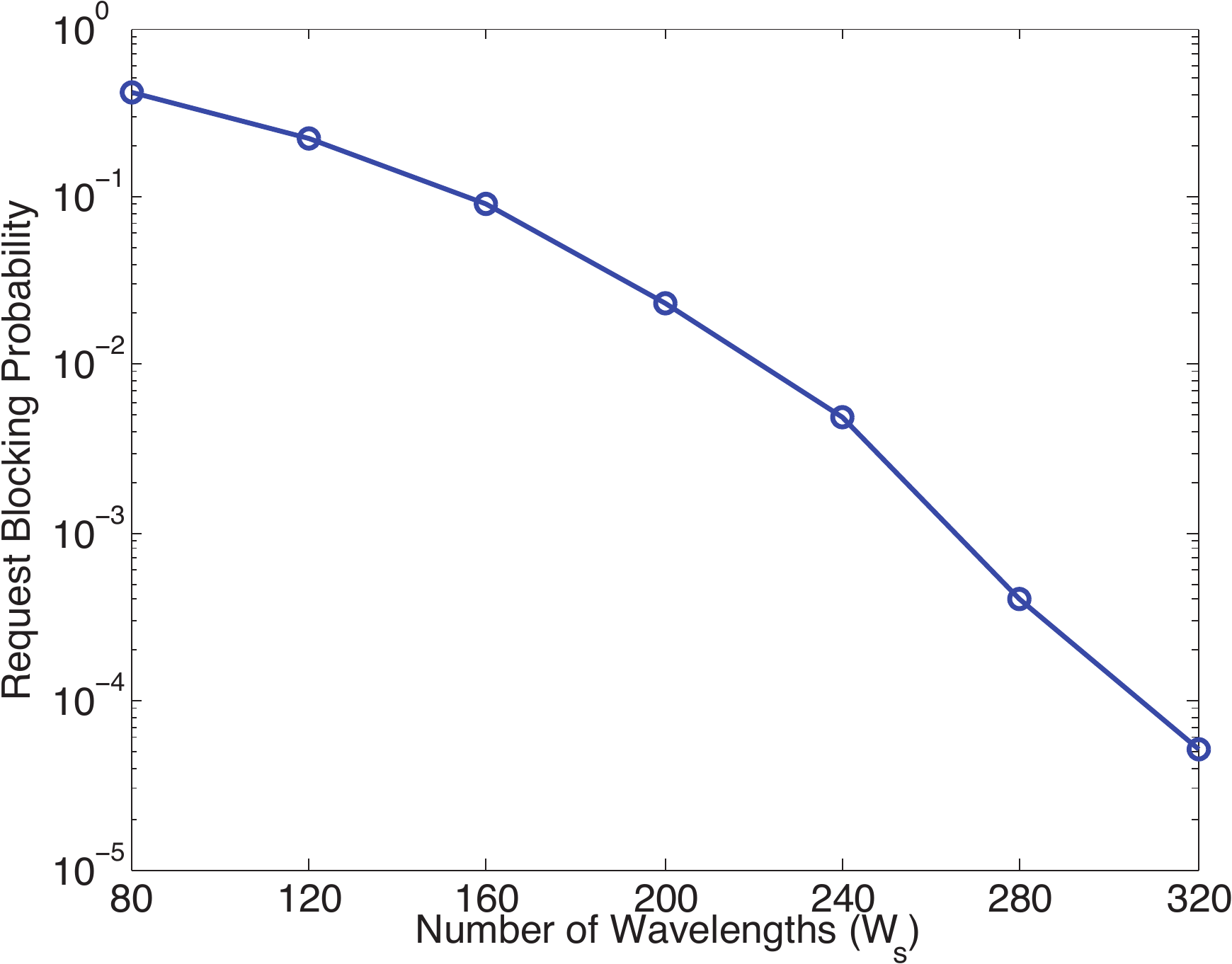}
			\caption{\label{WDMwavebanding_BvsW_NSF} Request blocking probability versus number of wavelengths in the NSF network at a load of 800 Erlangs.}
		\end{minipage}
		\hspace{0.1in}
		\begin{minipage}[t]{2.8in}
			\centering
			\includegraphics[scale=0.38]{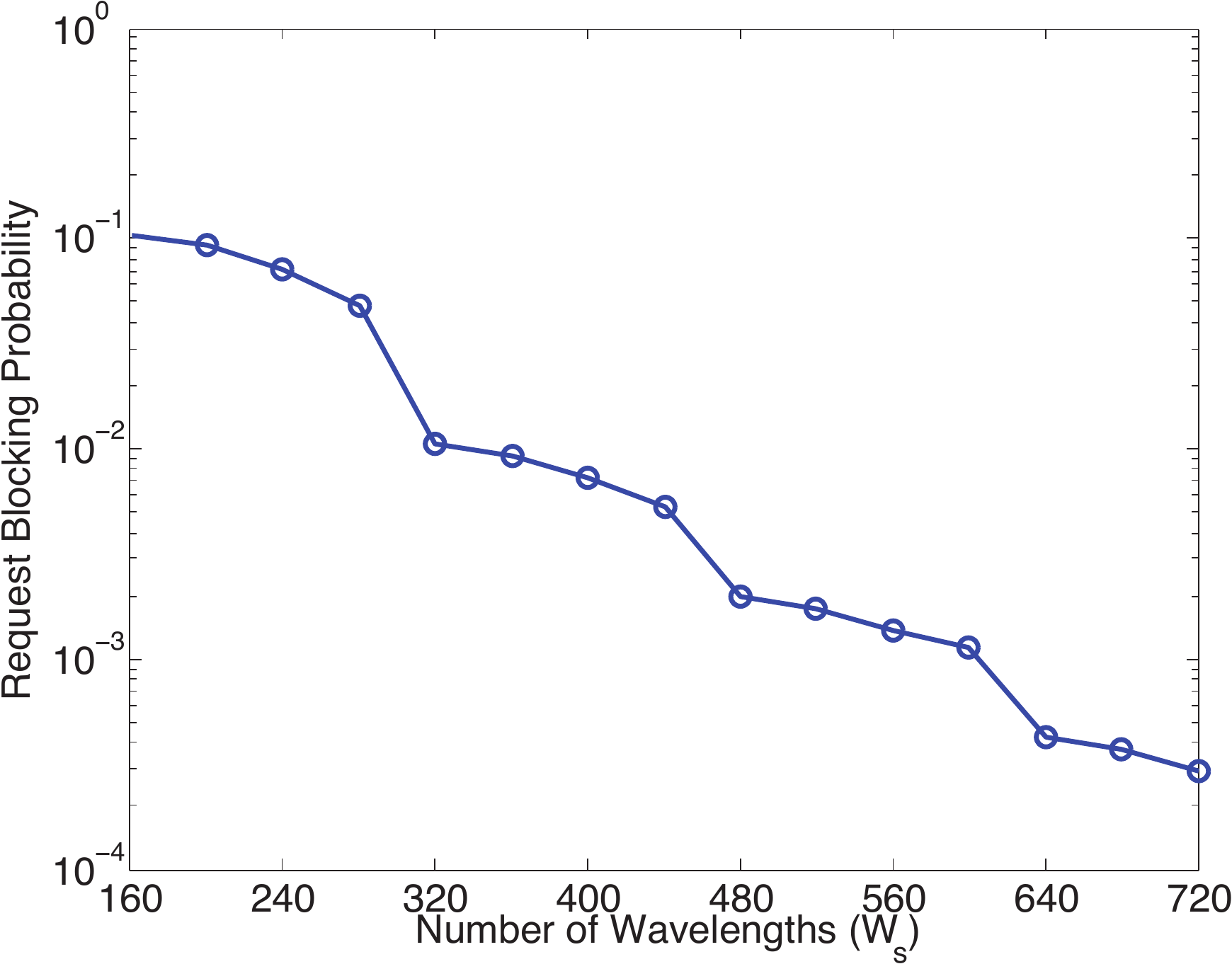}
			\caption{\label{WDMwavebanding_BvsW_Euro} Request blocking probability versus number of wavelengths in the pan-European network at a load of 200 Erlangs.}
		\end{minipage}
	\end{tabular}
\end{figure}

\section{Conclusion}\label{sec.label35}
\indent In this chapter, we proposed nonuniform wavebanding strategies to serve a given traffic in mesh networks \cite{wu2015optimal}. We formulated a band minimization problem and transformed it into a modified version of Minimum Weight Hamiltonian Path Problem. We presented an ILP formulation for this problem, and proposed heuristic algorithms for the problem. The numerical simulation results show that heuristic results are very close to ILP results for a small network. Numerical results for the comparison of number of wavebands and number of switching elements for three different larger networks were also presented. Significant reduction in terms of switching elements was achieved through wavebanding. We also applied our nonuniform waveband minimization framework to the dynamic traffic case, and the performance of the network in terms of request blocking probability with dynamic stochastic traffic was evaluated.

\chapter{Evaluation and Performance Modeling of OXC Architectures}
\label{chap_4}

\indent Despite the static configuration of wavebands in the previous chapter, this chapter introduces Optical Cross-Connect (OXC) node architectures built on top of wavelength-selective switches (WSSs) which facilitate dynamic configuration in WDM networks. An evaluation of Two Optical Cross-Connect (OXC) node architectures with multiple fibers per link -- one, a conventional architecture, and the second, a hierarchical architecture that has lower complexity than the first one is presented. Resource (fiber and wavelength) assignment and analytical models for computing the blocking probability of connection requests are studied.

\section{Related Work}\label{sec.label41}

\indent Current OXCs are built using WSSs with a single input and multiple outputs \cite{wall2008wss,yuan2008fully}. WSSs have the capability of demultiplexing, multiplexing, and switching. In a $1\times n$ WSS, each wavelength from the input can be independently switched to any of the $n$ outputs. Large port counts (i.e., large values of $n$) are required for large-scale OXCs. However, current technology limits the port count to around 20 in commercially available WSSs \cite{ban2013evaluation}. 

\indent A conventional way to build large port-count OXCs is to cascade small port-count WSSs. However, the nature of cascading leads to a square-order increment of the required hardware. To relieve the hardware requirement as well as costs, a novel OXC architecture utilizing hierarchical routing techniques was proposed in \cite{HIER}. This architecture can accommodate almost as much traffic as the conventional one, while using much fewer WSSs \cite{wu2015comparison}. For WDM networks, fixed-grid WSSs are utilized. Each optical path in the WDM network corresponds to a wavelength. The wavelengths of each fiber are located in a fixed channel frequency spacing, generally 50/100 GHz.

\indent One of the key performance metrics in optical networks is blocking probability, which is the probability that a connection request cannot be accommodated due to a lack of resources or certain constraints. There have been numerous models for studying the blocking performance of optical networks both at the network level~\cite{subramaniam1996all,zang2001dynamic,lu2002blocking,lu2004blocking,sridharan2004blocking} and the node level~\cite{yang2005node}.  

\section{OXC Node Architectures}\label{sec.label42}

\indent Suppose the OXC node has a physical degree of $D$, i.e., $D$ is the number of connected neighbor nodes or the number of input/output links. The number of parallel fibers on link $i$ is denoted as $f_i \; (i=1,2,\ldots,D)$, and $N = \sum_{i=1}^{D} f_i$ represents the total number of input/output fibers. $N$ is usually larger than the limited port count of WSSs. We will analyze an $N\times N$ OXC node with different implementations utilizing port-count-limited WSSs. 

\begin{figure}
	\includegraphics[scale=0.5]{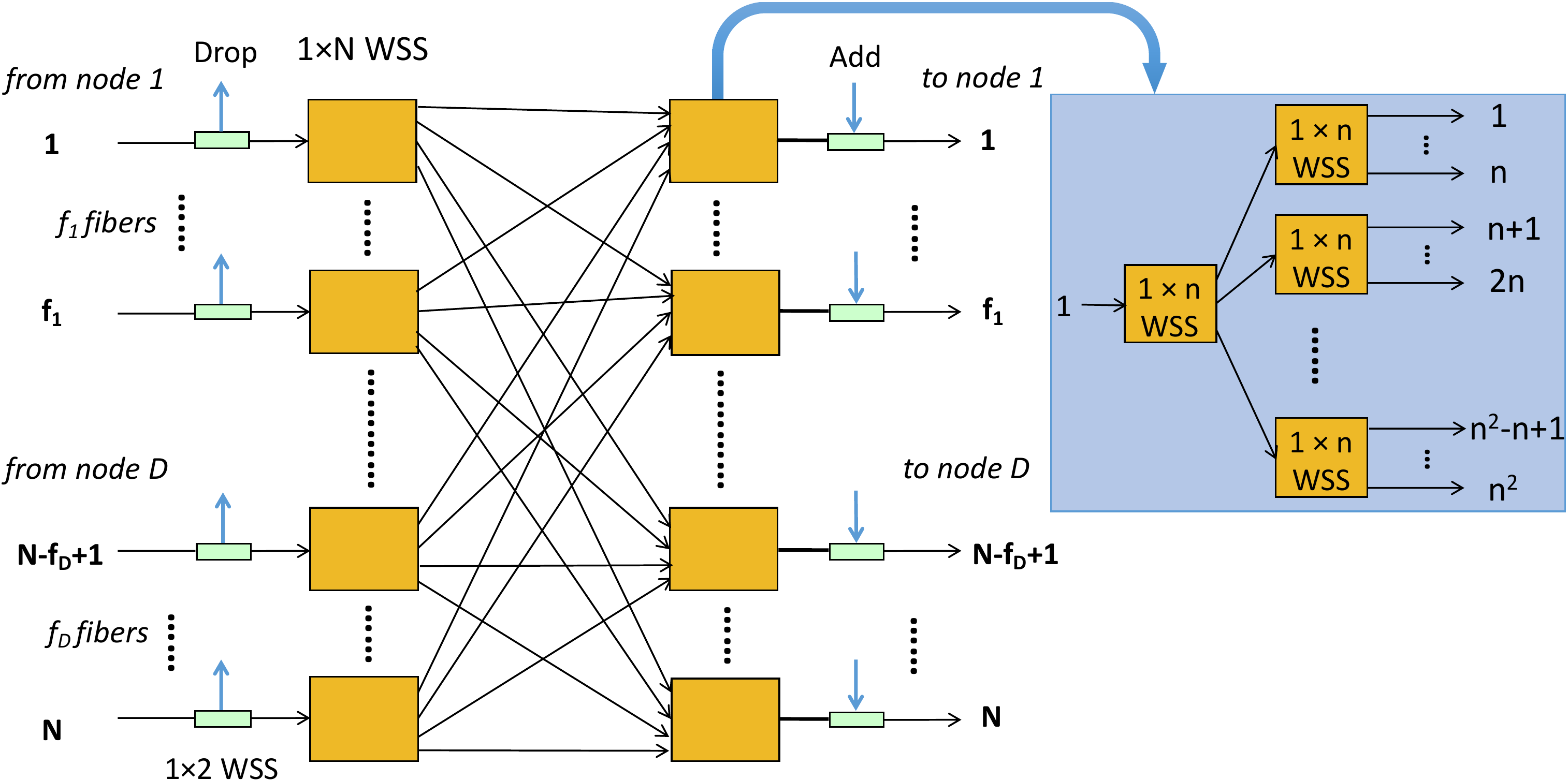}
	\centering
	\caption{\label{fig:CONV_Arch} An $N\times N$ WSS-based conventional OXC node architecture (``FLEX'').}
\end{figure}

\indent The conventional OXC node architecture (``FLEX'') implemented by cascading WSSs is shown in Figure \ref{fig:CONV_Arch}. Using this architecture, $2N$ of $1\times N$ WSSs are required, and the $1\times 2$ WSSs are used to add or drop local traffic. Suppose $1\times 4$ WSSs are utilized to build ``FLEX'' nodes. For a $4\times 4$ OXC node with $D=4$ and a single fiber per link, $8$ $1\times 4$ WSSs are needed, whereas for a $16\times 16$ OXC node with $D=4$ and $4$ parallel fibers per link, $32$ $1\times 16$ WSSs are needed. Each $1\times 16$ WSS could be realized by cascading $5$ $1\times 4$ WSSs. Thus, in total $160$ $1\times 4$ WSSs are needed for the $16\times 16$ OXC node. We can see that the required number of WSSs to realize an OXC node suffers from a square-order increment. The number of cascaded $1\times 4$ WSSs required by a large $1\times N$ WSS can be approximately calculated by
\begin{equation}
	\label{equ41}
	S(N)\approx \frac{N}{4}+\frac{N}{4^2}+\ldots+1
	=\frac{N}{3}\left(1-(1/4)^{\log_{4}N}\right).
\end{equation}

\noindent
In this architecture there is no switching constraint, i.e., {\em any} incoming request can be switched to {\em any} output fiber at any time, if there is a common available wavelength.

\begin{figure}
	\includegraphics[scale=0.51]{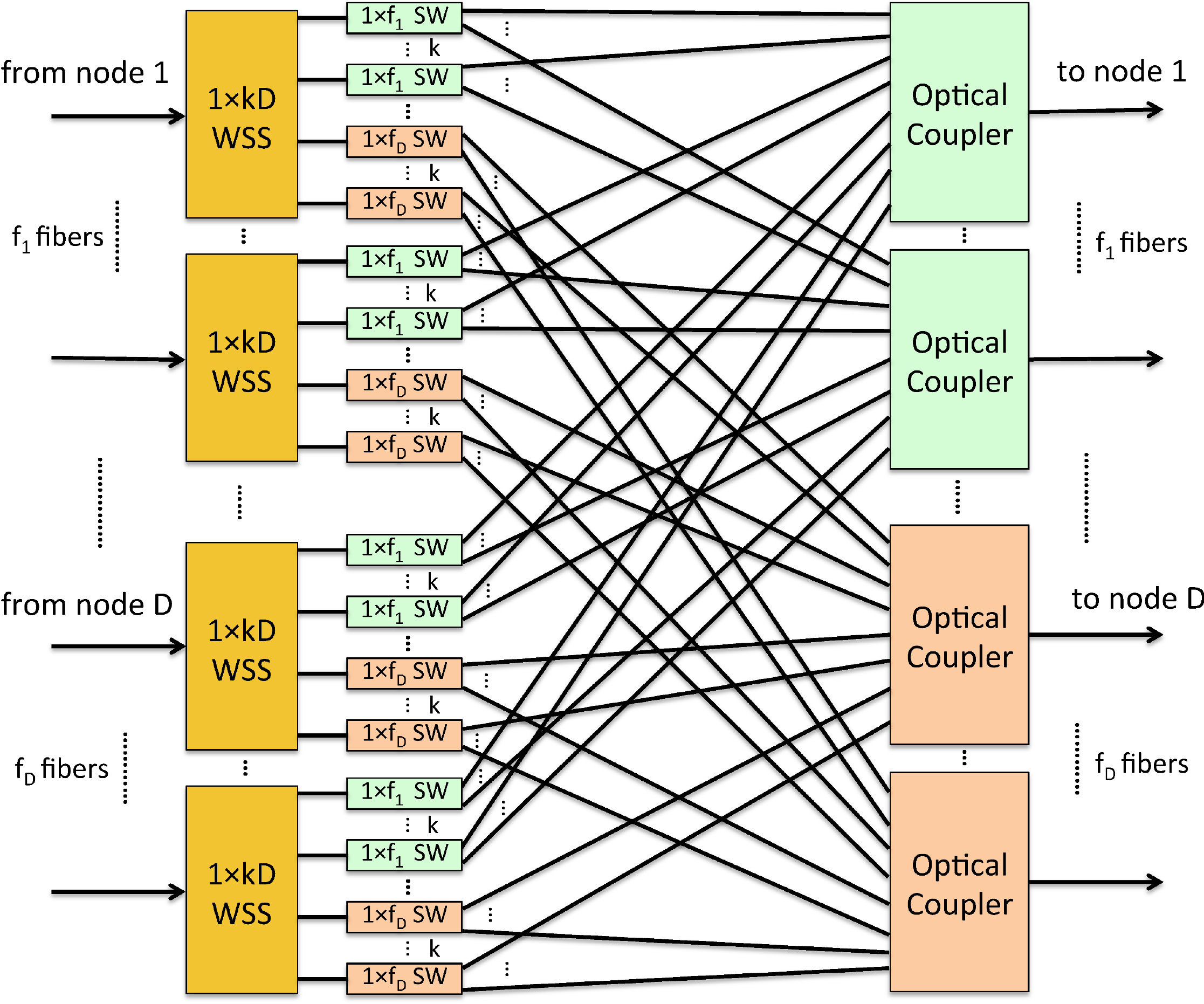}
	\centering
	\caption{\label{fig:EACH_Arch} Node architecture based on wavelength grouping and fiber selection (``HIER'').}
\end{figure}

\indent Figure \ref{fig:EACH_Arch} shows another implementation of an OXC node via wavelength grouping and fiber selection~\cite{HIER}. For each input fiber, the number of selectable fibers on each output link is $k$, where $k$ is a parameter of the OXC. The incoming traffic request is switched to one of $k$ parallel selectable fibers on the desired output link through the $1\times kD$ WSS. The input WSS partitions the incoming traffic requests into $kD$ groups, $k$ for each of the $D$ output links. Each group $j \ (j=1,\ldots, kD)$ is then switched to one of the $f_i$ parallel fibers on the output link $i$ via a $1\times f_i$ switch (SW). Even though the traffic requests can be switched to any of the $f_i$ parallel output fibers, the set of requests to an adjacent node $i$ can only be assigned to up to $k$ different fibers at the same time. $N$ $1\times kD$ WSSs, $k{\cdot}D{\cdot}N$ SWs, and $N$ optical couplers are required to implement this architecture, which is referred to as ``HIER''. For example, $k=1$ implies that all the optical paths intended for the same adjacent node from an input fiber are routed to a single fiber on the link at a given time (the single fiber may be different at different times). 

\section{Analytical Models}\label{sec.label43}

\indent In this section, we propose the resource allocation (wavelength assignment and fiber selection) algorithms and analytical models for evaluating the blocking performance of the OXC node architectures. To simplify the analysis, we make the following assumptions:
\begin{enumerate}
	\renewcommand{\labelenumi}{(\theenumi)}
	\item There are $W$ wavelengths per fiber. Each input/output link has $F$ fibers; however, this can be easily generalized to varying numbers of fibers per link. There are $N$ fibers in total, i.e., $N = D{\cdot}F$. Let $C=W{\cdot}F$ be the total number of {\em channels} on each link.
	\item Wavelength conversion is not used.
	\item A static traffic demand model in which the total number of requests at the input fibers is binomially distributed with $p$ being the probability that there is a request for an input channel. Thus, the mean number of requests from all input fibers is $NWp$. The output link of each request is chosen uniformly randomly from the $D$ output links.
	\item A wavelength is \textit{free} on a link if it is idle on at least one fiber of the link. If there is no \textit{free} wavelength along the route from input to output link, the request is blocked.
\end{enumerate}
\subsection{Conventional OXC Architecture}
\indent We first show that the only reason for blocking in this case is the lack of available resource on the requested output link, i.e., a request can be assigned a fiber and wavelength (note that the same wavelength must be assigned on the input fiber as well) as long as there is at least one idle channel on that output link.

\indent The basic idea is to show that we can assign at most $C$ requests to any output link, while making assigned requests to any output fiber to be no larger than $W$.

\begin{theorem}
	Given a set of requests on input fibers in ``FLEX'', the only reason for blocking is the lack of resource in the requested output links.
\end{theorem}

\begin{proof}
	Define $Q_{f,d}$ as the number of requests destined to output link $d$ from input fiber $f$, and suppose that $\sum_{f=1}^N Q_{f,d} \leq C, \: \forall d$. We will show that all requests can be assigned a fiber and wavelength through the following algorithm. First, we select a fiber on the desired output link for each traffic request. We consider the request assignments from input fiber $1$ to input fiber $N$ in order. For requests from input fiber $f, \: f \in [1,N]$, we do assignments in the order of destined output links. Let $\widetilde{Q}_{f,d}$ denote the remaining number of requests that need to be assigned from input fiber $f$ to output link $d$. Let $r_{\varepsilon_d}$ denote the number of requests that have been already assigned to output fiber $\varepsilon_d$ of link $d$. $\varepsilon_d$ is set to 1 initially. If the remaining capacity of fiber $\varepsilon_d$ is not less than $\widetilde{Q}_{f,d}$, we assign $\widetilde{Q}_{f,d}$ requests to $\varepsilon_d$. Otherwise, we partition $\widetilde{Q}_{f,d}$ into two parts to guarantee that the number of requests assigned to any output fiber is no larger than $W$. Also, since the number of requests from any input fiber is no larger than $W$, the requests from any input fiber to any output link is assigned to at most two different output fibers. The number of requests of the first part equals $W-r_{\varepsilon_d}$ and we assign them to $\varepsilon_d$. The remaining requests are assigned to fiber $\varepsilon_d+1$, which does not have any assigned request yet. This guarantees that at most one output fiber of any output link is partially occupied (having larger than $0$ and less than $W$ assigned requests). Since there is no routing constraint within the OXC, we can assign $C$ requests to any output link. The pseudocode is shown in Algorithm \ref{alg41}. Based on the algorithm, we can give three conclusions and summarize them in Lemma \ref{lemma41}.
	
	\begin{algorithm}    
		\caption{Requests assignment algorithm for ``FLEX''}
		\label{alg41}
		\begin{algorithmic}[1]               
			\FOR{$f=1:N$}
			\FOR{$d=1:D$}
			\IF{$\widetilde{Q}_{f,d}>0$ and $W-r_{\varepsilon_d}\geq \widetilde{Q}_{f,d}$}
			\STATE{assign $\widetilde{Q}_{f,d}$ requests to output fiber $\varepsilon_d$}
			\STATE{$r_{\varepsilon_d}$ += $\widetilde{Q}_{f,d}, \widetilde{Q}_{f,d} = 0$}
			\ELSIF{$\widetilde{Q}_{f,d}>0$ and $W-r_{\varepsilon_d}<\widetilde{Q}_{f,d}$}
			\STATE{assign $W-r_{\varepsilon_d}$ requests to output fiber $\varepsilon_d$}
			\STATE{$\widetilde{Q}_{f,d}=\widetilde{Q}_{f,d}-(W-r_{\varepsilon_d}$)}
			\STATE{$r_{\varepsilon_d}=W$, $\varepsilon_d$ += $1$}
			\STATE{assign $\widetilde{Q}_{f,d}$ requests to output fiber $\varepsilon_d$}
			\STATE{$r_{\varepsilon_d}$ += $\widetilde{Q}_{f,d}, \widetilde{Q}_{f,d} = 0$}
			\ENDIF
			\ENDFOR
			\ENDFOR
			\STATE{use edge coloring algorithm for wavelength assignment}
		\end{algorithmic}
	\end{algorithm}

	\begin{lemma}\label{lemma41} The requests assignment algorithm guarantees that
		\begin{enumerate}
			\item The number of requests assigned to any output fiber is no larger than $W$;
			\item At most one output fiber of any output link is partially occupied; and
			\item Requests from any input fiber to any output link are assigned to at most two different output fibers.
		\end{enumerate}
	\end{lemma}

	Now, wavelength assignment can be done simply as follows. We generate a bipartite multigraph by denoting input fibers as left-hand vertices $r_a$ $(a=1,\cdots,N)$ and output fibers as right-hand vertices $c_b$ $(b=1,\cdots,N)$; there are $q_{ab}$ parallel edges between $r_a$ and $c_b$, where $q_{ab}$ is the number of requests from input fiber $a$ to output fiber $b$ as assigned by the algorithm above. From the given constraint, we know that the maximum degree of this bipartite multigraph is $W$, which is equal to the fiber capacity. The wavelength assignment problem is thus the same as coloring the edges of the bipartite multigraph, which can be done in polynomial time \cite{cole2001edge}. Thus, blocking occurs only if the number of traffic requests to an output link exceeds the link's capacity.
\end{proof}

\begin{table}
	\caption{Traffic demands for the small example node.}
	\label{table41}
	\centering
	\small
	\begin{tabular}{|c|c|c|c|c|c|c|}
		\hline
		\backslashbox{Output link}{Input fiber} & 1 & 2 & 3 & 4 & 5 & 6 \\ \hline
		1 & 3 & 2 & 1 & 3 & 0 & 3   \\ \hline
		2 & 1 & 2 & 3 & 1 & 4 & 1   \\ \hline
	\end{tabular}
	\centering
	\vspace{10pt}
\end{table}

\indent We illustrate the above algorithm with a small example. Suppose we have a small node with $D=2$, $F=3$, and $W=4$. Let $r_{i,d}$ denote the set of requests from input fiber $i$ to output link $d$. Let $o^d_j$ denote the output fiber $j$ in link $d$. The traffic demands ($|r_{i,d}|$) for this node are listed in Table \ref{table41}. Taking output link $1$ as an example, using Algorithm \ref{alg1}, we have $r_{1,1}$ and one of $r_{2,1}$ assigned to $o^1_1$; the remaining request of $r_{2,1}$, $r_{3,1}$, and two requests of $r_{4,1}$ assigned to $o^1_2$; the remaining request of $r_{4,1}$, and $r_{6,1}$ assigned to $o^1_3$. Thus, all the requests destined to output link $1$ can be accommadated.

\indent We now present the analytical model for computing the blocking probability.

\indent Let $R_t$ denote the total number of requests from all input fibers. Then
\begin{equation}\label{equ42}
	P(R_t=g)=\binom{NW}{g}p^{g}(1-p)^{NW-g}, \quad g = 0, 1, \ldots, NW.
\end{equation}

\indent Let $R_j$ denote the number of requests destined to output link $j$ from all input fibers. Then we have
\begin{equation}\label{equ43}
		P(R_j=n|R_t=g)=\binom{g}{n} \left(\frac{1}{D}\right)^n\left(1-\frac{1}{D}\right)^{g-n}, \quad n = 0, 1, \ldots, g.
\end{equation}

\indent Blocking occurs only when the number of requests destined to an output link exceeds the link's capacity. Thus, the blocking probability for output link $j$ is:
\begin{equation}\label{equ44}
	\begin{split}
		P(B_j)=\sum_{g=C+1}^{NW}\sum_{n=C+1}^{g}&\frac{n-C}{g}{\cdot} P(R_j=n,\: R_t=g).
	\end{split}
\end{equation}

\indent The total blocking probability for the node is
\begin{equation}\label{equ45}
	P(B) = D{\cdot} P(B_j).
\end{equation}

\subsection{Hierarchical OXC Architecture}
\indent We next analyze the hierarchical OXC architecture (``HIER'') separately for the following two cases: $k=1$, and $k \geq 2$.

\indent Case 1: $k=1$. In order to minimize blocking, the total number of requests routed to any fiber exceeding its capacity must be minimized. In other words, we have to minimize
\begin{equation}\label{equ46}
	\sum_{f'=1}^N\max\{0,\sum_{f=1}^N{q_{f,f'}}-W\},
\end{equation}

\noindent
where $q_{f,f'}$ is the number of requests assigned from input fiber $f$ to output fiber $f'$. We will present a heuristic later to minimize the above expression.

\indent In order to facilitate the analysis, we assume that the output fiber for all requests from an input fiber to a particular output link is chosen uniformly randomly from the $F$ possible fibers. (Note that since $k  = 1$, all requests from an input fiber to a given output link must be switched to a single fiber on that output link.) We call this assignment method as HIER Random Fiber Selection Algorithm (\emph{HRFS}).

\indent Let $X_{f}^j$ denote the number of input fibers which choose the output fiber $f$ for output link $j$. We have
\begin{equation}\label{equ47}
		P(X_{f}^j=m)=\binom{N}{m}\left(\frac{1}{F}\right)^{m}\left(1-\frac{1}{F}\right)^{N-m}, \quad m = 0, 1, \ldots N.
\end{equation}

\indent Let $R_x$ denote the number of requests from the $X_{f}^j$ input fibers. Then we have
\begin{equation}\label{equ48}
		P(R_x=n|X_{f}^j=m)=\binom{Wm}{n}p^n(1-p)^{Wm-n}, \quad n = 0, 1, \ldots Wm.
\end{equation}

\indent Let $R_{f}^j$ denote the number of requests assigned to fiber $f$ of output link $j$.
\begin{equation}\label{equ49}
	\begin{split}
		P(R_{f}^j=e)=&\sum_{m=0}^{N}\sum_{n=e}^{Wm}P(R_{f}^j=e|R_x=n,X_{f}^j=m){\cdot} P(R_x=n|X_{f}^j=m){\cdot} P(X_{f}^j=m)\\
		=&\sum_{m=0}^{N}\sum_{n=e}^{Wm}\binom{n}{e}\left(\frac{1}{D}\right)^{e}\left(1-\frac{1}{D}\right)^{n-e}{\cdot} P(R_x=n|X_{f}^j=m){\cdot}P(X_{f}^j=m).
	\end{split}
\end{equation}

\indent The blocking probability for one output fiber is:
\begin{equation}\label{equ410}
	\begin{split}
		P(B_f)=\sum_{g=W+1}^{NW}\sum_{e=W+1}^{g}P(B_f|R_{f}^j=e,R_t=g){\cdot} P(R_{f}^j=e)P(R_t=g),
	\end{split}
\end{equation}

\noindent
where
\begin{equation}\label{411}
	P(B_f|R_{f}^j=e,R_t=g)=\frac{e-W}{g}.
\end{equation}

\indent Recall that $R_t$ was defined before \ref{equ42}. Then the total blocking probability is
\begin{equation}\label{412}
	B=N{\cdot} P(B_f).
\end{equation}

\begin{algorithm}
	\caption{\emph{HSA} algorithm for ``HIER'' with $k=1$}
	\label{alg42}
	\begin{algorithmic}[2]               
		\FOR{$d=1:D$}
		\STATE{group all requests to link $d$ based on the input fiber indexes}
		\STATE{sort all groups in non-increasing order of the number requests and the order is indexed by $s$}
		\WHILE{not all groups have been assigned}
		\IF{output fiber $\varepsilon_d$ has the largest remaining capacity}
		\IF{$W-r_{\varepsilon_d}$$\geq$$|g_s|$}
		\STATE{assign $g_s$ to output fiber $\varepsilon_d$}
		\STATE{$r_{\varepsilon_d}$ += $|g_s|$}
		\ELSE
		\STATE{assign $W-r_{\varepsilon_d}$ requests chosen from $g_s$ to output fiber $\varepsilon_d$}
		\STATE{the remaining $|g_s|-(W-r_{\varepsilon_d})$ requests are blocked}
		\STATE{$r_{\varepsilon_d}=W$}
		\ENDIF
		\ENDIF
		\STATE{$|g_s|=0$, $s=s+1$}
		\ENDWHILE
		\ENDFOR
		\STATE{use edge coloring algorithm for wavelength assignment}
	\end{algorithmic}
\end{algorithm}

\indent The \emph{HRFS} assignment algorithm lends itself to analysis, but performs poorly (as expected and as shown later) because it assigns fibers randomly and independently of other requests. We therefore propose the HIER Sort Assignment Algorithm (\emph{HSA}) to solve the problem. We consider output links one by one. For each output link, we first group all requests based on the input fiber index. Then, we sort the groups by the number of requests in the group in non-increasing order. Based on the sorted order, we assign each group to the output fiber with the largest remaining capacity, while making the number of requests assigned to any output fiber to be no larger than $W$. Finally, we use edge coloring algorithm for wavelength assignment. The pseudocode is shown in Algorithm \ref{alg42}. Here, $g_s$ represents the $s$th largest group of requests.

\indent As an example, if \emph{HRFS} is used to assign requests destined to output link $1$ in Table \ref{table41}, one possible assignment can be as follows. Input fiber $1$ and $6$ select $o^1_1$, input fiber $2$ and $4$ select $o^1_2$, and input fiber $3$ selects $o^1_3$. In this case, $3$ requests are blocked. If \emph{HSA} is used instead, input fiber $1$ and $2$ select $o^1_1$, input fiber $3$ and $4$ select $o^1_2$, and input fiber $6$ selects $o^1_3$. Thus, the number of blocked requests can be reduced to just one in this example.

\indent Case 2: $k\geq 2$. We show that in ``HIER'' with $k=2$, we can achieve the same blocking performance as that of ``FLEX'' by using Algorithm \ref{alg41}.

\begin{theorem}
	Given the requests on input fibers in ``HIER'' with $k=2$, the only reason for blocking is the lack of resource in the desired output links.
\end{theorem}
\begin{proof}
	The only difference between ``FLEX'' and ``HIER'' with $k=2$ is that in ``FLEX'', we can assign requests from a fiber to any of the $F$ fibers on a link, while in ``HIER'' with $k=2$, we can assign requests from a fiber to at most $2$ fibers on a link. However, from Lemma \ref{lemma41}, we can see that Algorithm \ref{alg41} assigns requests from any input fiber to a given output link to at most two different output fibers. Thus, there is no difference in blocking between ``FLEX'' and ``HIER'' with $k=2$.
\end{proof}

\section{Simulation Results}\label{sec.label44}
\indent We first validate our analytical models by comparing analytical results with those from simulations. Connection requests for each input fiber are generated according to a binomial distribution with parameters $W$ and $p$, and the requested output link is chosen randomly from $1$ to $D$. For each data point of simulation results in the graphs, we simulated $1000$ instances and obtained the average results.

\indent We present results for the OXC nodes with node degree $D=4$ and the number of wavelengths per fiber $W = 32$. Note that our analytical model for the two architectures can be easily extended to the case of different numbers of fibers on different links. For simplicity, in the evaluation we choose the number of fibers to be $10$ for all links (i.e., a $40 \times 40$ OXC node).

\begin{figure}
	\begin{tabular}{cc}
		\begin{minipage}[t]{2.85in}
			\centering
			\includegraphics[scale=0.36]{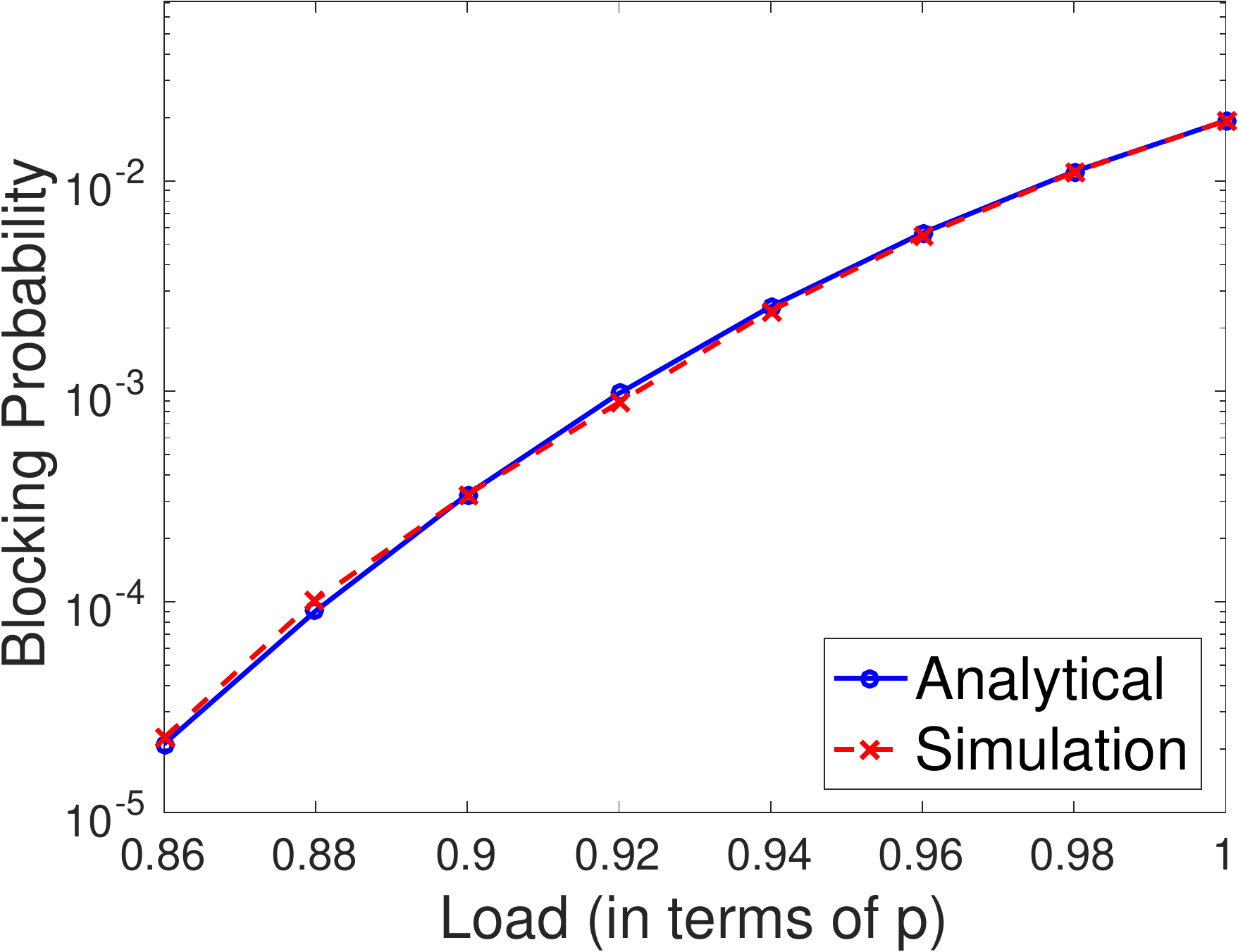}
			\caption{The validation results for ``FLEX''.}
			\label{fig43}
		\end{minipage}
		\hspace{0.1in}
		\begin{minipage}[t]{2.85in}
			\centering
			\includegraphics[scale=0.36]{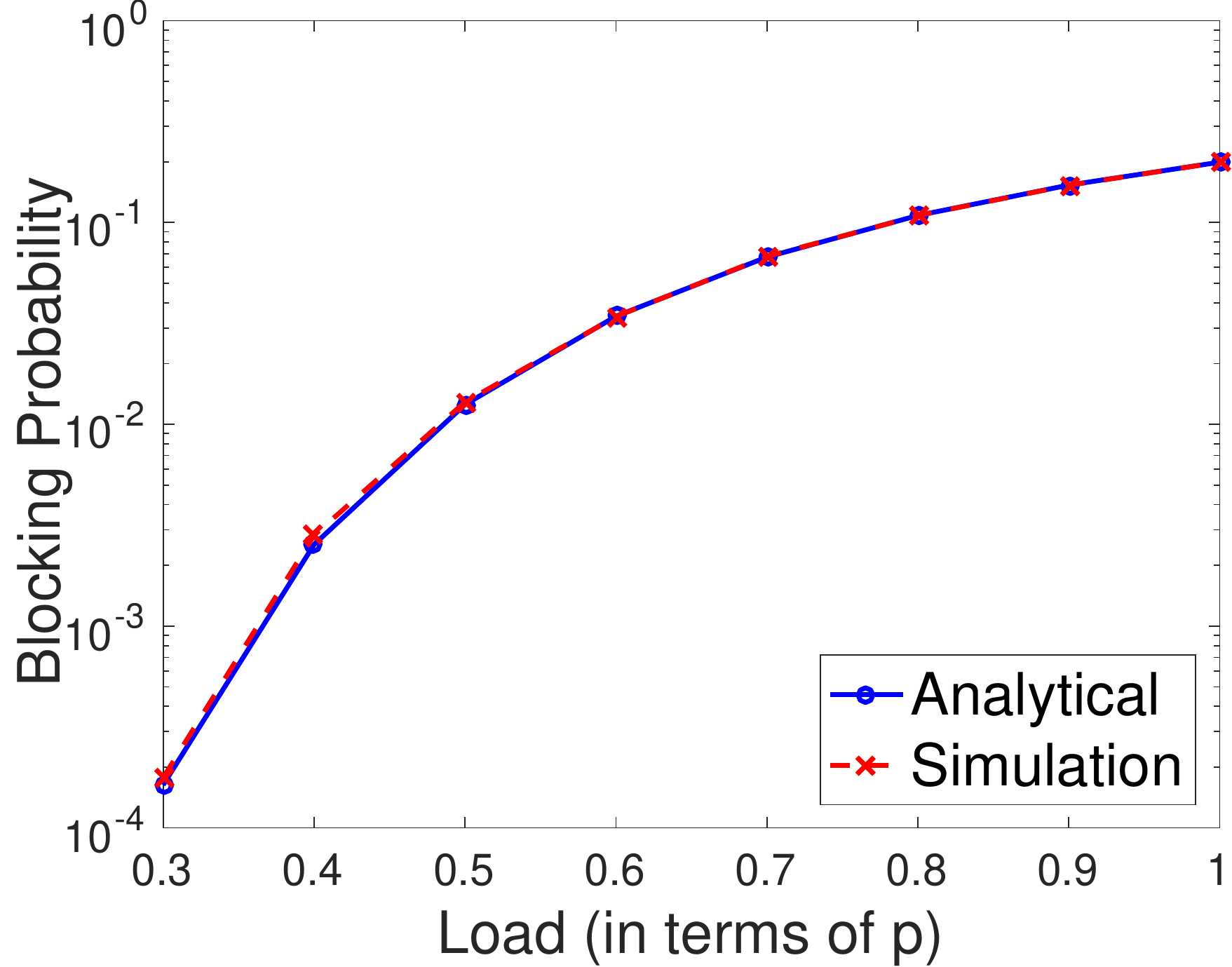}
			\caption{The validation results for ``HIER'' ($k=1$).}
			\label{fig44}
		\end{minipage}
	\end{tabular}
\end{figure}

\indent Figures \ref{fig43} and \ref{fig44} show the comparison of analytical and simulation results for ``FLEX'' and ``HIER'' ($k=1$) using \emph{HRFS}, respectively. Analytical estimations match simulation results very well for both architectures. As shown in Figure \ref{fig45}, the blocking probability of ``HIER'' ($k=1$) using \emph{HRFS} is much higher than that of ``FLEX''. This is due to the routing constraint caused by $k=1$ and the random fiber selection. Figure \ref{fig45} also shows the simulation results when we use \emph{HSA} for request assignments. Obviously, we can see that this algorithm gives much better blocking performance. Also, we can see that the blocking probabilities of ``HIER'' ($k=1$) are only a little bit higher than those of ``FLEX'' when using \emph{HSA}.

\begin{figure}
	\includegraphics[scale=0.5]{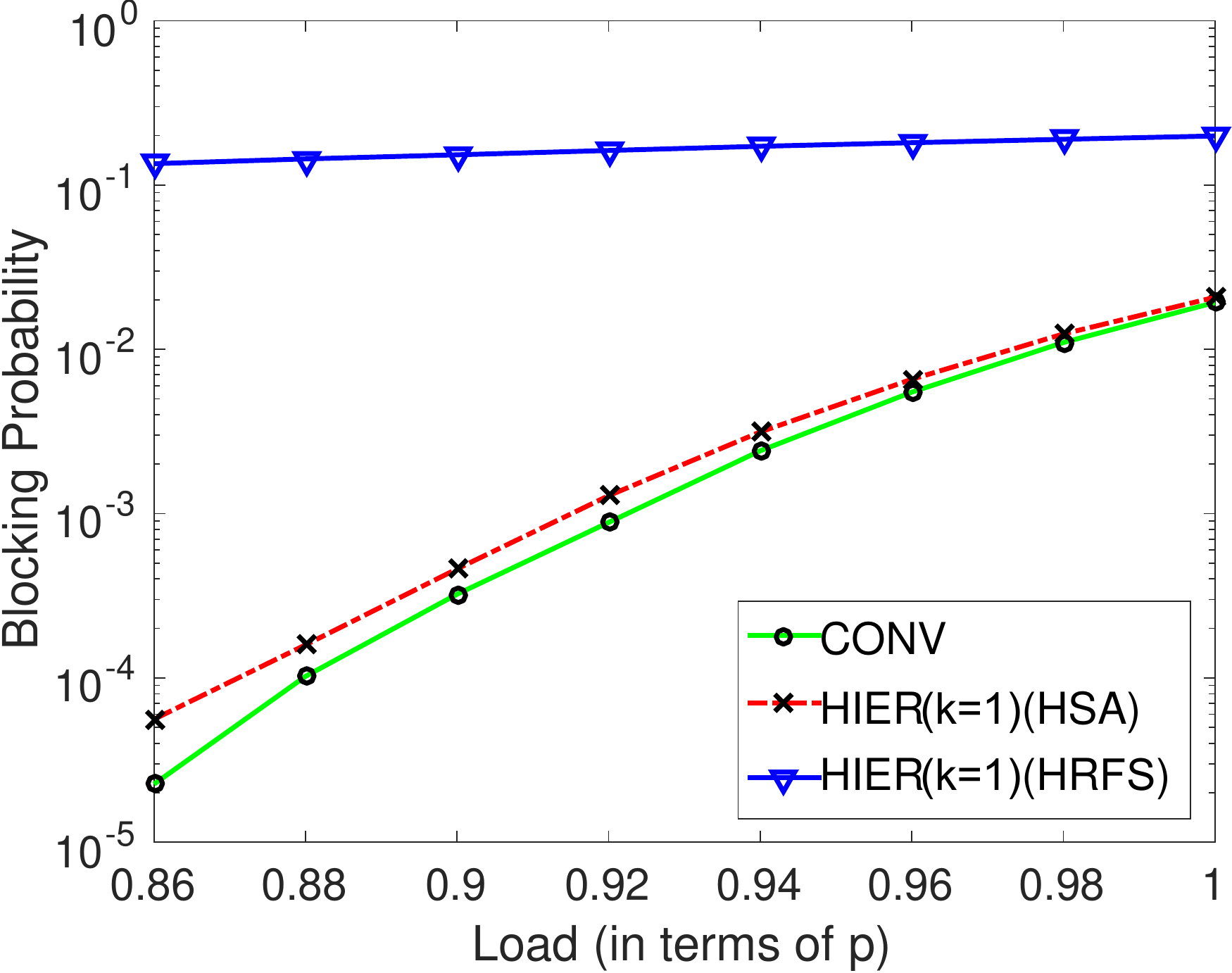}
	\centering
	\caption{\label{fig45} Comparison between ``FLEX'' and ``HIER'' ($k=1$).}
\end{figure}

\indent We next examine the cost and power consumption of the two OXC architectures.
The complexity of each architecture is as follows. Since WSS is the most costly element in the architecture, we count the number of WSSs to denote the complexity. For ``FLEX'', if $1\times 4$ WSS is utilized, the total number of WSSs required for a $40\times 40$ node is $1120$, whereas for ``HIER'' ($k=2$), $N$ $1\times kD$ WSSs are required at the node, which corresponds to $120$ $1\times 4$ WSSs. Thus, ``HIER'' requires much less hardware than ``FLEX'' without hurting the blocking performance.

\indent We compare ``FLEX'' and ``HIER'' in terms of power consumption and capital expenditure (CapEx). ``FLEX'' consists of $2N$ $1\times N$ WSSs. ``HIER'' consists of $N$ $1\times kD$ WSSs, $k{\cdot}D{\cdot}N$ $1\times F$ optical switches based on MEMS, and $N$ $N\times 1$ optical couplers. The total number of WSS ports in a $1\times N$ WSS can be denoted as $Z(N)=4{\cdot} S(N)$ (see equation \ref{equ41}). The power consumption and CapEx are calculated by summing up the consumed power and dollar cost of each component. A summary is shown in Table \ref{table42} \cite{xu2015podca}\cite{singla2010proteus}.

\renewcommand\arraystretch{1.2}
\begin{table}
	\caption{Power consumption and cost of the optical components in ``FLEX'' and ``HIER'' architectures.}
	\label{table42}
	\centering
	\small
	\begin{tabular}{|c|c|c|c|}
		\hline
		Component & WSS (port) & MEMS switch (port) & Coupler \\ \hline
		Power(Watts) & 1 & 0.25 & 0  \\ \hline
		Cost(Dollars) & 1000 & 255 & 195  \\ \hline
		CONV & $2N{\cdot} Z(N)$ & $0$ & $0$    \\ \hline
		HIER & $N{\cdot} Z(kD)$ & $k{\cdot} N{\cdot} D{\cdot} F$ & $N$ \\ \hline
	\end{tabular}
	\centering
	\vspace{10pt}
\end{table}

\indent The difference in power consumption between ``FLEX'' and ``HIER'' is
\begin{equation}\label{413}
	2N{\cdot} Z(N)-N{\cdot} Z(kD) - 0.25{\cdot}kN^2.
\end{equation}
The CapEx difference between ``FLEX'' and ``HIER'' is
\begin{equation}\label{414}
	2000N{\cdot} Z(N)-1000N{\cdot} Z(kD)-255kN^2-195N.
\end{equation}
The power consumption and CapEx difference between ``FLEX'' and ``HIER'' ($k=2$) for an OXC node with $D=4$ and varying $F$ are shown in Figures \ref{fig46} and \ref{fig47}.
\begin{figure}
	\begin{tabular}{cc}
		\begin{minipage}[t]{2.75in}
			\centering
			\includegraphics[scale=0.42]{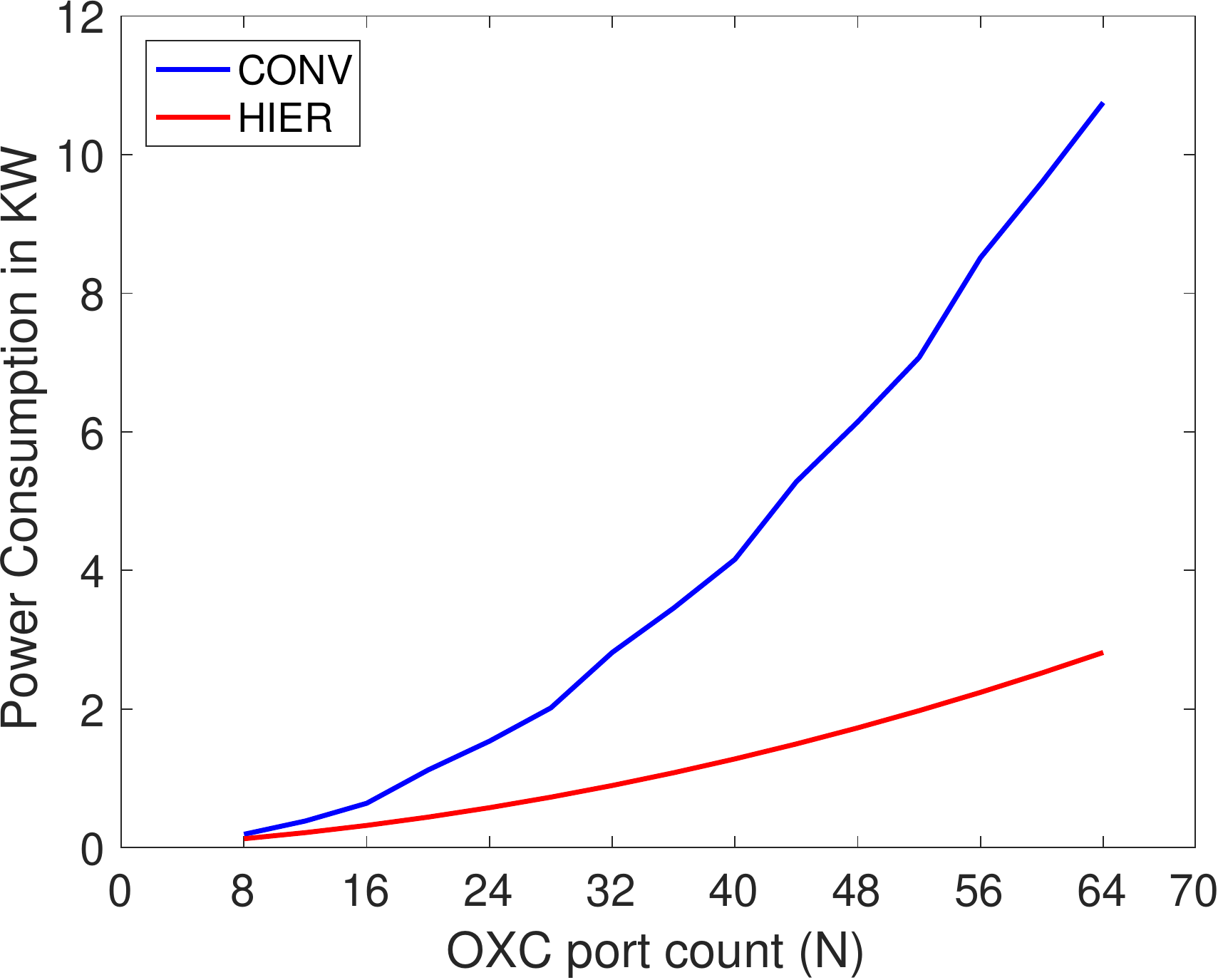}
			\caption{Power consumption comparison.}
			\label{fig46}
		\end{minipage}
		\hspace{0.1in}
		\begin{minipage}[t]{2.75in}
			\centering
			\includegraphics[scale=0.42]{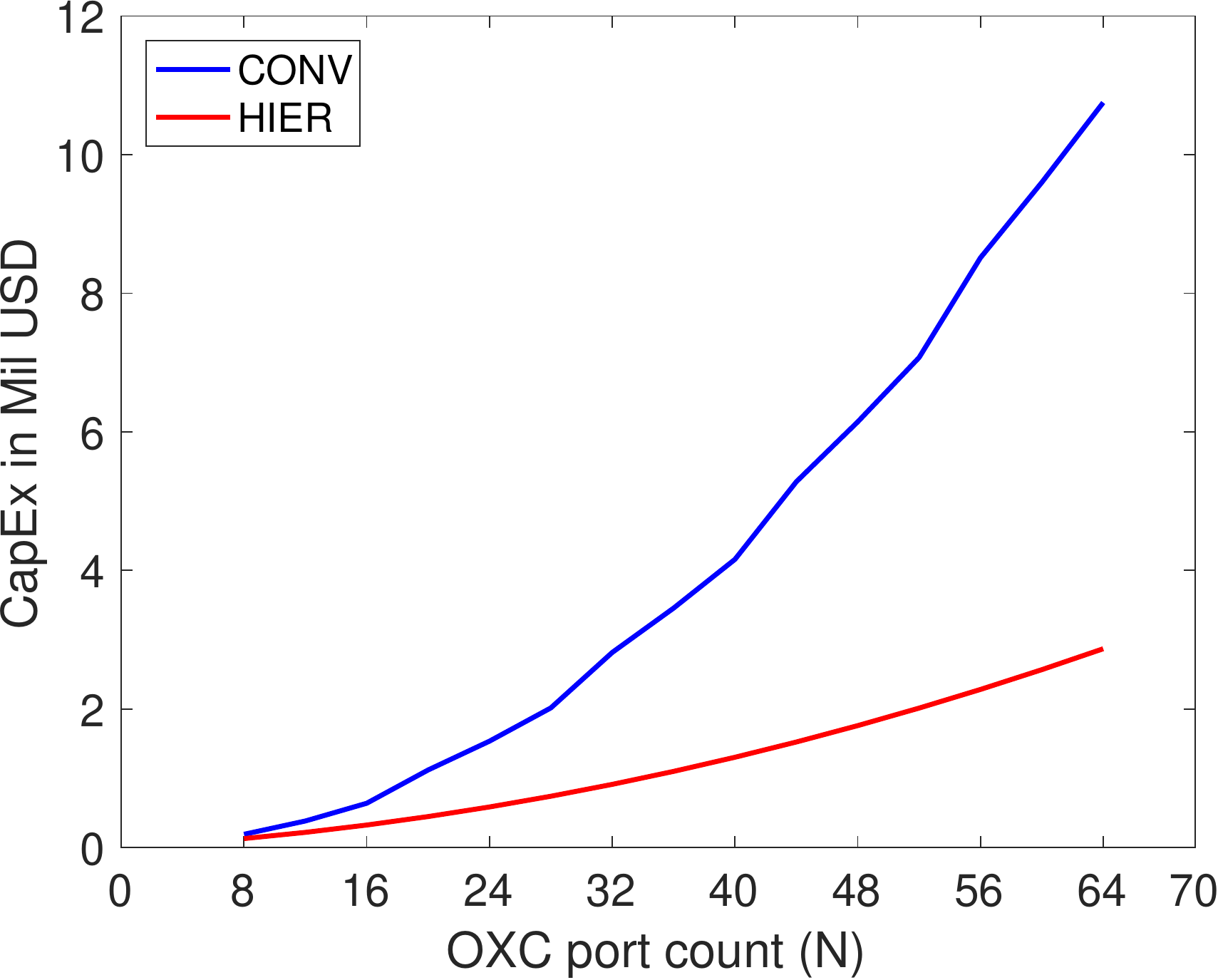}
			\caption{CapEx comparison.}
			\label{fig47}
		\end{minipage}
	\end{tabular}
\end{figure}

\section{Conclusion}\label{sec.label45} 
\indent In this chapter, we investigate the resource assignment problem and analytical models for predicting the blocking performance of multi-fiber OXC nodes using the conventional architecture and a cost-efficient hierarchical architecture \cite{jingxin2016sarnoff}. The numerical results validate the accuracy of our proposed analytical models and show that the hierarchical architecture can achieve similar blocking performance as the conventional architecture while using much less hardware. A comparison of the power consumption and cost of the two architectures also reveals the advantages of hierarchical architecture. Our results show that the hierarchical architecture which utilizes wavebanding features exhibits a good balance between performance, cost, and power consumption.

\chapter{Routing, Fiber, Band, and Spectrum Assignment (RFBSA) for Multi-granular Elastic Optical Networks}
\label{chap_5}

\indent In this chapter, we consider the wavebanding feature in EONs. A flexible waveband (``FLEX'') multi-granular architecture to increase waveband path utilization was presented in \cite{hasegawa2015flexible}. This architecture enables non-uniform and non-contiguous flex-grid wavebands \cite{wu2017routing}. We solve the routing, fiber, waveband, and spectrum assignment (RFBSA) problem introduced by elastic optical networking and flexible wavebanding.  

\section{Background and Problem Statement}\label{sec.label52}
\indent We first present a comparison of the conventional optical cross-connect (OXC) architecture and the flexible waveband cross-connect architecture for EONs.

\indent Consider an OXC node with a physical node degree of $D$, which is the number of physical nodes connected to this node. Each connectivity is represented by an input or output link. Each physical link contains several fibers. $N=\sum_{i=1}^{D}x_i$ denotes the total number of input/output fibers to/from the node, where $x_i$ denotes the number of parallel fibers on link $i$.

\subsection{``FLEX'' Architecture}
\indent The details of this architecture has been described in Section \ref{sec.label42}. In EONs, the WSSs are flex-grid in order to switch frequency slots. Current commercially available flex-grid WSSs typically have a port count limit of 4, 9, or 20, which are not easily scalable.

\subsection{``FLEX'' Waveband Architecture}

\begin{figure}
	\centering
	\includegraphics[scale=0.5]{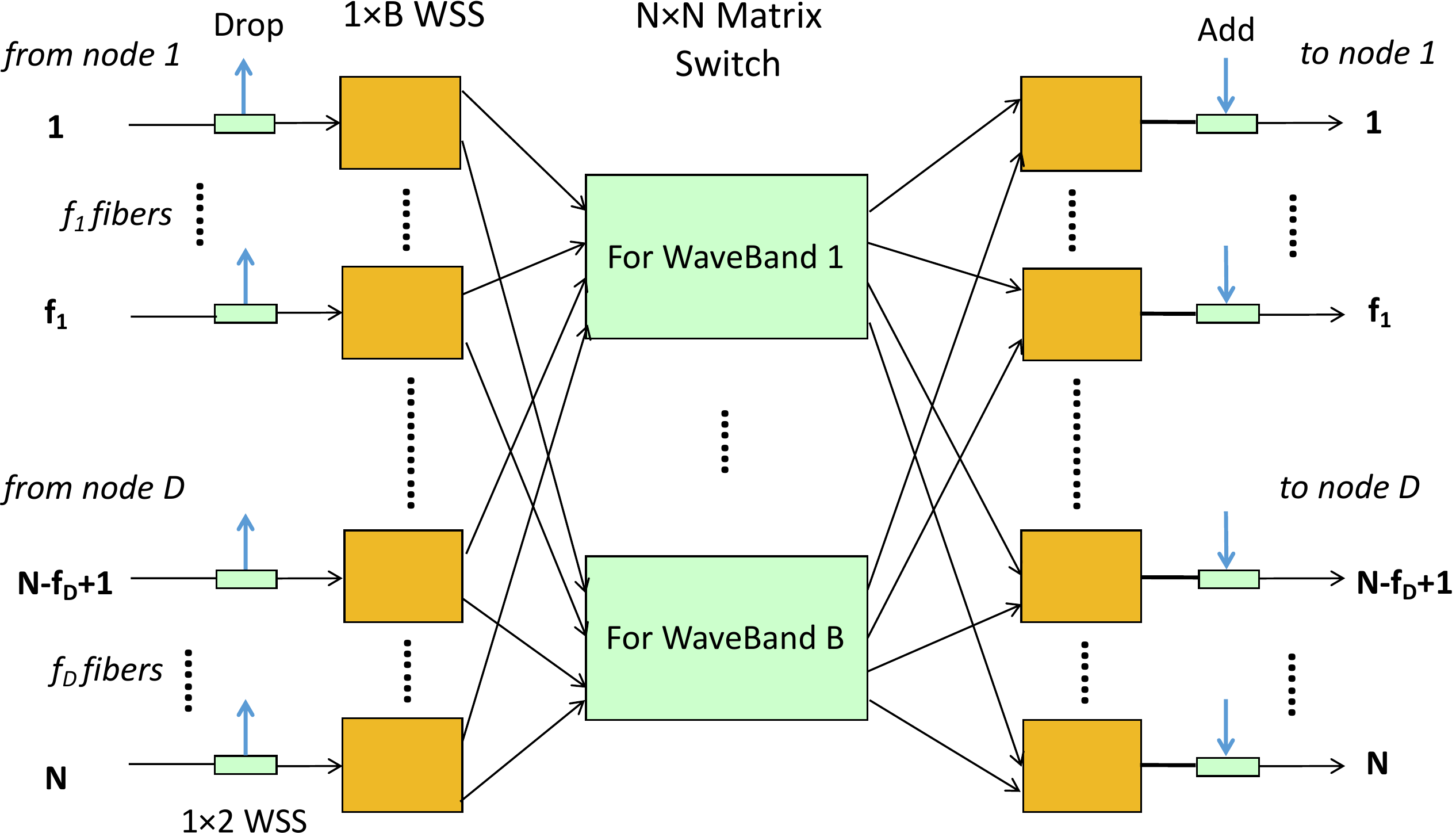}
	\caption{Flexible waveband node architecture.}
	\label{fig:FLEX_Arch}
\end{figure}

\indent Figure \ref{fig:FLEX_Arch} shows the flexible waveband OXC architecture (``FLEX'') proposed in~\cite{hasegawa2015flexible}. ``FLEX'' is composed of small-port-count $1{\times}B$ flex-grid WSSs and $B$ cost-effective matrix switches. Requests from an incoming fiber can be partitioned into $B$ groups and each group is switched as a whole to one of the output fibers of the OXC node. Note that the flex-grid WSSs provide the capability of {\em independently} switching a variable-sized set of contiguous FSs (corresponding to one lightpath) to {\em any} of its output fibers. Therefore, even though the FSs that are allocated to each lightpath must be contiguous to satisfy the spectrum contiguity constraint, the sets of lightpaths (different sets of contiguous FSs) that are switched as a band need not occupy a contiguous spectral range in ``FLEX''. This is in contrast to \cite{patel2012hierarchical} which requires the wavebands to be of uniform size and each waveband occupies a contiguous spectral range.

\indent Since $B$ is quite small compared to $N$, the hardware cost can be saved in terms of fewer number of costly WSSs in this architecture. However, the switching capability is reduced, since an entire group of lightpaths needs to be switched as a single entity, and the number of groups that can be switched simultaneously cannot exceed $B$.

\subsection{Comparison Between Architecture}
\indent We compare the two architectures in terms of power consumption and hardware cost. Suppose the port count of the node is $N$. A ``FLEX'' node requires $2N$ $1{\times}N$ WSSs, each of which is constructed with $S(N)$ $1{\times}4$ WSSs. $S(N)$ is presented in Equation \ref{equ41}. A ``FLEX'' node consists of $2N$ $1{\times}B$ WSSs, as well as $B$ $N{\times}N$ cost-effective matrix switches. Each $N{\times}N$ matrix switch can be constructed with $N$ $1{\times}N$ MEMS optical switches and $N$ $N{\times}1$ optical couplers. We calculate the power consumption and hardware cost of an OXC node by summing up the consumed power and dollar cost of each component. Table \ref{table51} gives a summary of the comparison.

\renewcommand\arraystretch{1.2}
\begin{table}
	\centering
	\small
	\caption{Power consumption and cost of the optical components in ``FLEX'' and ``FLEX'' architectures.}
	\label{table51}
	\begin{tabular}{|c|c|c|c|}
		\hline
		Component & WSS (port) & MEMS switch (port) & Coupler \\ \hline
		Power(Watts) \cite{xu2018podca} & 1 & 0.25 & 0  \\ \hline
		Cost(Dollars) \cite{xu2018podca} & 1000 & 255 & 195  \\ \hline
		CONV & $2N{\cdot}S(N){\cdot} 4$ & $0$ & $0$    \\ \hline
		FLEX & $2N{\cdot} B$ & $B{\cdot} N{\cdot} N$ & $B{\cdot} N{\cdot} N$ \\ \hline
	\end{tabular}
\vspace{10pt}
\end{table}

\indent Typically, $B$ equals 4. The power consumption of ``FLEX'' is higher than that of ``FLEX'' with a difference of
\begin{equation}\label{equ51}
8{\cdot}N{\cdot}S(N)-8{\cdot}N-N^2=O(N^2).
\end{equation}

\noindent
The difference in hardware costs between the two architectures is
\begin{equation}\label{equ52}
8000{\cdot}N{\cdot}S(N)-8000{\cdot}N-1800{\cdot}N^2=O(N^2).
\end{equation}

\noindent
Thus, the power consumption and cost of a ``FLEX'' node is $O(N^2)$ better than that of ``FLEX''. We can see that significant savings can be realized when utilizing the ``FLEX'' node architecture.

\subsection{Motivation and Problem Definition}

\indent Routing and spectrum assignment is known to be an NP-complete problem in EONs \cite{Lezama2014}. Adding the fiber and flexible waveband selection further increases the problem complexity in another two dimensions. Before we proceed to formally state the problem, we present an example in Figure \ref{fig52} to illustrate the challenge in solving the problem. In this example, there are 4 ``FLEX'' nodes, 4 parallel fibers per link, and an input fiber can be switched to at most 2 output fibers, i.e., $B = 2$. The green dashed lines represent already established connections, i.e., existing lightpaths are switched from the indicated input fibers to the corresponding output fibers. The matrix shows the already occupied FSs on the various fibers. We denote by $f_{v_s,v_d,i}$ the $i^{\rm th}$ fiber in the direction from $v_s$ to $v_d$. Our objective is to find a lightpath from node 1 to node 3 for a request whose bandwidth requirement is 2 FSs, and minimize the maximum spectral usage. Due to spectrum contiguity, spectrum continuity, and spectrum non-overlapping constraints, we cannot establish a lightpath by using any of (a) $f_{1,2,1}$ and $f_{2,3,1}$, (b) $f_{1,2,1}$ and $f_{2,4,1}$, or (c) $f_{1,2,2}$ and $f_{2,3,2}$ without increasing the maximum index of used FSs. Also, we cannot establish a new waveband connection from $f_{1,2,1}$ without exceeding the waveband constraint (since $B=2$). One possible solution is to establish a new waveband connection from $f_{1,2,2}$ to $f_{2,4,2}$ in node 2 and the lightpath passes through $f_{4,3,1}$. The lightpath uses FS 1 and 2, and this solution does not increase the spectral usage and keeps the maximum number of used FSs to be 4. The example illustrates the point that a judicious choice of fibers, bands, and FSs must be jointly made in order to optimize spectrum usage.

\begin{figure}[ht]
	\centering
	\includegraphics[scale = 0.6]{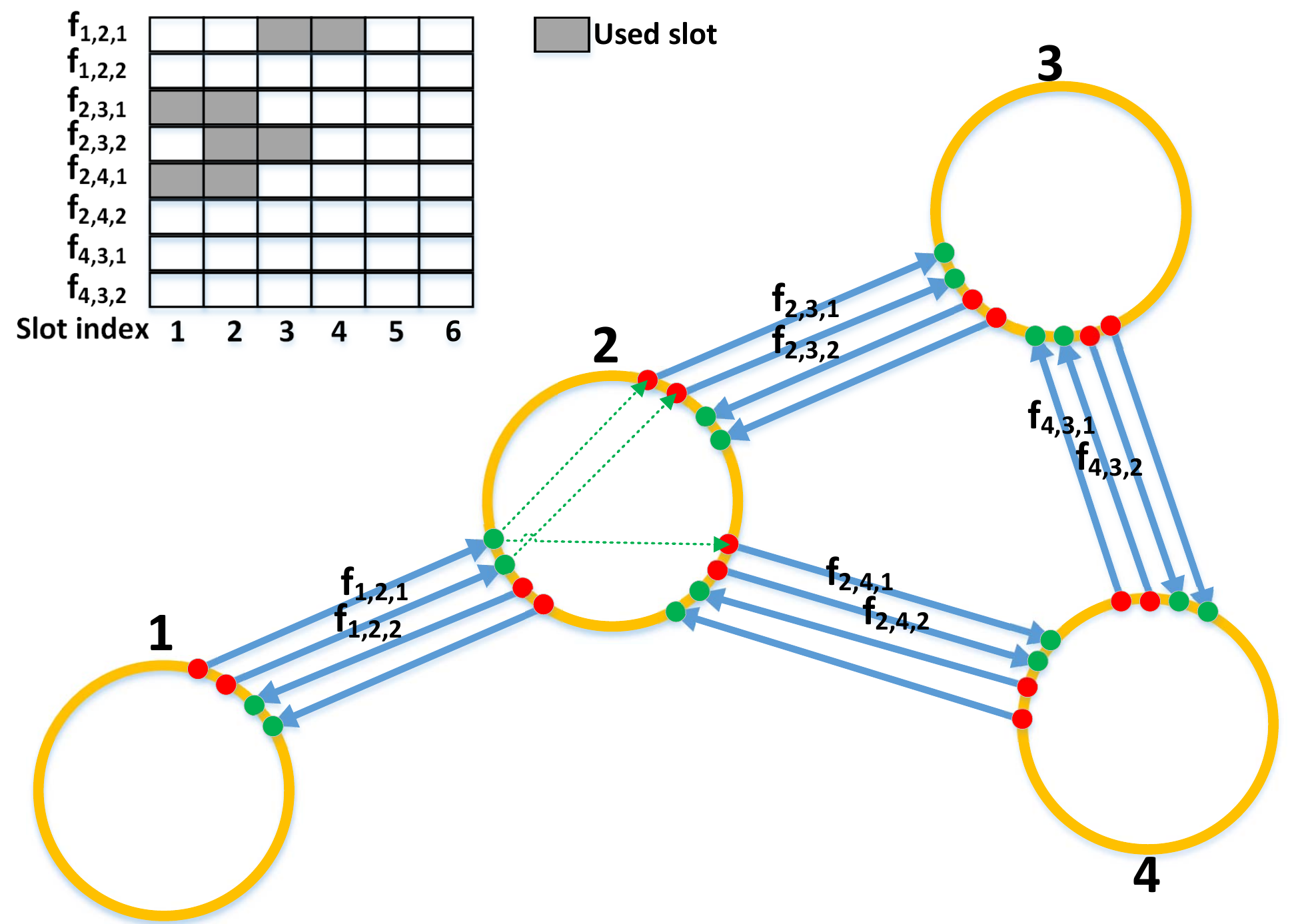}
	\caption{Illustrative example.}
	\label{fig52}
\end{figure}

\indent We formally define the problem as follows. $G=(V,E)$ represents the physical topology, wherein $V$ denotes the set of OXC nodes and $E$ denotes the set of unidirectional links. For each pair of adjacent nodes, there are two unidirectional links and $f_e$ unidirectional fibers on the unidirectional link $e$. We are given a set of requests ${\cal R}$. A request $r \in {\cal R}$ has source node $s_r$, destination node $d_r$, and a bandwidth requirement $w_r$ FSs. The resource allocation algorithm assigns a path, which consists of links and fibers, and FSs to the request $r$. The number of available slots on each fiber is assumed to be unlimited. Our objective is to minimize the maximum spectrum usage (MSU), i.e., the highest FS index used over all fibers to establish all traffic demands.

\indent We propose a new auxiliary layered-graph framework for solving the RFBSA problem. This framework generates an auxiliary graph for each outgoing fiber of the source node of a request, and generates one or more layered graphs for each auxiliary graph. The framework uses pluggable cost functions, and is therefore suitable for handling different network design objectives.

\section{The Proposed RFBSA Algorithm}\label{sec.label53}
\indent We present the details of our proposed RFBSA framework in this section.

\subsection{Auxiliary Layered-graph Framework}
\indent We start by describing the auxiliary layered-graph framework for solving the RFBSA problem. In the following, we denote $|f_{v_s,v_d}|$ as the number of fibers from $v_s$ to $v_d$.

\indent For each request $r$, we generate $\sum_{v_d}|f_{s_r,v_d}|$ auxiliary graphs. Each auxiliary graph involves one outgoing fiber from $s_r$ and includes all remaining nodes and fibers in the network. By fixing one outgoing fiber from the source node, we only need to consider the available FSs on that outgoing fiber, thereby reducing the complexity of searching for a lightpath. For a specific FS starting index $SI$ in an auxiliary graph $G_a=(V_a,E_a)$, we define a layered graph $G_{a,SI}=(V_{a,SI},E_{a,SI})$, which consists of all fibers of $G_a$ and FSs ranging from $SI$ to $SI+w_r-1$. The range of SI is from 1 to the current MSU.

\indent Since layered graph is at the FS level, it helps us to include spectrum contiguity, spectrum continuity, and spectrum non-overlapping constraints in each auxiliary graph. In the illustrative example, we can build two auxiliary graphs, each including one of the two outgoing fibers from node 1. In the first auxiliary graph, the only edge between node 1 and node 2 is $f_{1,2,1}$, while the only edge between node 1 and node 2 is $f_{1,2,2}$ in the second auxiliary graph. Each auxiliary graph has all edges among nodes 2, 3, and 4. Also, each auxiliary graph has layered graphs with contiguous FSs from $SI$ to $SI+w_r-1$, where $w_r$ equals 2 in the example.

\subsection{Resource Assignment Algorithm}
\indent In this subsection, we describe the algorithm for solving the RFBSA problem based on the above auxiliary layered-graph framework. The pseudocode is shown in Algorithm \ref{alg51}.

\indent We first sort all requests based on their bandwidth requirements in non-increasing order and store them in an ordered list ${\cal R^*}$. We consider lightpath assignments from the first request in the list, i.e., the one with the highest bandwidth demand. For each request, we first use the algorithm {\sc Switching Cost Update} for updating costs of each incoming fiber in each ``FLEX'' node.

\indent The basic idea of {\sc Switching Cost Update} is that we prefer to use already established wavebands (if any) and leave more waveband choices for subsequent requests. For the input fiber $f_{in}$ in ``FLEX'' node $v$, we compute the switching cost for the outgoing fiber $f_{out}$ of node $v$. If the waveband from $f_{in}$ to $f_{out}$ is already established, the switching cost from $f_{in}$ to $f_{out}$ is 0. If the waveband from $f_{in}$ to $f_{out}$ has not been established and the number of established wavebands from $f_{in}$ to all outgoing fibers of node $v$ is $b < B$, the switching cost between the two fibers is $\alpha{\cdot}b$. Here, $\alpha$ is a tradeoff parameter. If the number of established wavebands from $f_i$ to all outgoing fibers of node $v$ already equals $B$, the switching cost from $f_{in}$ to $f_{out}$ is set to $\infty$ because we cannot establish any more wavebands in ``FLEX''. The pseudocode is shown in Algorithm \ref{alg52}, and the cost function is shown in Equation \ref{equ1}. Denote $C_{v,f_{in},f_{out}}$ as the switching cost from $f_{in}$ to $f_{out}$ in ``FLEX'' node $v$.

\begin{equation}\label{equ53}
	C_{v,f_{in},f_{out}}=\left\{
	\begin{array}{l l}
		0,  \ \ \ \ \mbox{if the waveband is established}\\
		\alpha{\cdot}b, \ \ \mbox{if $b < B$} \\
		\infty,  \ \ \ \mbox{otherwise.}
	\end{array}
	\right.
\end{equation}

\begin{algorithm}                      
	\caption{{\sc RSBFA Algorithm}}          
	\label{alg51}                           
	\begin{algorithmic}                    
		\STATE{Input: ${\cal R}$, $G=(V,E)$}
		\vspace{-4pt}
		\STATE{Output: ${Path}$}
		\vspace{-4pt}
		\STATE{Sort all requests based on bandwidth requirements in non-increasing order and store in ${\cal R^*}$}
		\vspace{-4pt}
		\FOR{$r$ : ${\cal R^*}$}
		\vspace{-4pt}
		\FOR{$\kappa$ : outgoing fibers connecting to $s_r$}
		\vspace{-4pt}
		\STATE{Create auxiliary graph $AG_{\kappa}$}
		\vspace{-4pt}
		\FOR{$\widetilde{s}$ : available FS sets on $\kappa$}
		\vspace{-4pt}
		\IF{$|\widetilde{s}|<w_r$}
		\vspace{-4pt}
		\STATE{$C_{r,\kappa}=\infty$}
		\vspace{-4pt}
		\STATE{CONTINUE}
		\vspace{-4pt}
		\ENDIF
		\vspace{-4pt}
		\STATE{$SI^{start}_{\widetilde{s}}$ = the start FS index of $\widetilde{s}$}
		\STATE{$SI^{end}_{\widetilde{s}}$ = the last FS index of $\widetilde{s}$}
		\STATE{$SI=SI^{start}_{\widetilde{s}}$}
		\WHILE{$SI{\leq}SI^{end}_{\widetilde{s}}-w_r+1$}
		\STATE{$\widetilde{ss}$ = FS set from $SI$ to $SI+w_r-1$}
		\STATE{{\sc Spectrum Cost Update}}
		\STATE{Use shortest path algorithm to find a lightpath with the smallest $C_{total,r,\kappa,\widetilde{ss}}$}
		\IF{$C_{total,r,\kappa,\widetilde{ss}}<\infty$}
		\vspace{-4pt}
		\STATE{$C_{total,r,\kappa}=C_{total,r,\kappa,\widetilde{ss}}$}
		\vspace{-4pt}
		\STATE{BREAK}
		\vspace{-4pt}
		\ENDIF
		\vspace{-4pt}
		\STATE{$SI=SI+1$}
		\vspace{-4pt}
		\ENDWHILE
		\vspace{-4pt}
		\IF{$C_{total,r,\kappa}<\infty$}
		\vspace{-4pt}
		\STATE{BREAK}
		\vspace{-4pt}
		\ENDIF
		\vspace{-4pt}
		\ENDFOR
		\vspace{-4pt}
		\ENDFOR
		\vspace{-4pt}
		\STATE{$path_r$ = a lightpath with the smallest $C_{total,r,\kappa}, \forall \kappa$}
		\STATE{${Path}.add(path_r)$}
		\STATE{{\sc Switching Cost Update}} along the $path_r$
		\ENDFOR
	\end{algorithmic}
\end{algorithm}

\indent After updating switching costs, we start to generate auxiliary graphs for each outgoing fiber from $s_r$ for request $r$. For auxiliary graph $AG_{\kappa}$ of outgoing fiber $\kappa$, we need to update spectrum costs of all connected fibers. Define $C_{r,\kappa}$ as the spectrum cost of output fiber $\kappa$ for request $r$. We consider contiguous FS sets on $\kappa$ one by one. If the number of FSs in FS set $\widetilde{s}$ is less than $w_r$, we set $C_{r,\kappa}$ to infinity. We only consider the case that the number of FSs in $\widetilde{s}$ is no less than $w_r$ in the following. Denote $SI^{start}_{\widetilde{s}}$ and $SI^{end}_{\widetilde{s}}$ as the first and the last FS index in $\widetilde{s}$, respectively. Further, denote $SI$ as the FS index and denote $\widetilde{ss}$ as the FS set ranging from $SI$ to $SI+w_r-1$. We increase $SI$ starting from $SI^{start}_{\widetilde{s}}$ and stopping at $(SI^{end}_{\widetilde{s}} - w_r+2)$. For each $SI$, {\sc Spectrum Cost Update} is used to compute the spectrum cost of $\widetilde{ss}$ of all fibers in $AG_{\kappa}$.

\begin{algorithm}                      
	\caption{{\sc Switching Cost Update}}          
	\label{alg52}                           
	\begin{algorithmic}                    
		\STATE{Input:$v$, $f_{in}$, $f_{out}$}
		\STATE{Output:$C_{v,f_{in},f_{out}}$}
		\FOR{$f_{in}$ : all incoming fibers in node v}
		\FOR{$f_{out}$ : all outgoing fibers in node v}
		\STATE{$b$ = the number of established wavebands from $f_{in}$}
		\IF{the waveband is established}
		\STATE{$C_{v,f_{in},f_{out}}=0$}
		\ELSIF{$b<B$}
		\STATE{$C_{v,f_{in},f_{out}}={\alpha}{\cdot}b$}
		\ELSE
		\STATE{$C_{v,f_{in},f_{out}}=\infty$}
		\ENDIF
		\ENDFOR
		\ENDFOR
	\end{algorithmic}
\end{algorithm}

\indent The basic idea of {\sc Spectrum Cost Update} is to minimize the maximum FS index on each fiber after the current request is set up. We denote $f_{sl}$ as the largest FS index of fiber $f$ and denote $C_{r,\kappa,\widetilde{ss},f}$ as the spectrum cost on fiber $f$ of $\widetilde{ss}$ in $AG_{\kappa}$ for request $r$. If $\widetilde{ss}$ is not available on fiber $f$, $C_{r,\kappa,\widetilde{ss},f}$ is set to $\infty$. This means that we cannot establish a lightpath by using fiber $f$. If $\widetilde{ss}$ is available and $SI+w_r-1$ is no more than $f_{sl}$, $C_{r,\kappa,\widetilde{ss},f}$ is set to 1. In this case, the maximum FS index on fiber $f$ would not increase, and the spectrum cost works as a counter for the number of hops on the path. Otherwise, $C_{r,\kappa,\widetilde{ss},f}$ is $SI+w_r-1$. In this case, the spectrum cost is proportional to the maximum FS index after the request is set up. The pseudocode is shown in Algorithm \ref{alg53} and $C_{r,\kappa,\widetilde{ss},f}$ is shown in Equation \ref{equ54}.
\begin{equation}\label{equ54}
	C_{r,\kappa,\widetilde{ss},f}=\left\{
	\begin{array}{l l}
		1,  \ \ \ \ \ \mbox{if $SI+w_r-1 \leq f_{sl}$} \\
		\infty, \ \ \ \mbox{if $\widetilde{ss}$ is not available on $f$}\\
		SI+w_r-1,  \ \ \ \mbox{otherwise.}
	\end{array}
	\right.
\end{equation}

\begin{algorithm}                      
	\caption{{\sc Spectrum Cost Update}}          
	\label{alg53}                           
	\begin{algorithmic}                    
		\STATE{Input:$\widetilde{ss}$, $f$, $SI$, $w_r$}
		\STATE{Output:$C_{r,\kappa,\widetilde{ss}}$}
		\FOR{$f$ : all fibers in the network}
		\STATE{$f_{sl}$ = the largest FS index on $f$}
		\IF{$\widetilde{ss}$ is not available on $f$}
		\STATE{$C_{r,\kappa,\widetilde{ss},f}=\infty$}
		\ELSIF{$SI+w_r-1{\leq}f_{sl}$}
		\STATE{$C_{r,\kappa,\widetilde{ss},f}=1$}
		\ELSE
		\STATE{$C_{r,\kappa,\widetilde{ss},f}=SI+w_r-1$}
		\ENDIF
		\ENDFOR
	\end{algorithmic}
\end{algorithm}

\indent For a lightpath $l$, we denote $C_{total,r,\kappa,\widetilde{ss}}$ as the total cost combining Switching Cost $C_{v,f_in,f_{out}}$ of the ``FLEX'' node set $V_l$ and Spectrum Cost $C_{r,\kappa,\widetilde{ss},f}$ of the fiber set $F_l$. This is defined in Equation \ref{equ55}.
\begin{equation}\label{equ55}
	C_{total,r,\kappa,\widetilde{ss}}=\sum_{f_{in},f_{out}{\in}V_l}C_{v,f_{in},f_{out}}+\beta{\cdot}\sum_{f{\in}F_l}C_{r,\kappa,\widetilde{ss},f},
\end{equation}

\noindent
where $\beta$ is a tradeoff parameter between Switching Cost and Spectrum Cost.

\indent Once the costs for the edges in the layered graph are assigned, we use Dijkstra's algorithm to find the shortest paths from the source node output fiber to different input fibers of destination node, and choose the path $p$ with smallest total cost. If the cost of path $p$ is less than infinity, in order to keep the algorithm's complexity low, we stop generating further layered graphs for this auxiliary graph, and consider $p$ as a {\em candidate path}. We then proceed to find a candidate path on the next auxiliary graph, and so on, until one (or zero, if no paths have finite cost) candidate path is found for each auxiliary graph. Finally, we compare the total costs of all the candidate paths and choose a path with the smallest total cost. Assume $N$ is total number of nodes in the network with max node degree $D$, $L$ is total number of links with max number of fibers $F$ per link, and $\eta$ is the estimated upper bound of the maximum spectral usage. For each request, for each output fiber of the source node, $O(\eta)$ auxiliary graphs are created. Each graph consists of $2LF$ vertices and $LF+N(DF)^2$ edges. The dominant part is to find the shortest path in the auxiliary graphs, thus the time complexity of this algorithm is $O(|{\cal R}|DF\eta N(DF)^2\log(LF))$.

\indent In the illustrative example, we first consider auxiliary graph with $f_{1,2,1}$ and start to generate layered graphs from the first available contiguous FSs. In this auxiliary graph, we cannot build new waveband connections without making the switching cost to be infinity. In the first layered graph, which consists of FSs 1 and 2, the spectrum costs of $f_{2,3,1}$ and $f_{2,4,1}$ are both infinity. Then, we consider the second layered graph, which consists of FSs 5 and 6, and $SI$ equals 5. In this layered graph, the spectrum cost on each fiber is $SI+w_r-1=6$ and the spectrum cost of a lightpath is the number of fibers of the lightpath multiplied by the spectrum cost on each fiber. Thus, the spectrum cost by using $f_{1,2,1}$ and $f_{2,3,1}$ is $2(SI+w_r-1)=12$, and the spectrum cost by using $f_{1,2,1}$, $f_{2,4,1}$, and $f_{4,3,1}$ (or $f_{4,3,2}$) is $3(SI+w_r-1)=18$. Thus, the best path of the first auxiliary graph is the path consisting of $f_{1,2,1}$ and $f_{2,3,1}$ and the total cost is ${\beta}{\cdot}12$. Here, the switching cost is 0. The second auxiliary graph starts from $f_{1,2,2}$. The spectrum cost of the path consisting of $f_{1,2,2}$ and $f_{2,3,2}$ is $2(SI+w_r-1)=10$, where $SI=4$, and switching cost is 0. The spectrum cost of the path consisting of $f_{1,2,2}$, $f_{2,4,2}$, and $f_{4,3,1}$ is $3(SI+w_r-1)=6$, where $SI=1$ and switching cost is $\alpha$. Suppose we set $\alpha=0.1$ and $\beta=1$, the second path is the best path in the second auxiliary graph. Finally, we compare total costs of the best path from each auxiliary graph and the path consisting of $f_{1,2,2}$, $f_{2,4,2}$, and $f_{4,3,1}$ is chosen as the best path. The lightpath for the request is thus established on the fibers and FSs that correspond to this best path.

\section{Simulation Results}\label{sec.label54}
\indent In this section, we present simulation results to demonstrate the effectiveness of our proposed RFBSA algorithm and to compare the performances of the ``FLEX'' and ``FLEX'' architectures. The network topology used for simulations is the NSFNET network topology as shown in Figure \ref{fig:NSF}, which has 14 physical nodes and 22 bidirectional links \cite{wu2015comparison}. Each link has multiple parallel fibers. We apply two different settings: 1) all links have the same number of parallel fibers ($F=2$ per link); 2) the number of fibers on each link is randomly selected according to a uniform distribution between 5 and 10 ($F=[5,10]$ per link).

\indent We use a static traffic demand model -- a given set of traffic requests is to be allocated resources in the network. Each request denotes a connection between a pair of nodes in the network. The source and destination nodes for each connection request are uniformly randomly selected from the physical nodes of the network. We assume three different types of demands with different number of desired frequency slots \cite{hasegawa2015flexible}. The number of frequency slots required by each request is chosen from the following distribution: 3 slots with probability 0.2, 4 slots with probability 0.5, and 7 slots with probability 0.3. For the ``FLEX'' node architecture, the limited number of wavebands $B$ is assumed to be 4, which is equal to the port count of commercially available $1 \times 4$ WSS.

\indent For each given set of static traffic demands, we record the MSU, which is the largest used slot index over all fibers of the network. The results from our proposed algorithm are compared with those from a commonly-used baseline algorithm for both ``FLEX'' and ``FLEX'' architectures. The performances are also compared between the two architectures.

\indent The baseline algorithm we adapt to our multi-fiber network is the Shortest Path First Fit (SPFF) algorithm. In this algorithm, a single shortest path (i.e., smallest number of hops) is pre-defined for each source-destination pair. For each request, we try to find the first fit (i.e., a set of FSs with the lowest starting index) slot set (number of contiguous slots equal to the request demand) on the first fit fibers (considering any switching constraints) along the pre-calculated route. If multiple shortest paths are pre-defined, the baseline algorithm could be adapted to a K-Shortest Path First Fit (KSPFF) algorithm.

\indent Since the conventional node architecture has no switching constraints, any input fiber can be switched to any output fiber on each node. The selection of fibers on previous links does not affect the fiber selection on other links. For all links along the route, if a block of slots is available on at least one fiber on each link, this block of slots is allocated to the request. The first fiber containing available resources on each link is selected to accommodate the request. We use $CONV_{SPFF}$ to denote this baseline algorithm for the conventional architecture.

\indent For the ``FLEX'' architecture, there is a switching constraint defined by $B$ -- each input fiber can be switched to up to $B$ output fibers for each node. The fiber selection on the previous link of the route will affect the decision on the current link, and then all the following links on the route. If an input fiber to a node is already switched to $B$ different output fibers, with a subset belonging to the desired output link, then the fiber selection on the next link can be chosen only from that subset. Taking this switching constraint into consideration, we first filter out all unreachable fibers along the route, and then find the first fit slot set on the first fit fibers. We call this baseline algorithm for the ``FLEX'' architecture as $FLEX_{SPFF}$.

\begin{figure}
	\includegraphics[scale=0.55]{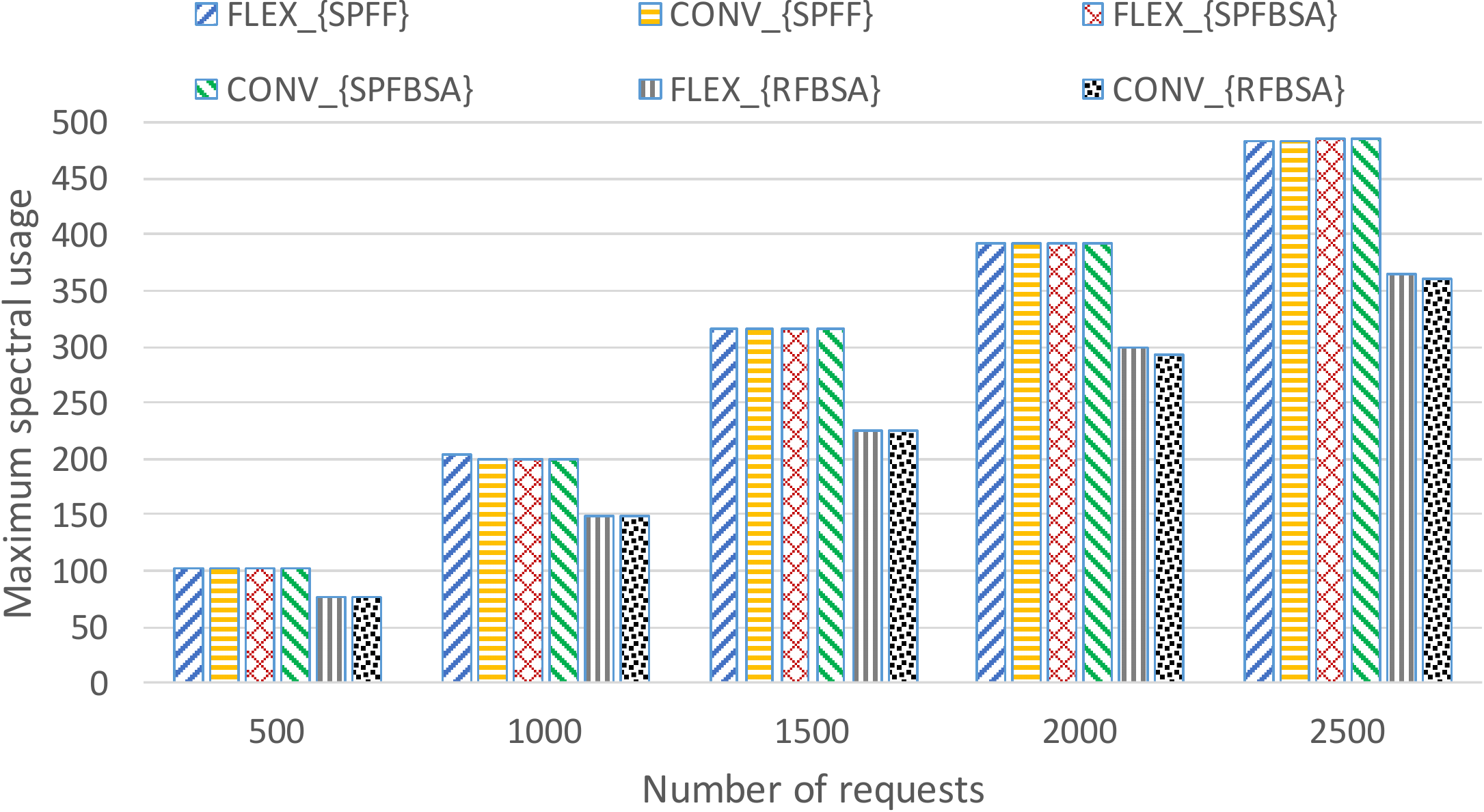}
	\centering
	\caption{\label{fig53} Performance comparison results for NSF network (F=2).}
\end{figure}

\begin{figure}
	\includegraphics[scale=0.55]{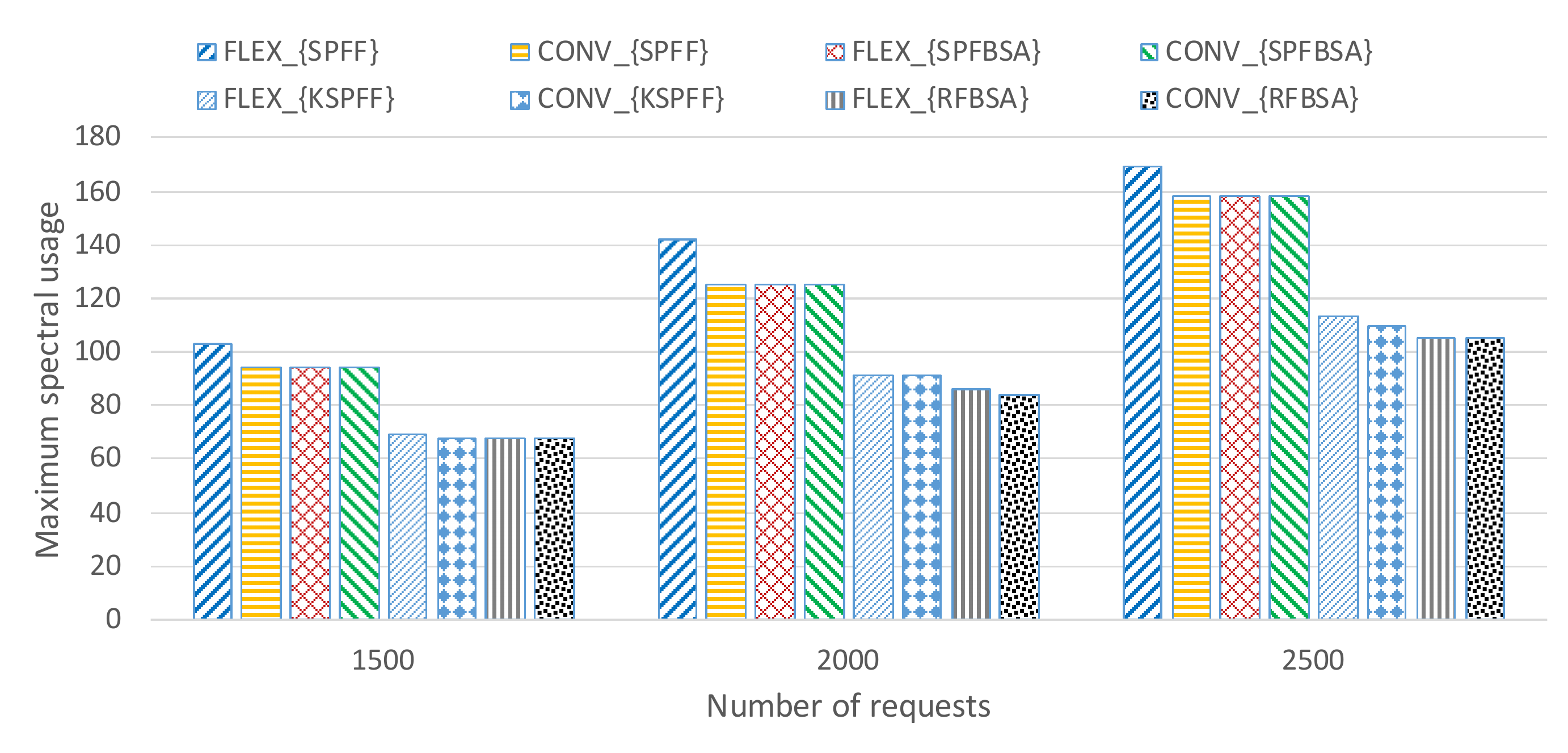}
	\centering
	\caption{\label{fig54} Performance comparison results for NSF network (F=[5,10]).}
\end{figure}

\indent Figure \ref{fig53} and \ref{fig54} show the performance results when utilizing our proposed algorithm, given different number of static traffic requests for the simulation settings $F=2$ and $F=[5,10]$, respectively. The parameter settings for the above results are: $\alpha = 1$, and $\beta = 1$. This parameter set works well because both the waveband and slot status are jointly considered with the same weight when making allocation decisions. The performance results are compared with those of $CONV_{SPFF}$ and $FLEX_{SPFF}$. We also compare the performance of our joint RFBSA algorithm with that of an algorithm in which the routing is fixed (smallest number of hops) but the frequency slots are selected using our auxiliary layered-graph approach. We call this algorithm $FLEX_{SPFBSA}$ and $CONV_{SPFBSA}$ for the two architectures. In Figure \ref{fig54}, we also compare the performance of our joint RFBSA algorithm with that of KSPFF for both architectures (denoted as $FLEX_{KSPFF}$ and $CONV_{KSPFF}$).

\indent We make several observations from the results. First, we note that SPFBSA and SPFF perform very similarly for both the ``FLEX'' and the ``FLEX'' architectures when $F = 2$. This implies that the first fit slot assignment is sufficient when SP routing is used in networks with a small number of fibers. Further, when $F = 2$, the performance of the ``CONV'' architecture and the ``FLEX'' architecture are almost the same using any algorithm. The waveband constraint does not have much influence because the number of fibers per link is small. For the $F = [5,10]$ case, due to the waveband switching constraints, the MSU of $FLEX_{SPFF}$ is slightly higher than that of $CONV_{SPFF}$. In this case, $FLEX_{SPFBSA}$ can find better results (almost the same as $CONV_{SPFF}$ and $CONV_{SPFBSA}$) than $FLEX_{SPFF}$.  

\indent Our RFBSA algorithm gives much better results than the baseline algorithms with shortest path routing for both the ``CONV'' architecture and ``FLEX'' architecture. Note that there is an improvement of $25-33\%$ in slot usage. Since MSU may be limited by a single fiber in the network, we conduct another set of comparison between the KSPFF and our RFBSA algorithm by using the average MSU across the network as performance indicator. Figure \ref{fig55} shows that RFBSA outperforms KSPFF.

\begin{figure}
	\includegraphics[scale=0.55]{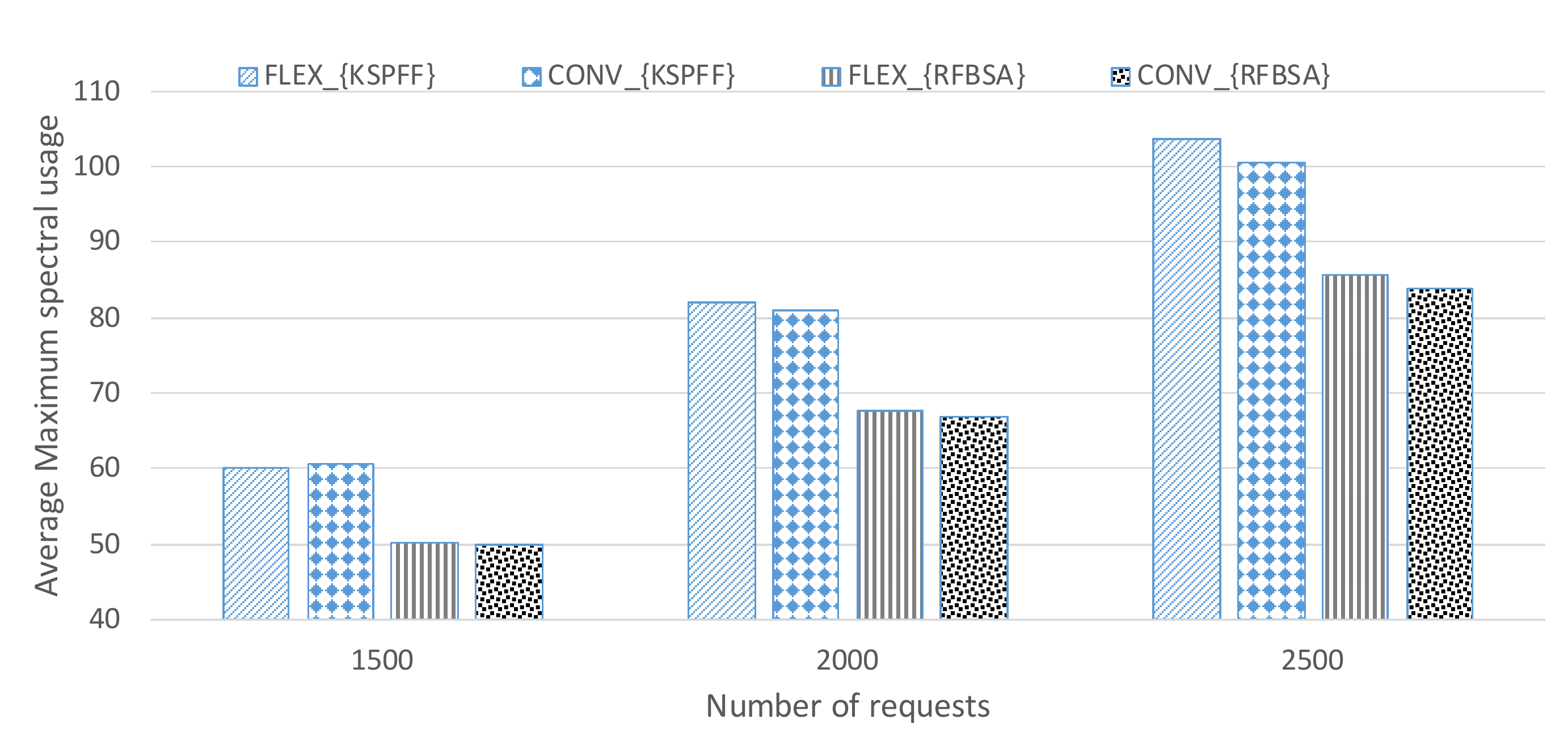}
	\centering
	\caption{\label{fig55} Comparison between KSPFF and RFBSA for NSF network (F=[5,10]).}
\end{figure}

This demonstrates the effectiveness of joint routing, fiber, band, and spectrum selection. We also note that the ``FLEX'' architecture, even with a small $B (= 4)$, can achieve a performance comparable to the ``CONV'' architecture through judicious resource allocation. Recall that the ``FLEX'' architecture is much less expensive than the ``CONV'' architecture.

\section{Conclusions}\label{sec.label55}
\indent In this chapter, the joint routing, fiber, waveband, and spectrum allocation (RFBSA) problem for EONs with a novel cross-connect architecture and multi-fiber links is investigated. Our main contributions can be summarized as follows:
\begin{itemize}
	\item We formulate a new and practical problem called the RFBSA problem. To the best of our knowledge, this is the first work that considers the RFBSA problem.
	\item We propose an auxiliary layered-graph framework for solving the RFBSA problem. The framework is suitable for solving different objective functions, and is cost-function pluggable. We also propose cost functions for minimizing the maximum used FS.
	\item Simulation results show the effectiveness of our framework by comparing its performance results with the results from a standard baseline algorithm.
\end{itemize}

\chapter{Joint Banding-node Placement and Resource Allocation for Multi-granular Elastic Optical Networks}
\label{chap_6}

\indent In this chapter, we continue our work based on the previous chapter. The following problem is addressed. Given the total number of available WSSs for the network as a budget for network design, determine how many FLEX nodes to deploy and where to deploy them, and solve the RFBSA problem jointly to optimize the network performance.

\section{Background and Problem Statement}\label{sec.label61}

\indent Due to the switching constraint in the FLEX architecture, deploying FLEX nodes in the network may cause worse performance than utilizing CONV nodes. To achieve both cost efficiency and good performance, the deployment of CONV and FLEX nodes should be carefully addressed. Different node placements may need different numbers of WSSs. If a budget in terms of number of available WSSs is given for network planning, we would like to determine how many FLEX nodes should be deployed and where to place them. 

\indent Each request can be represented by a source node, a destination node, and a bandwidth requirement in terms of number of required FSs. A path (including links and fibers) that satisfies routing constraints caused by limited wavebands, and a set of contiguous FSs should be assigned to accommodate the request. The node placement will have a direct impact on the resource assignment and performance of the network, as different FLEX node placements will cause different switching constraints. Thus, we should jointly consider the {\em RFBSA} and node placement problem. 

\indent For a given set of traffic requests and a network planning budget in terms of the number of WSSs of a given fixed port count, our objective is to find the number and locations of FLEX nodes as well as the resource allocation (FSs on fibers) of requests to minimize the total maximum spectrum usage (MSU), which is the sum of the maximum slot indices used on all fibers in the network. In the case of dynamically arriving and departing traffic requests, given a planning budget, our objective is to determine node placements and perform resource allocation to minimize the demand blocking ratio for dynamic traffic requests. The demand blocking ratio is defined as the ratio of the sum of bandwidths of blocked requests to the sum of bandwidths of all requests. Since requests can have different bandwidth requirements, we use demand blocking ratio instead of lightpath blocking ratio as the performance indicator.

\section{The Integer Linear Programming Model}\label{sec.label62}
\label{sec:ILP}

\indent We now present an ILP formulation for the joint problem. The ILP model can be used to solve small problem instances for a set of static connection requests.

\subsection{Notations}
\indent The input parameters to the ILP formulation are shown in Table \ref{table:NP_ILPnotations}. Consider a network $\mathcal{G}=(\mathcal{V},\mathcal{L})$, where $\mathcal{V}$ denotes the set of OXC nodes, and $\mathcal{L}$ denotes the set of unidirectional links.

\renewcommand\arraystretch{1.1}
\begin{table}\footnotesize
	\centering
	\caption{\label{table:NP_ILPnotations} Notation for the banding-node placement problem}
	\begin{tabular}{|>{\centering}m{1.5cm}| m{13.5cm}<{\centering}|}
		\hline
		\bf{Symbol}& \bf{Meaning} \\ \hline
		$G$ & the network topology \\ \hline
		$\mathcal{V}$ & the set of OXC nodes \\ \hline
		$\mathcal{L}$ & the set of unidirectional links \\ \hline
		$f_l$ & the number of parallel fibers on link $l{\in}\mathcal{L}$ \\ \hline
		$M$ &  number of physical nodes in the network \\ \hline
		$L$ &  number of physical links in the network\\ \hline
		$F$ &  total number of fibers in the network\\ \hline
		$B$ &  the band limit of FLEX nodes\\ \hline
		$T$ &  the budget, the maximum available number of WSSs\\ \hline
		$K$ &  the number of predetermined paths between each node pair\\ \hline
		$v$ &  an arbitrary network node\\ \hline
		$e$ &  an arbitrary network link\\ \hline
		$f$ &  an arbitrary network fiber\\ \hline
		$s$ &  an arbitrary slot\\ \hline
		$p$ &  an arbitrary path ($p_{s, d, r}$ is the r\emph{th} shortest path from node $s$ to node $d$)\\ \hline
		$P_{s,d,r}^{e}$ & = 1 if link $e$ is on path $p_{s, d, r}$; = 0, otherwise\\ \hline
		$\xi_{e}^{f}$ & = 1 if fiber $f$ is on link $e$; = 0, otherwise\\ \hline
		$IN_{v}^{f}$ & = 1 if fiber $f$ is an input fiber of node $v$; = 0, otherwise\\ \hline
		$OUT_{v}^{f}$ & = 1 if fiber $f$ is an output fiber of node $v$; = 0, otherwise\\ \hline
		$\zeta_{v}$ & the number of WSSs needed if node $v$ is CONV\\ \hline
		$\Gamma _{v}$ & the number of WSSs needed if node $v$ is FLEX\\ \hline
		$\Omega$ & estimated upper bound of maximum slot usage\\ \hline
		$\mathcal{J}$ & a given set of requests\\ \hline
		$J$ & number of requests \\ \hline
		$j$ & an arbitrary request\\ \hline
		$s^j$ & source node of request $j$, $s^j{\in}\mathcal{V}$\\ \hline
		$t^j$ & destination node of request $j$, $t^j{\in}\mathcal{V}$\\ \hline
		$d^{j}$ & the required number of slots for request $j$\\ \hline
	\end{tabular}
\end{table}

\subsection{Formulations}

\indent We use $\mathcal{J}$ to denote a given set of requests. Request $j{\in}\mathcal{J}$ requires $d^{j}$ FSs with source node $s^j$ and destination node $t^j$. The $K$ shortest paths for each request are precomputed. $T$ denotes the network planning target, which is the total number of available WSSs. Our goal is to determine the placement of FLEX nodes and resource allocation for requests. The MSU for a fiber is defined as the highest slot index utilized to accommodate demands on that fiber.

Variables:

a)
\begin{equation*}
a^{j}_{s}=
\begin{cases}
1, \quad\makebox{if starting slot of request } j \makebox{ is } s;\\
0, \quad\makebox{otherwise }
\end{cases}
\end{equation*}

b)
\begin{equation*}
z^{j}_{s}=
\begin{cases}
1, \quad\makebox{if request } j \makebox{ uses slot } s;\\
0, \quad\makebox{otherwise }
\end{cases}
\end{equation*}

c) 
\begin{equation*}
y^{j}_{f}=
\begin{cases}
1, \quad\makebox{if fiber } f \makebox{ is allocated to accommodate request } j;\\
0, \quad\makebox{otherwise }
\end{cases}
\end{equation*}

d)\vspace{-5pt}
\begin{equation*}
x^{j}_{f,s}=
\begin{cases}
1, \quad\makebox{if slot } s \makebox{ on fiber } f \makebox{ is allocated to request } j;\\
0, \quad\makebox{otherwise }
\end{cases}
\end{equation*}

e)\vspace{-5pt}
\begin{equation*}
w^{j,v}_{f_a,f_b}=
\begin{cases}
1, \quad\makebox{if there is a switching from fiber } f_a \makebox{ to}\\ \quad\quad \makebox{fiber } f_b \makebox{ at node } v \makebox{ caused by request } j;\\
0, \quad\makebox{otherwise }
\end{cases}
\end{equation*}

f)\vspace{-5pt}
\begin{equation*}
W^{v}_{f_a,f_b}=
\begin{cases}
1, \quad\makebox{if there is a switching from fiber } f_a\\ \quad\quad \makebox{to fiber } f_b \makebox{ at node } v;\\
0,\quad\makebox{otherwise }
\end{cases}
\end{equation*}

g)\vspace{-5pt}
\begin{equation*}
C_{v}=
\begin{cases}
1, \quad\makebox{if node } v \makebox{ is chosen to use CONV architecture};\\
0,\quad\makebox{if node } v \makebox{ is chosen to use FLEX architecture}
\end{cases}
\end{equation*}

h)\vspace{-5pt}
\begin{equation*}
\lambda^{j}_{r}=
\begin{cases}
1, \quad\makebox{if request } j \makebox{ uses path } r \makebox{ of the $K$ shortest paths};\\
0,\quad\makebox{otherwise }
\end{cases}
\end{equation*}

\indent Our objective is to minimize the average MSU, or, equivalently, the total MSU over all fibers, as the number of fibers for a single network is constant:
\begin{center}Minimize \quad $\sum_{f=1}^{F}\max_{s} (s\cdot\sum_{j=1}^{J}x^{j}_{f,s})$\end{center}
Here, $(s\cdot\sum_{j=1}^{J}x^{j}_{f,s})$ denotes the index of a slot that is used by some request on fiber $f$.

Constraints:

a) There is only one starting slot index for each request. For all $j$,
\begin{equation*}
\sum_{s=1}^{\Omega-d^{j}+1} a^{j}_{s} = 1, \quad
\sum_{s=\Omega-d^{j}}^{\Omega} a^{j}_{s} = 0
\end{equation*}

b) The spectrum contiguity should be met for each request. The slots assigned to accommodate a request must use consecutive slots from its starting slot. For all $j, s$,
\begin{equation*}
\sum_{i=0}^{d^{j}-1}z^{j}_{s+i} \geq d^{j} \cdot a^{j}_{s}, \quad 
\sum_{s=1}^{\Omega}z^{j}_{s} = d^{j}
\end{equation*}

c) Only one of the $K$ precomputed shortest paths should be chosen for a request. For all $j$,
\begin{equation*}
\sum_{r=1}^{K}\lambda^{j}_{r}=1
\end{equation*}

d) Only one fiber on each link along the chosen path is selected to route each request. For all $j, e$,
\begin{equation*}
\sum_{f=1}^{F}\xi_{e}^f\cdot y^{j}_{f} = \sum_{r=1}^{K}P_{s,d,r}^{e}\cdot \lambda^{j}_{r}
\end{equation*}

e) The value of $x^{j}_{f,s}$ should be based on both $y^{j}_{f}$ and $z^{j}_{s}$. Slot $s$ on fiber $f$ is used by request $j$ if and only if slot $s$ is assigned to $j$, and fiber $f$ is allocated on the path. For all $j, f, s$,
\begin{equation*}
y^{j}_{f} + z^{j}_{s} \leq 1 + x^{j}_{f,s}, \quad y^{j}_{f} + z^{j}_{s} \geq 2x^{j}_{f,s}
\end{equation*}

f) Any slot on any fiber can accommodate at most one request. For all $f,s$,
\vspace{-10pt}
\begin{equation*}
\sum_{j=1}^{J}x^{j}_{f,s} \leq 1
\end{equation*}

g) The value of $w^{j,v}_{f_a,f_b}$ is based on the value of $y^{j}_{f_a}$ and $y^{j}_{f_b}$. For any two fibers $f_a$ and $f_b$, there is a switching between them to route request $j$ if $f_a$ is an input fiber and $f_b$ is an output fiber of a node, and both $f_a$ and $f_b$ are allocated to the request. For all $j, v, f_a, f_b$,
\begin{equation*}
y^{j}_{f_a}\cdot IN_{v}^{f_a} + y^{j}_{f_b}\cdot OUT_{v}^{f_b} \leq 1 + w^{j,v}_{f_a,f_b}
\end{equation*}
\begin{equation*}
y^{j}_{f_a}\cdot IN_{v}^{f_a} + y^{j}_{f_b}\cdot OUT_{v}^{f_b} \geq 2w^{j,v}_{f_a,f_b}
\end{equation*}

h) There is a waveband from fiber $f_a$ to fiber $f_b$ if at least one request is routed by this switching. For all $v,f_a,f_b$,
\begin{equation*}
\sum_{j=1}^{J}w^{j,v}_{f_a,f_b} \geq W^{v}_{f_a,f_b},\quad
\sum_{j=1}^{J}w^{j,v}_{f_a,f_b} \leq J\cdot W^{v}_{f_a,f_b}
\end{equation*}

i) The total required number of WSSs for the whole network should be no larger than the given budget.
\vspace{-5pt}
\begin{equation*}
\sum_{v=1}^{M}C_{v}\zeta_{v}+(1-C_{v})\Gamma _{v}\leq T
\end{equation*}

j) The routing capacity limit should be satisfied for each input fiber on any FLEX node. Here, $\Lambda$ is a large value to loosen the constraint for CONV nodes. For all $v,f_a$,
\vspace{-5pt}
\begin{equation*}
\sum_{f_b=1}^{F}W^{v}_{f_a,f_b}\leq B + \Lambda C_{v}
\end{equation*}

\section{Heuristic Solutions}\label{sec.label63}

\indent Our proposed heuristics to solve the problem are presented in detail in this section. We first describe our framework to solve the {\em RFBSA} problem for static instances. Then we propose the node placement scheme which is based on both topology information and results of {\em RFBSA}. Finally we apply the proposed schemes to accommodate dynamic traffic requests.

\subsection{Route, Fiber, Band and Slot Assignment Framework}

\indent For a given set of requests, we first sort them based on the bandwidth requirements in non-increasing order. The requests are stored in an ordered list $\mathcal{J}$. We consider resource assignments in order from the list, i.e., starting with the request having highest bandwidth demand.

\indent We start by describing the auxiliary layered-graph framework for solving the assignment problem. The vertexes in an auxiliary graph represent all input and output fibers (denoted by I vertices and O vertices). The edges between vertexes are classified into two types. The first type of edges is from an I vertex to an O vertex, representing the switching inside a corresponding physical node. Since there is a band limit on FLEX nodes in the network, we define a switching cost between each input fiber and each output fiber on a physical node. The costs of this type of edges are related to the wavebanding status of the node. The other type of edges from an O vertex to an I vertex represents the physical link connection. The costs of such edges are related to the spectrum usage on the particular fiber.

\indent In the following, we briefly describe the algorithm for route, fiber, band, and slot assignment. The pseudocode is shown in Algorithm \ref{alg0}. We modify the algorithm in our previous work \cite{wu2017joint} to reduce the time complexity.

\indent For each request $j$, we examine the $K$ shortest paths. For each path $k$, we generate a basic auxiliary graph $AG^k$ based on the links and fibers on this particular path. The algorithm {\sc Switching Cost Update} is used to update the costs of the first type of edges in $AG^k$. 

\indent The switching cost function is shown in Equation (\ref{equ1}). We denote $C_{v,f_{in},f_{out}}$ as the switching cost or cost of the edge from vertex $f_{in}$ to vertex $f_{out}$ in physical FLEX node $v$.
\begin{equation}\label{equ1}
C_{v,f_{in},f_{out}}=\left\{
\begin{array}{l l}
0,  \ \ \ \ \ \mbox{if \ waveband \ is \ established}\\
\alpha{\cdot}\frac{b}{B}, \ \ \mbox{if \ $b < B$} \\
\infty,  \ \ \ \mbox{otherwise}
\end{array}
\right.
\end{equation}

\noindent
The basic idea of Switching Cost Update (Equation (\ref{equ1})) is to give preference to existing wavebands (if any) so as to leave more waveband choices for later requests. When computing the switching cost from an input fiber $f_{in}$ to an outgoing fiber $f_{out}$ of a FLEX node $v$, we use the following rules. If the waveband from $f_{in}$ to $f_{out}$ already exists, the switching cost from $f_{in}$ to $f_{out}$ is set to 0. Otherwise, if the number of already established wavebands from $f_{in}$ to other outgoing fibers of node $v$ is $b < B$, the switching cost is set to $\alpha{\cdot}\frac{b}{B}$, where $\alpha$ is a tuning parameter. If the number of established wavebands corresponding to $f_i$ equals $B$, the switching cost is set to be $\infty$, which means that no new waveband can be established from $f_{in}$. If $v$ is a CONV node, the switching costs from any input fiber to any output fiber of that node are set to 0.

\begin{algorithm}                      
	\caption{{\sc {\em RFBSA} Algorithm}}          
	\label{alg0}                           
	\begin{algorithmic}                    
		\STATE{Input: $G=(V,E)$, request $j$, $K$ shortest paths}
		\STATE{Output: $path_j$ --- Fiber and slot assignment for ${j}$}
		\FOR{$k$ : $K$ shortest paths of request $j$}
		\STATE{Create auxiliary graph $AG_{k}$}
		\STATE{{\sc Switching Cost Update}} along path $k$
		\STATE{$C_{total}^k = \infty$}
		\FOR{$a$ : outgoing fibers from $s^j$}
		\STATE{Copy $AG_{k}$ to create auxiliary graph $AG_{a}^k$}
		\STATE{$C_{a}^k = \infty$}
		\FOR{$\widetilde{ss}$ : available FS sets of size $d^j$ on $a$}
		\STATE{$SI$ = the start FS index of $\widetilde{ss}$}
		\STATE{Create layered graph $AG_{a,SI}^k$}
		\STATE{{\sc Spectrum Cost Update}}
		\STATE{Use Dijkstra's algorithm to find a lightpath with the smallest cost $C_{a,\widetilde{ss}}^k$}
		\IF{$C_{a,\widetilde{ss}}^k<\infty$}
		\STATE{$C_{a}^k=C_{a,\widetilde{ss}}^k$}
		\STATE{BREAK}
		\ENDIF
		\ENDFOR
		\IF{$C_{a}^k<C_{total}^k$}
		\STATE{$C_{total}^k=C_{a}^k$}
		\ENDIF
		\ENDFOR
		\ENDFOR
		\STATE{$path_r$ = a lightpath $p$ with the smallest $C_{total}^k, \forall k$}
	\end{algorithmic}
\end{algorithm}

\indent Then we copy $AG^k$ to generate an auxiliary graph for each outgoing fiber from $s^j$ in the first link of that path for the request. Each auxiliary graph involves one outgoing fiber from source node $s^j$, and includes all remaining nodes and fibers in the $k$th shortest path. Let $|f^k_{l_0}|$ denote the number of fibers on the first link of the $k$th shortest path. In total, we generate $\sum_{k}|f^k_{l_0}|$ auxiliary graphs. 

\indent By fixing one outgoing fiber $a$ in the first link of the path, we only need to consider the available FSs on that fiber, thereby reducing the complexity of searching for a valid lightpath. In an auxiliary graph $AG_a^k=(V_a^k,E_a^k)$, we define a layered graph on FS level $AG_{a,SI}^k=(V_{a,SI}^k,E_{a,SI}^k)$ for a specific spectrum range from FS index $SI$ to $SI + d^j - 1$. Thus, in the layered graphs, spectrum contiguity, spectrum continuity, and spectrum non-overlapping constraints are automatically satisfied. For auxiliary layered graph $AG_{a,SI}^k$ corresponding to the outgoing fiber $a$ of $s^j$, we need to do Spectrum Cost Update for the second type of edges, which correspond to all fibers on the physical links along the path regarding an available slot set $\widetilde{ss}$ on fiber $a$.  

\begin{algorithm}                      
	\caption{{\sc Switching Cost Update}}          
	\label{alg1}                           
	\begin{algorithmic}                    
		\STATE{Input: $v$, $f_{in}$, $f_{out}$}
		\STATE{Output: $C_{v,f_{in},f_{out}}$}
		\FOR{$f_{in}$ : all incoming fibers in node v}
		\FOR{$f_{out}$ : all outgoing fibers in node v}
		\STATE{$b$ = the number of established wavebands from $f_{in}$}
		\IF{the waveband is established}
		\STATE{$C_{v,f_{in},f_{out}}=0$}
		\ELSIF{$b<B$}
		\STATE{$C_{v,f_{in},f_{out}}={\alpha}{\cdot}b/B$}
		\ELSE
		\STATE{$C_{v,f_{in},f_{out}}=\infty$}
		\ENDIF
		\ENDFOR
		\ENDFOR
	\end{algorithmic}
\end{algorithm}

\begin{algorithm}                      
	\caption{{\sc Spectrum Cost Update}}          
	\label{alg2}                           
	\begin{algorithmic}                    
		\STATE{Input: $f$, $\widetilde{ss}$, $SI$, $d^j$}
		\STATE{Output: $C_{a,\widetilde{ss},f}^k$}
		\STATE{$m_{f}$ = the largest FS index on $f$}
		\IF{$\widetilde{\widetilde{ss}}$ is not available on $f$}
		\STATE{$C_{a,\widetilde{ss},f}^k=\infty$}
		\ELSIF{$SI+d^j{\leq}m_{f}$}
		\STATE{$C_{a,\widetilde{ss},f}^k=1$}
		\ELSE
		\STATE{$C_{a,\widetilde{ss},f}^k=(SI + d^j - m_{f})/\Omega + 1$}
		\ENDIF
	\end{algorithmic}
\end{algorithm}

\indent For each contiguous slot set $\widetilde{ss}$ ranging from $SI$ to $SI + d^j - 1$ that is available on source ouput fiber $a$ of path $k$, the spectrum cost is defined in Equation (\ref{equ2}).
\begin{equation}\label{equ2}
C_{a,\widetilde{ss},f}^k=\left\{
\begin{array}{l l}
1,  \ \ \ \ \ \mbox{if \ $SI + d^j {\leq} m_{f}$}\\
\infty, \ \ \mbox{if \ $\widetilde{ss}$ \ is \ not \ available \ on \ $f$} \\
(SI + d^j - m_{f})/\Omega + 1,  \ \ \ \mbox{otherwise}
\end{array}
\right.
\end{equation}

\noindent
The basic idea of Spectrum Cost Update (Equation (\ref{equ2})) is to try to not increase MSU on each fiber after the request is established. We denote $m_{f}$ as the largest FS index on fiber $f$ and $C_{a,\widetilde{ss},f}^k$ as the spectrum cost on fiber $f$ for the slot set $\widetilde{ss}$ in $AG_{a}^k$. If $\widetilde{ss}$ is not available on fiber $f$, $C_{a,\widetilde{ss},f}^k$ is set to $\infty$. This means that we cannot establish a lightpath by using the slot set $\widetilde{ss}$ on fiber $f$. If $\widetilde{ss}$ is available and the ending FS index is no more than $m_{f}$, $C_{a,\widetilde{ss},f}^k$ is set to 1. In this case, the local MSU on fiber $f$ would not increase, and the spectrum cost works as a hop count. Otherwise, $C_{a,\widetilde{ss},f}^k$ is set to be $(SI + d^j - m_{f})/\Omega + 1$, which is proportional to the number of slots by which the local MSU increases on that fiber, while taking the hop count into consideration.

\indent For a lightpath $l$ in the $k$th path, we denote $C_{a,\widetilde{ss}}^k$ as the total cost combining Spectrum Costs $\sum_{f}C_{a,\widetilde{ss},f}^k$ of the fiber set $F_l$ and Switching Costs $\sum_{v}C_{v,f_{in},f_{out}}$ of the node set $V_l$ which the lightpath is going through. Once the costs of all edges in the layered graph are determined, we use Dijkstra's algorithm to find the shortest paths from one output fiber of the source node to different input fibers of the destination node, and choose the valid lightpath $p$ with smallest path cost. In order to decrease the algorithm's complexity, we stop generating further layered graphs for this auxiliary graph, and consider $p$ as a {\em candidate lightpath}. We compare the path costs of all the candidate lightpaths and choose a path with the smallest total cost.    

\indent Suppose the total number of links in the network is $L$, the maximum number of fibers per link is $\varepsilon$, and $\Omega$ is an estimated upper bound of the MSU. For each request, the number of hops in the path of a mesh network, denoted as $h$, is usually small compared to $L$. For each output fiber of the source node, $O(\Omega)$ auxiliary graphs are created. Since each path is precalculated, each auxiliary graph consists of $O(h\varepsilon)$ vertex and $O(h\varepsilon^2)$ edges ($O(\varepsilon^2)$ first type of edges on each intermediate node of the path and $O(\varepsilon)$ second type of edges on each link of the path). The dominant part is to find the shortest path in the auxiliary graphs, thus the time complexity of the algorithm is $O(\Omega h\varepsilon^3\log(h\varepsilon))$, which is much less than the time complexity of our previous algorithm $O(\Omega L\varepsilon^3\log(L\varepsilon))$, that generates the auxiliary graphs based on the whole network.

\subsection{Node Placement Schemes}

\indent Given only the budget, i.e., the total number of available small-port-count WSSs for the network, the number and location of FLEX nodes should be determined. Assume all nodes are initially CONV nodes; by replacing CONV nodes with FLEX nodes, the required number of WSSs will be decreased. Then the problem transforms to how many and which CONV nodes should be replaced to meet the budget. Of course, the performance of the network after replacement should be taken into consideration when making the replacements.

\indent Indeed, how many CONV nodes need to be replaced also depends on where they are. Replacing a node with a larger port count may have the same reduction in number of WSSs as replacing two nodes with small port counts. Therefore, we divide the joint placement and assignment problem into two stages. The first stage is to decide the location of FLEX nodes. Then a routing path, fibers on links along the path, waveband, and a slot set are selected for each request by our {\em RFBSA} algorithm.

\subsubsection{Random Node Placement}
\indent We use a random node placement scheme as the baseline. We first try to randomly select a CONV node to replace. The number of required WSSs is reduced accordingly. The procedure of randomly replacing the remaining CONV nodes and updating the required number of WSSs is performed repeatedly until the budget constraint is met.

\subsubsection{Traffic-aware Node Placement}
\indent For a better node placement scheme, the network topology and the traffic requests should be taken into consideration. The first intuition is that if fewer CONV nodes are replaced with FLEX, the performance degradation would be less. In this case, we would like to replace the physical nodes with larger port count first. The physical nodes are sorted by their port counts in descending order and replaced one by one to meet the budget.

\indent On the other hand, physical nodes with large port counts usually have more traffic going through them. Replacing these nodes might affect more connection requests and add routing constraints to them, leading to worse performance. Accordingly, the node placement should consider the traffic pattern, as well as the port counts of physical nodes. 

\indent We assume all nodes are CONVs initially, and the {\em RFBSA} is applied to accommodate all requests in the network. To facilitate node placement, the band usage information and the total number of requests going through each physical node are recorded. A cost function related to the traffic for each node $v$ is defined in Equation (\ref{equ3}). Let $maxb_{v}$ denote the maximum number of bands originated from all input fibers of node $v$. $maxb_{v}$ is related to the waveband usage on each node. $\Delta_{v}$ denotes the number of traffic demands going through node $v$ (excluding requests originating or terminating at this node). Let $D_{v}$ denote the port count of node $v$. The larger $D_{v}$ is, the more reduction in required number of WSSs that is achieved.  When $maxb_{v}$ is not larger than the band limit $B$, we set the cost to 0, otherwise we set the cost to the traffic density ($\Delta_{v} / D_{v}$) that would be affected by the limited banding. The larger $C_{v}$ a node has, the more traffic it will affect when replaced by FLEX. Nodes with a small cost $C_{v}$ are preferred to be replaced. We sort the nodes in ascending order of the costs $C_{v}$, and replace them one by one until the budget constraint is met.
\begin{equation}\label{equ3}
C_{v} = \left\{
\begin{array}{l l}
0,   \ \ \mbox{if $maxb_{v} {\leq} B$}\\
\Delta_{v} / D_{v},  \ \ \ \mbox{otherwise}
\end{array}
\right.
\end{equation}

\indent When all the FLEX nodes are placed, the route, fiber, band and slot assignment for each request should be determined in the mixed network. The {\em RFBSA} algorithm described in the previous subsection is utilized to assign the route, fibers, wavebands and a contiguous slot set to accommodate each request. We finally select the node placement by comparing the performance of the above two strategies for each budget.

\subsection{Extension to Dynamic Instances}
\indent To accommodate dynamically arriving traffic requests, we should determine the node placements first. Given the information on network topology and traffic (request size distribution and traffic pattern), we first generate a large number of static requests according to the given traffic pattern, and use the Traffic Aware Placement introduced in the previous subsection to find appropriate FLEX node deployments. Then the {\em RFBSA} algorithm is applied to the mixed network for the arriving requests. In the dynamic case, the goal is to minimize demand blocking ratio. We adapt the spectrum cost function in the {\em RFBSA} algorithm to achieve good performance. $C_{a,\widetilde{ss},a}^k = SI$ denotes the spectrum cost on one outgoing fiber $a$ from the source node for its available slot set $\widetilde{ss}$ and $C_{a,\widetilde{ss},f}^k = 1$ if $\widetilde{ss}$ is available on fiber $f$ along the path $k$. Here we try to compact the total utilized network resources to leave more resources for later requests. The cost function emphasizes both starting slot and hop count of the path.

\section{Simulation Results}\label{sec.label64}

\indent In this section, we present performance evaluation results for both static and dynamic instances to demonstrate the effectiveness of our proposed method. ILP results for small network instances and simulation results for large network instances are generated to examine the effects of changing budget on the performance of the network. 

\indent Given a set of connection requests and a budget for the network, the number and locations of FLEX nodes should be determined, and networking resources need to be allocated to the requests. Each request represents a connection from a source node to a destination node with a demand size requirement. A uniform traffic pattern means the source/destination nodes are randomly selected from the physical nodes of the network. We assume three types of demands in terms of the number of required FSs~\cite{hasegawa2015flexible}, with the following distribution: 3 slots (40 Gbps) with probability 0.2, 4 slots (100 Gbps) with probability 0.5, and 7 slots (400 Gbps) with probability 0.3. For the FLEX nodes, the maximum number of wavebands $B$ is assumed to be 4 or 9.

\indent For a given set of static traffic requests, we record the total MSU. The average MSU over the total number of fibers is used as the performance measure. The parameter $\alpha$ is set to 1 for RFBSA.\footnote{This was chosen based on observed performance for a large set of parameters.} In the node selection scheme, both band usage and the traffic distribution are considered. RP is used to denote the random node placement strategy, while TAP represents the traffic aware node placement strategy in the following results.

\subsection{Results for a Small Network}
\indent We get single-run ILP results for a 5-node small network as shown in Figure \ref{fig:NPsmallnet}. The number of parallel fibers, $x$, on each link is randomly distributed between 2 and 3 ($x=[2,3]$ per link). The band limit $B = 4$.

\begin{figure}
	\centering
	\includegraphics[scale=0.6]{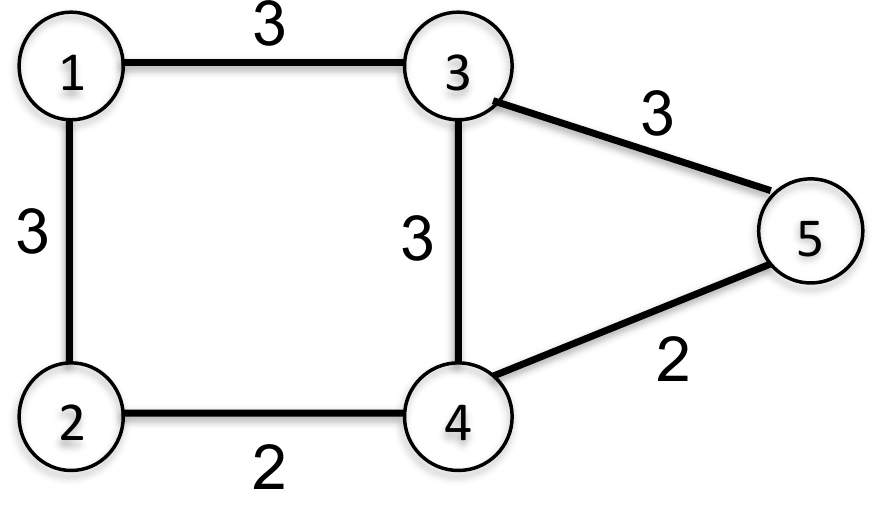}
	\caption{\label{fig:NPsmallnet} Small network. The numbers are the number of fibers on the links.}
\end{figure}

\indent Table \ref{NPILP} shows the average MSUs over all fibers achieved by the ILP and our proposed heuristics for different numbers of requests in the 5-node small network. The given budget is 89 WSSs (which is between the required number of WSSs for all FLEX nodes - 64, and all CONV nodes - 125). We only simulate the $K=1$ case here due to the complexity of the ILP. For example, the execution time of the heuristics is a few seconds, while that of ILP is more than 6 hours. We can see that the results from TAP are better than for RP, and are not far from the ILP's results.

\renewcommand\arraystretch{1.1}
\begin{table}
	\caption{\label{NPILP} Average MSU results for the small network.}
	\centering
	\begin{tabular}{|c|c|c|c|}
		\hline
		\bf{Number of Requests} & RP & TAP & ILP\\ \hline
		\bf{40} & 10.25 & 9.5 & 9.03125 \\ \hline
		\bf{50} & 13.53125 & 12.46875 & 12.125 \\ \hline
		\bf{60} & 15.8125 & 14.96875 & 14 \\ \hline
		\bf{100} & 24.75 & 23.0625 & 22.21875\\ \hline
	\end{tabular}
\end{table}

\subsection{Results for Larger Topologies}

\indent The larger topologies used for simulations are the NSFNET network and pan-European network. The NSFNET as shown in Figure \ref{fig:NSFNetwork-Detail} has 14 physical nodes and 22 bidirectional links \cite{wu2015comparison}, while the pan-European network as shown in Figure \ref{fig:pan} consists of 28 physical nodes and 43 bidirectional links \cite{wu2016optimal}. The number of fibers, $x$, on each link is randomly distributed between 3 and 5 ($x=[3,5]$ per link). The network topology is fixed for every run of the simulation.

\begin{figure}
	\centering
	\includegraphics[scale=0.4]{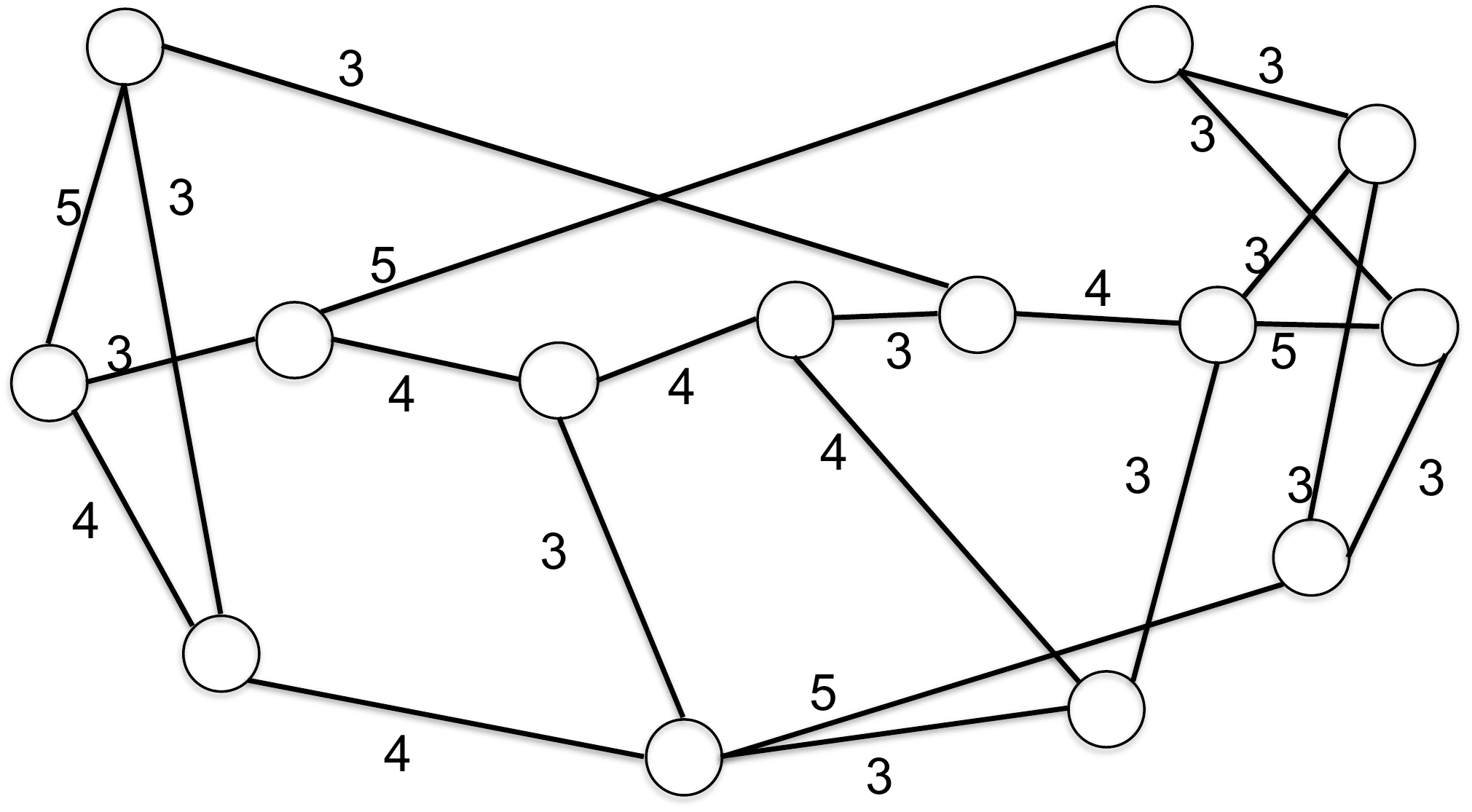}
	\caption{\label{fig:NSFNetwork-Detail} NSF network. The numbers are the number of fibers on the links.}
\end{figure}

\begin{figure}
	\centering
	\includegraphics[scale=0.9]{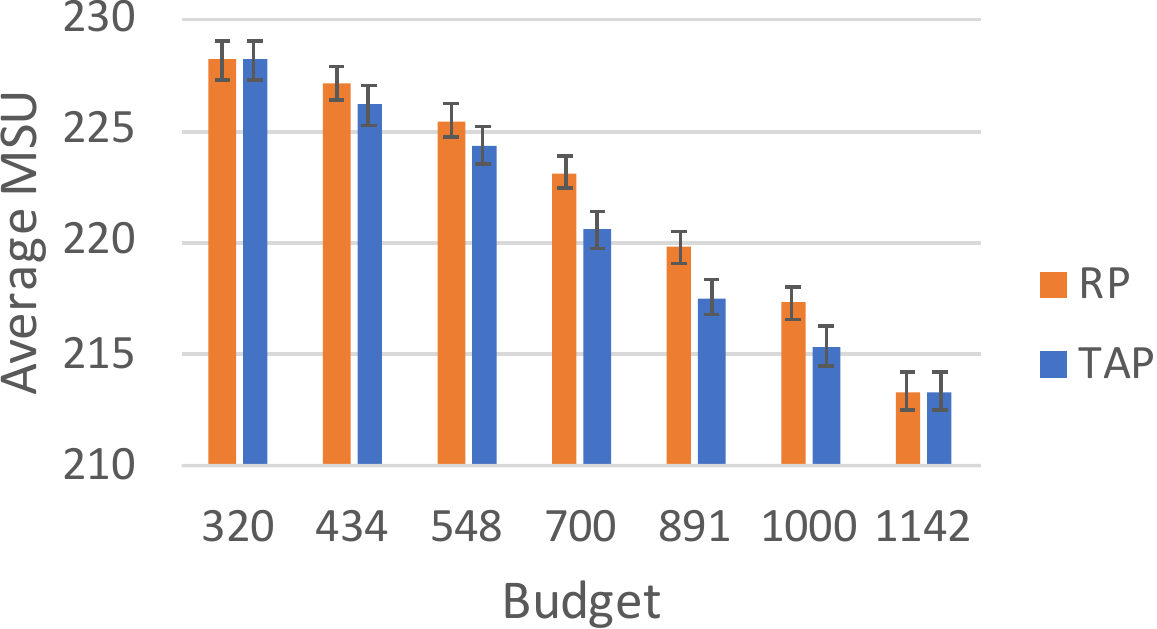}
	\caption{\label{fig:compk1} Network performance vs. budget for NSF network ($B$=4).}
\end{figure}

\begin{figure}
	\centering
	\includegraphics[scale=0.92]{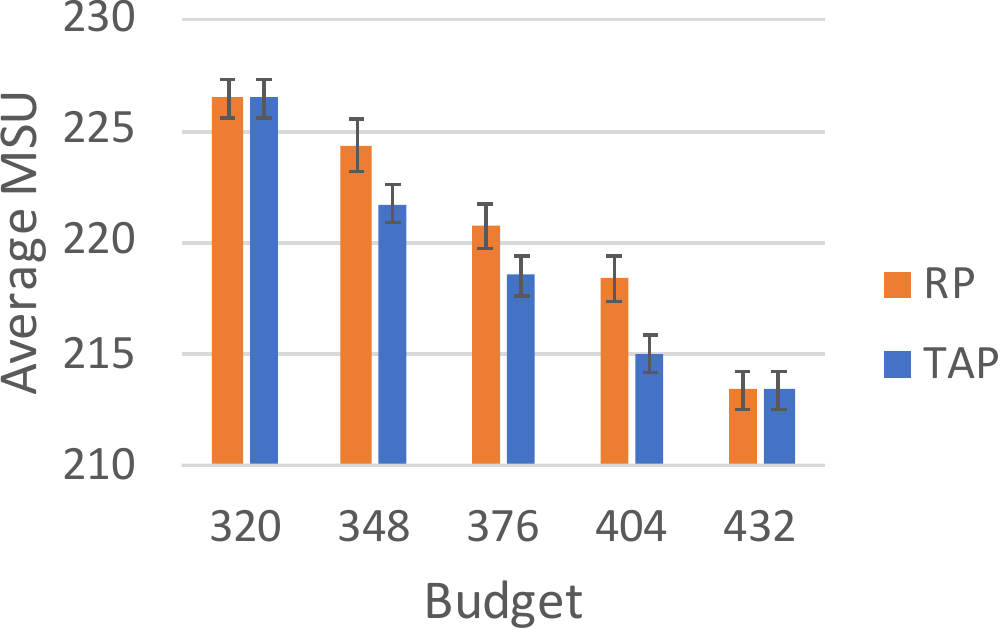}
	\caption{\label{fig:compk2} Network performance vs. budget for NSF network ($B$=9).}
\end{figure}

\renewcommand\arraystretch{1.1}
\begin{table}
	\caption{\label{flexno} Average number of FLEX nodes for NSF network ($B$=4).}
	\centering
	\begin{tabular}{|c|*{7}{c|}}
		\hline
		\backslashbox[0.5em]{Schemes}{Budget}
		&\makebox[0.6em]{320}&\makebox[0.6em]{434}&\makebox[0.6em]{548}&\makebox[0.6em]{700}&\makebox[0.6em]{891}&\makebox[0.6em]{1000}&\makebox[1.5em]{1142}\\\hline
		\bf{TAP} & 14 & 11.66 &	9.52 &	6.36 &	3.98 &	2.68 & 0 \\ \hline
		\bf{RP} & 14 & 12.436 &	10.52 &	8.03 &	4.78 &	2.838 & 0 \\ \hline
	\end{tabular}
\end{table}

\indent Figures \ref{fig:compk1} and \ref{fig:compk2} show the performance change in the NSF network as the given budget changes. We conducted 50 trials of simulation, each consisting of 3000 traffic requests, and show the average results with 95$\%$ confidence intervals. Figure \ref{fig:compk1} shows the comparison results when $B=4$, while Figure \ref{fig:compk2} is for $B=9$. We use $K=3$ for all simulations. The range of the budget (number of available WSSs) for the network varies from 320 (all nodes are FLEX) to 1142 (all nodes are CONV) when $B=4$, and to 432 when $B=9$. Table \ref{flexno} shows the illustrative results depicting the average number of FLEX nodes for each budget in the 14-node NSF network when $B = 4$.

\begin{figure}
	\centering
	\includegraphics[scale=0.9]{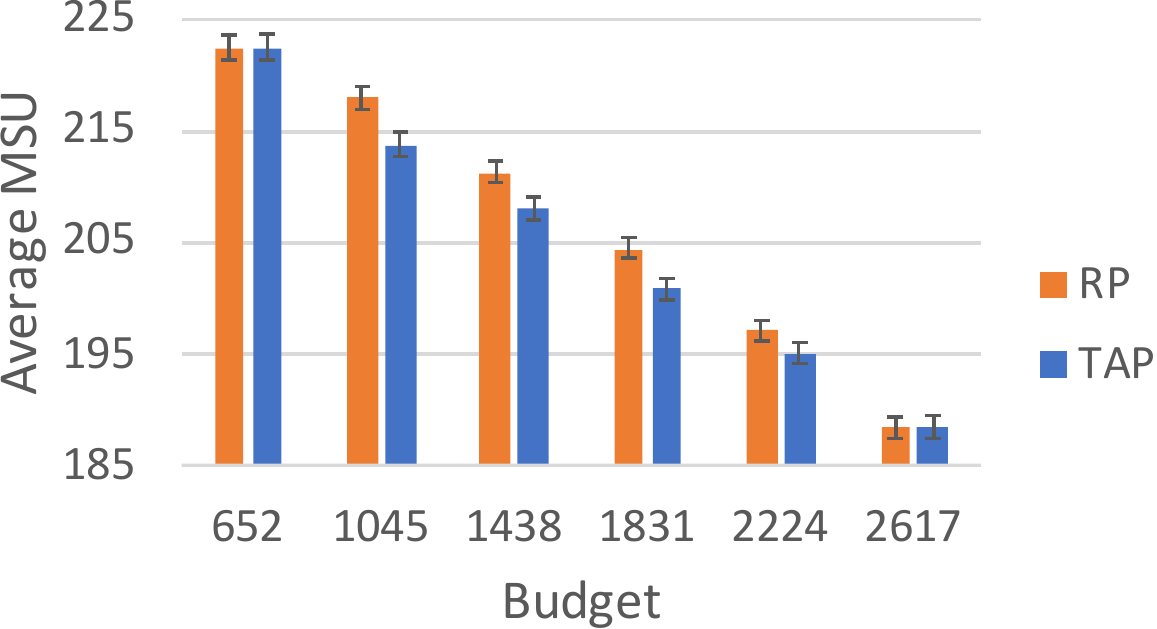}
	\caption{\label{fig:compk3} Network performance vs. budget for pan-European network ($B$=4).}
\end{figure}

\begin{figure}
	\centering
	\includegraphics[scale=0.92]{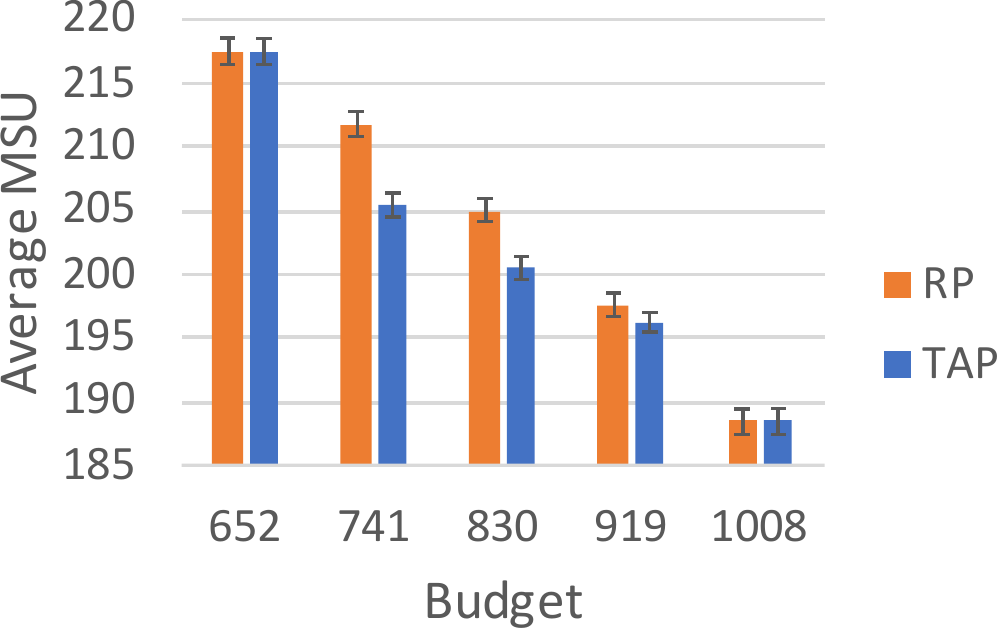}
	\caption{\label{fig:compk4} Network performance vs. budget for pan-European network ($B$=9).}
\end{figure}

\indent Similarly, Figures \ref{fig:compk3} and \ref{fig:compk4} show the performance change in the pan-European network as the given budget changes. We use similar settings for the simulations. Figure \ref{fig:compk3} shows the comparison results when $B=4$, while Figure \ref{fig:compk4} is for $B=9$. The range of the budget (number of available WSSs) for the network varies from 652 (all nodes are FLEX) to 2617 (all nodes are CONV) when $B=4$, and to 1008 when $B=9$.

\indent As expected, the average MSU increases, which means network performance degrades, as the budget decreases (more FLEX nodes and fewer CONV nodes). The comparison of the two node placement strategies is shown in the figures. When all nodes are FLEX or CONV, there is no difference between RP and TAP as the network is the same for both policies. In general, TAP performs better than RP in the sense that fewer slots are required to accommodate all requests. However, even in the worst case, the degradation of network performance is relatively small, which means that it may be possible to replace most CONV nodes with FLEX nodes in the network without a significant penalty. We note that the MSU is increased by $10\%$ to $15\%$ for the two topologies when all CONV nodes are replaced by FLEX nodes. However, the number of WSSs needed can be reduced by a factor of 2 or more in many cases. These results suggest that the slight increase in spectrum cost is more than offset by the savings in WSSs.

\subsection{Results for Dynamic Requests}

\indent The traffic demand is a set of dynamically arriving connection requests. Connection requests arrive to the network according to a Poisson process. Each request has a mean holding time of 1 (arbitrary time unit), and the arrival rate of traffic requests is varied in order to examine the network performance under varying offered loads (denoted by $L$). We use the demand blocking ratio of dynamic traffic requests to indicate the network performance. The parameter $\alpha$ is set to 100\footnote{This was chosen based on observed performance for a large set of parameters.} for {\em RFBSA}. For each simulation, the results of 200,000 dynamic requests excluding 10,000 warm-up requests are recorded.

\indent The NSNET topology is used for this set of simulations. Each link has a random number of fibers, $x$, that are uniformly distributed between 5 and 10 fibers ($x=[5,10]$ per link). We assume the fiber capacity of 352 frequency slots, with each slot having a bandwidth of 12.5 GHz. The banding limit of FLEX nodes is assumed to be 4 (i.e., $B = 4$). The range of the budget for the network varies from 660 (all nodes are FLEX) to 5137 (all nodes are CONV).

\indent We first show a set of simulation results for the uniform traffic pattern. A set of static connection requests is first generated according to this traffic pattern in order to determine the node placement to satisfy a given budget of the network. Then the demand blocking ratio of dynamic requests is evaluated. Figures \ref{fig:u6900} and \ref{fig:u7400} show the performance change in the NSF network as the given budget changes for different traffic loads. We can see that the network performance degrades (demand blocking ratio increases) as the budget decreases. In general, TAP performs better than RP (lower demand blocking ratio under each budget). When the traffic load is small, the difference between TAP and RP is more obvious. 

\begin{figure}
	\begin{tabular}{cc}
		\begin{minipage}[t]{2.8in}
			\centering
			\includegraphics[scale=0.4]{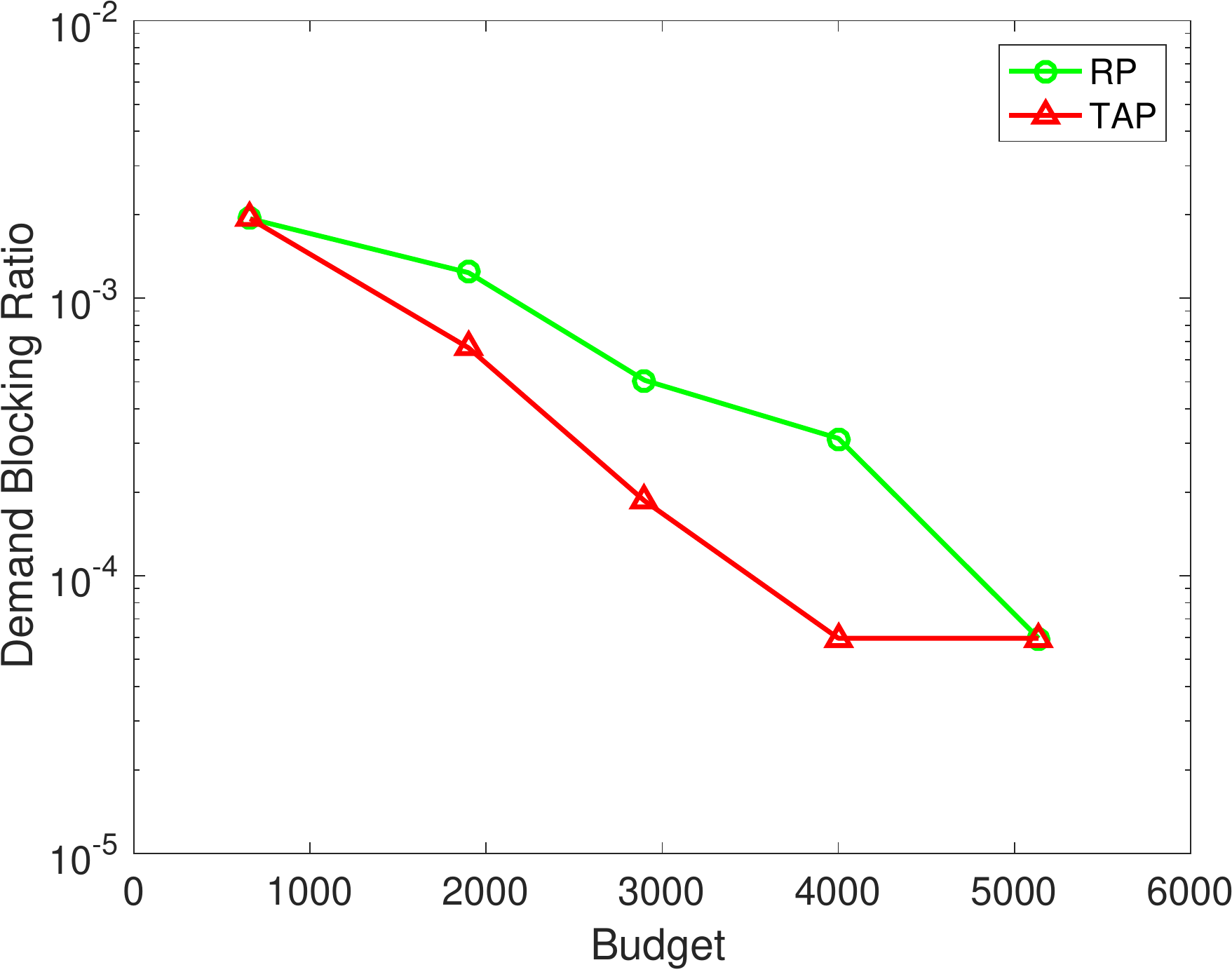}
			\caption{\label{fig:u6900} Network performance vs. budget for NSF network with uniform traffic pattern ($L$=6900).}
		\end{minipage}
		\hspace{0.1in}
		\begin{minipage}[t]{2.8in}
			\centering
			\includegraphics[scale=0.4]{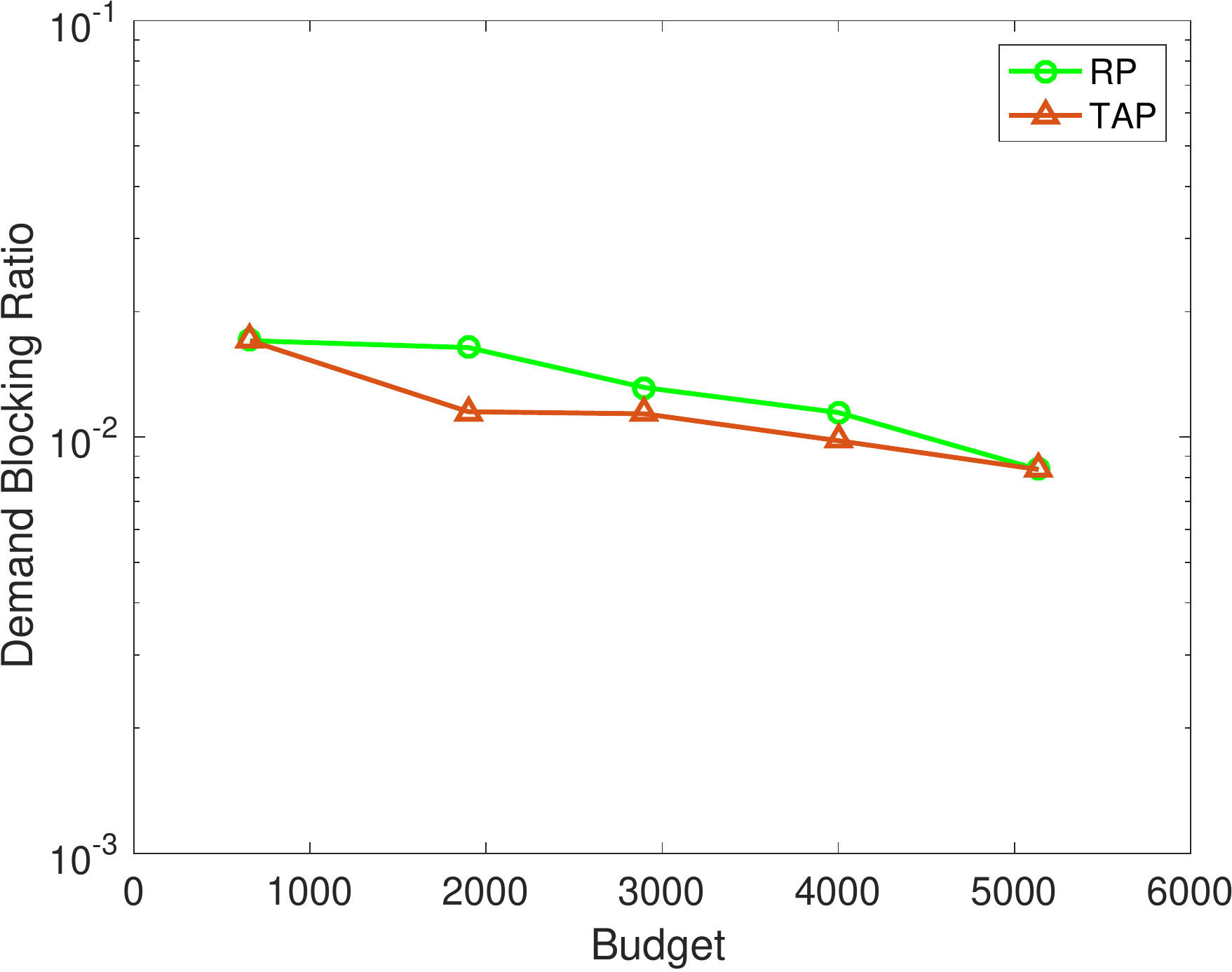}
			\caption{\label{fig:u7400} Network performance vs. budget for NSF network with uniform traffic pattern ($L$=7400).}
		\end{minipage}
	\end{tabular}
\end{figure}

\indent We also conduct simulations for a non-uniform traffic pattern. We assume that nodes with higher connectivity have larger opportunity to send and receive traffic, and choose the probability that a node $v$ is selected as a source or destination, $u_v$, in proportion to the node's physical degree. Again, a set of static connection requests of this traffic pattern is generated in order to determine the node placement for each budget. The comparison results for different network loads are shown in Figures \ref{fig:nu7100} and \ref{fig:nu7400}. We can see similar trends as in previous results.

\begin{figure}
	\begin{tabular}{cc}
		\begin{minipage}[t]{2.8in}
			\centering
			\includegraphics[scale=0.4]{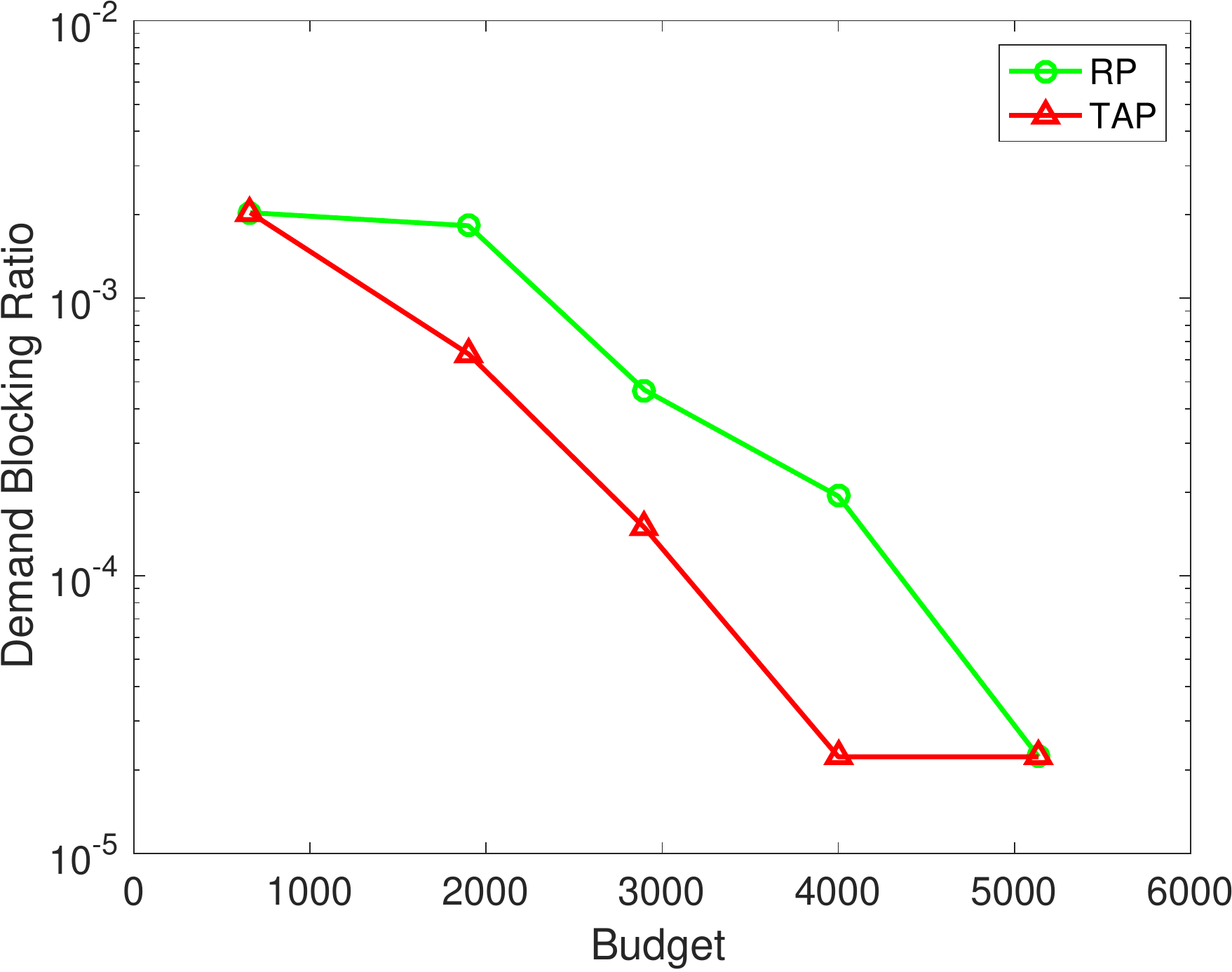}
			\caption{\label{fig:nu7100} Network performance vs. budget for NSF network with non-uniform traffic pattern ($L$=7100).}
		\end{minipage}
		\hspace{0.1in}
		\begin{minipage}[t]{2.8in}
			\centering
			\includegraphics[scale=0.4]{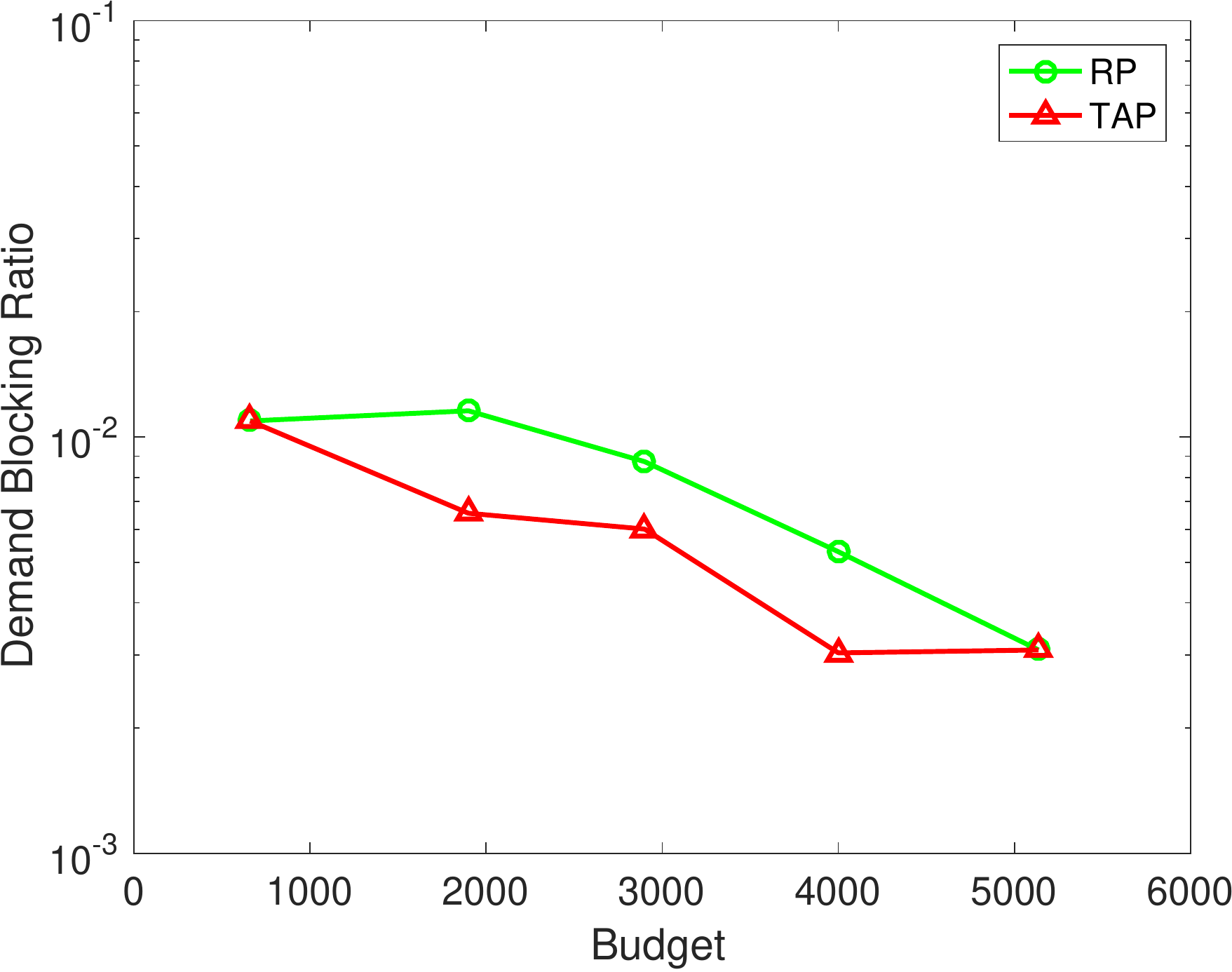}
			\caption{\label{fig:nu7400} Network performance vs. budget for NSF network with non-uniform traffic pattern ($L$=7400).}
		\end{minipage}
	\end{tabular}
\end{figure}

\begin{figure}
	\centering
	\includegraphics[scale=0.8]{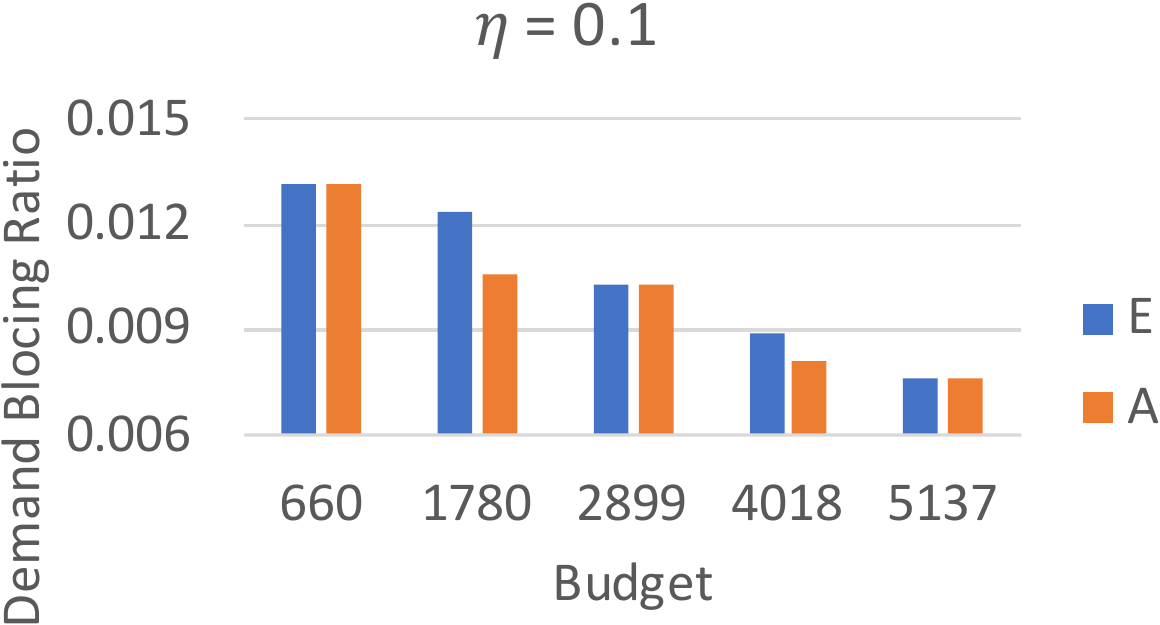}
	\caption{\label{fig:ste1} Sensitivity test for traffic pattern with a perturbation of $0.1$.}
\end{figure}

\indent We also conduct sensitivity tests to see how well our TAP scheme performs when the actual traffic pattern is slightly different from what the network was planned for. Here, we perturb the source/destination selection probabilities $u_v$ by a factor $\eta_v$ ($\eta_v$ could be either positive or negative, with the average of $|\eta_v|$ being $\eta$). In other words, we set the node selection probabilities to be ($u_1(1 + \eta_1), u_2(1 + \eta_2), \cdots, 1 - \sum_{v = 1}^{N-1}{u_v}(1+\eta_v)$). Then, we compare the performance obtained by applying the original node placement to the changed traffic pattern with the performance obtained by applying the node placement for the changed traffic pattern. Figure \ref{fig:ste1} shows results for $\eta = 0.1$. The results are represented by ``A" (with placement based on exact information) and ``E" (without exact information). There is no performance difference between ``E" and ``A" when all nodes are FLEX or CONV because the node placements are the same. For other budgets, the difference in performance between ``E" and ``A" is relatively small, indicating that the algorithm is not sensitive to small changes in traffic pattern. 

\section{Conclusion}\label{sec.label65}
\indent Flexible waveband OXCs require much less hardware cost than the conventional OXCs, with the penalty of some switching constraints. In this work, we jointly consider RFBSA and FLEX node placement to satisfy a network planning budget in terms of the total number of available WSSs \cite{wu2018joint}. In addition to an integer linear programming formulation, we present node placement schemes and extend the cost-function-pluggable auxiliary layered-graph framework in our previous work to solve this problem. The simulation results demonstrate that our heuristic solution saves network resources and achieves good network performance indicated by the average maximum spectrum usage. The framework is also demonstrated to have good performance for dynamic traffic requests.

\chapter{Dynamic Routing and Spectrum Assignment for Multi-fiber Elastic Optical Networks}
\label{chap_7} 

\indent In this chapter, we study dynamic Routing and Spectrum Assignment problem in elastic optical networks with multiple fibers per link to further improve network performance. 
We use the demand blocking ratio for dynamically arriving requests as the performance indicator. 

\section{Related Work}\label{sec.label71}

\indent The first comprehensive study on the RSA problem is \cite{wang2011study}, which formally defined the problem and proved its NP-hardness. In order to solve the RSA problem efficiently, many heuristics have been proposed in the literature. Dealing with routing and spectrum allocation problems jointly usually requires high complexity \cite{christodoulopoulos2010routing, wu2017routing,wu2017joint}. Many of the literature works solve this problem by decomposing it and solving the two subproblems in sequence. The shortest path with maximum spectrum reuse algorithm and balanced load spectrum allocation algorithm (which determines the routing by balancing the load in the network) are proposed in \cite{wang2011study} to solve the static RSA problem. In \cite{shirazipourazad2013routing}, given a set of traffic requests, the authors try to find disjoint paths to route requests in order to increase slot reuse in SA. For the dynamic RSA problem, bandwidth fragmentation and spectrum misalignment caused by dynamic set up and tear down of traffic requests hurt the network performance. Many RSA schemes have been proposed in order to overcome the bandwidth fragmentation issue \cite{jinno2010distance, christodoulopoulos2011elastic, zhang2014dynamic, sone2011routing, zhu2013dynamic, chen2015fragmentation, Hsu2016}. Defragmentation algorithms which reroute connections are developed in \cite{patel2012routing, zhang2013bandwidth}. Another scheme to eliminate bandwidth fragmentation without rerouting connections is to partition the spectrum for heterogeneous bandwidth demands. In \cite{wang2012spectrum, wang2014spectrum}, different partition schemes are investigated, and the well-known First Fit spectrum allocation is used. In \cite{fadini2014subcarrier}, the spectrum is partitioned by classifying connection groups. All the above RSA schemes are proposed for EONs with a single fiber per link.

\indent To accommodate increasing traffic demands, deploying multiple fibers on a physical link is needed. The wavelength / slot assignment in networks with multi-fiber links differ from prior schemes of single-fiber networks in that multi-fiber links provide more flexibility in switching wavelengths / frequency slots. It's more complicated to design WA / SA algorithms to fully utilize this flexibility. The WA problem in multi-fiber WDM networks has been explored in the literature \cite{subramaniam1997wavelength, zhang1998wavelength, xu2000wavelength, xu2000dynamic, li2000wavelength, liu2009routing, coiro2011power}. \cite{subramaniam1997wavelength, zhang1998wavelength} take the network state information into consideration and show good performance. Different cost functions based on the network state are utilized in the literature. Due to the heterogeneity of demands in EONs, which fiber to use on each link matters and will cause unnecessary fragmentation if not carefully addressed. In order to fully utilize the flexibility in multi-fiber links while satisfying constraints introduced by EONs, an efficient RSA scheme is necessary. 

\section{Contributions} \label{sec.label72}
\indent In this study, we propose a novel and efficient solution for the dynamic RSA problem in multi-fiber elastic optical networks. Each link in the network contains multiple fibers. A network planning formulation based on the topology information is proposed so that the candidate paths for each source-destination node pair can be selected according to certain predetermined probabilities. Due to the heterogeneous request bandwidths in the network, a partition scheme is applied so that each request size can use a particular spectrum range. In this case, each fiber can be viewed the same and fiber selection on each link can also be avoided. Given the spectrum partitioning, an SA algorithm based on both the network state information and the path selection probabilities is proposed. The Next State Aware SA algorithm is further improved by considering the resource sharing among different partitions. For each arriving request, a routing path is first selected by the precomputed probabilities, then the next state aware SA is performed to assign contiguous FSs to that request. If there are no available FSs, the request will be blocked. We use the demand blocking ratio (ratio of the sum of bandwidths of blocked requests to total bandwidths of all requests) to indicate network performance.

\indent Our contributions are: (1) we propose a multi-path selection method for the routing problem; (2) we utilize a spectrum management scheme to alleviate fragmentation caused by heterogeneous request bandwidths; (3) we adopt a spectrum assignment method that optimizes the state of network after the assignment. We demonstrate that each scheme is effective in improving the spectrum efficiency and network performance.

\section{Background and Motivatoin}\label{sec.label73} 
\indent Consider a network $\mathcal{G}=(\mathcal{V},\mathcal{E})$, where $\mathcal{V}$ denotes a set of optical cross-connects (OXCs), and $\mathcal{E}$ denotes a set of physical links. We assume that end nodes are connected to each OXC so that every OXC can be the origin or destination of connection requests. Each link contains multiple fibers, and the number of fibers on each link may be different. The spectrum resource on each fiber is carved up into frequency slots, with bandwidth of $12.5$ GHz each. All fibers consist of the same number of FSs. At each OXC, assume that there are no switching constraints from input fibers to output fibers. A FS on an input fiber can be switched to the {\em same} FS on any output fiber. 

\indent A connection request to the network is an end-to-end lightpath with a source, a destination node, and a data rate requirement. The number of FSs assigned to a request depends on the data rate requirement and the modulation format. In this work, the same modulation format is assumed. Thus the bandwidth demand for each connection only depends on the data rate requirement. In general, a connection request can be represented by a three-tuple $(source, destination, b)$: a source node, a destination node, and a request size in terms of number of contiguous FSs. Heterogeneous traffic requests with different data rate requirements will consume different number of FSs, which cause the fragmentation issue. 

\indent The routing and spectrum assignment (RSA) problem in multi-fiber EONs is to find a path and a set of contiguous FSs on {\em some} fiber on links along the path for an arriving request. 

\indent Spectrum management such as partitioning the spectrum into dedicated ranges is motivated by the heterogeneous request sizes. Let us take a look at a simple example with request sizes of 2 and 3 slots. If there is no spectrum management based on request classification and First Fit Spectrum/Slot Assignment algorithm is used, for the requests, the network status is as shown in Figure \ref{fig:frag}. When a new request of size 3 from node 2 to node 4 arrives, there are no available FSs and this request will be blocked. If we do spectrum management, such as dedicating the first 4 slots for requests of size 2 and the last 3 slots for requests of size 3, for the same set of requests, the slot assignments will change and the new arriving request can be accommodated. Figure \ref{fig:specmanage} shows the new assignments when utilizing spectrum management.
        
\begin{figure}
	\centering
	\begin{tikzpicture}[yscale=0.58,
	roundnode/.style={circle, minimum width=8mm,draw=blue!80!black,fill=white!30, very thick,inner sep=0pt}, slotnode/.style={rectangle, minimum width=2.5mm, minimum height=4mm,draw=blue!80!black,fill=white!30,inner sep=0pt},
	cross/.style={path picture={ \draw[black, very thick](path picture bounding box.south east) -- (path picture bounding box.north west) (path picture bounding box.south west) -- (path picture bounding box.north east);}}]
	\node[roundnode] (n1) at (0,0) {1};
	\node[roundnode] (n2) at (2.5,0) {2};
	\node[roundnode] (n3) at (5,0) {3};
	\node[roundnode] (n4) at (7.5,0) {4};
	\path[draw, very thick] (n1) -- (n2) -- (n3) -- (n4);
	
	\node[slotnode,fill=red] (1) at (0.5,1.5){};
	\foreach \i/\j in {2/1,3/2}
	\node[slotnode,fill=red] (\i) [right=0 of \j] {};
	\node[slotnode,fill=green] (4) [right=0 of 3] {};
	\node[slotnode,fill=green] (5) [right=0 of 4] {};
	\foreach \i/\j in {6/5,7/6}
	\node[slotnode] (\i) [right=0 of \j] {};
	
	\node[slotnode,fill=yellow] (8) at (3,1.5){};
	\node[slotnode,fill=yellow] (9) [right=0 of 8] {};
	\node[slotnode] (10) [right=0 of 9] {};
	\node[slotnode,fill=green] (11) [right=0 of 10] {};
	\node[slotnode,fill=green] (12) [right=0 of 11] {};
	\node[slotnode] (13) [right=0 of 12] {};
	\node[slotnode] (14) [right=0 of 13] {};
	
	\node[minimum width=20mm, minimum height=2mm] at (5.9, 3.5){New connection blocked};
	\node[minimum width=2mm, minimum height=3mm, cross] at (4.3, 2.5){};
	\node[slotnode] (15) at (5.5,1.5){};
	\foreach \i/\j in {16/15,17/16,18/17,19/18,20/19,21/20}
	\node[slotnode] (\i) [right=0 of \j] {};
	\end{tikzpicture}
	\caption{\label{fig:frag}Slot usages without spectrum management.}
\end{figure}
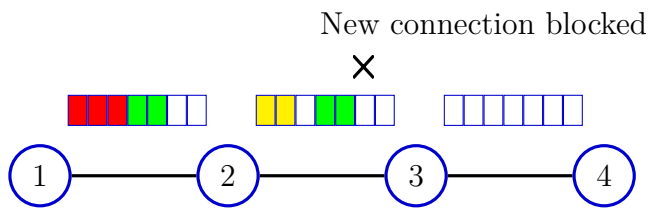

\begin{figure}
	\centering
	\begin{tikzpicture}[yscale=0.58,
	roundnode/.style={circle, minimum width=8mm,draw=blue!80!black,fill=white!30, very thick,inner sep=0pt}, slotnode/.style={rectangle, minimum width=2.5mm, minimum height=4mm,draw=blue!80!black,fill=white!30,inner sep=0pt}]
	\node[roundnode] (n1) at (0,0) {1};
	\node[roundnode] (n2) at (2.5,0) {2};
	\node[roundnode] (n3) at (5,0) {3};
	\node[roundnode] (n4) at (7.5,0) {4};
	\path[draw, very thick] (n1) -- (n2) -- (n3) -- (n4);
	
	\node[slotnode,fill=green] (1) at (0.5,1.5){};
	\node[slotnode,fill=green] (2) [right=0 of 1] {};
	\foreach \i/\j in {3/2,4/3}
	\node[slotnode] (\i) [right=0 of \j] {};
	\foreach \i/\j in {5/4,6/5,7/6}
	\node[slotnode,fill=red] (\i) [right=0 of \j] {};
	
	\node[slotnode,fill=green] (8) at (3,1.5){};
	\node[slotnode,fill=green] (9) [right=0 of 8] {};
	\node[slotnode,fill=yellow] (10) [right=0 of 9]{};
	\node[slotnode,fill=yellow] (11) [right=0 of 10] {};
	\foreach \i/\j in {12/11,13/12,14/13}
	\node[slotnode,fill={rgb,255:red,105; green,70; blue,80}] (\i) [right=0 of \j] {};
	
	\node[slotnode] (15) at (5.5,1.5){};
	\foreach \i/\j in {16/15,17/16,18/17}
	\node[slotnode] (\i) [right=0 of \j] {};
	\foreach \i/\j in {19/18,20/19,21/20}
	\node[slotnode,fill={rgb,255:red,105; green,70; blue,80}] (\i) [right=0 of \j] {};
	
	\end{tikzpicture}
	\caption{\label{fig:specmanage}Slot usages with spectrum management.}
\end{figure}
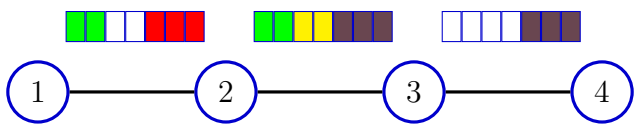

\section{Multi-Path Selection Scheme}\label{sec:Path}

\indent A path between the source and destination nodes should be determined to accommodate the request. Commonly used routing schemes such as using a precomputed fixed single shortest path or one of many paths selected dynamically (based on the network state each time a request arrives) suffer from either poor performance (in the former case) or high computation complexity (in the latter case). Thus we would like to fully utilize multiple paths between each node pair efficiently. We intend to select a routing path based on some predetermined rules, instead of calculating the best path each time for a new arrival request.

\indent Given the traffic loads, we propose to compute a set of candidate paths and the path selection probabilities for each node pair {\em offline}. This approach fully utilizes the multiple paths between each node pair while keeping computational complexity quite low. Given the network topology and traffic loads for each node pair, the selection probabilities for all candidate paths are computed offline via a Mixed Integer Linear Program to span loads over fibers. The result of this is that we obtain the probability distribution of the path to be selected for a request from $s$ to $d$, $p_{(s,d)}^k$, the probability that the $k^{\rm th}$ path should be selected for node pair $(s,d)$. Each time when a request arrives at the network, a path is selected to accommodate the request according to those predetermined probabilities.

\subsection{Notations}
\indent The input parameters are shown in Table \ref{table71}, including the detailed network topology information and the expected load for each route. Load for each route is provided based on the traffic pattern of the network. For example, if the source and destination of requests follow a uniform distribution, the load for each route can be set as 1.

\renewcommand{\arraystretch}{1.2}
\begin{table}
	\centering
	\footnotesize
	\caption{\label{table71} Notation for dynamic RSA problem}
	\begin{tabular}{|c|m{12cm}<{\centering}|}
		\hline
		\bf{Symbol}& \bf{Meaning} \\ \hline
		$N$ &  number of nodes in the network \\ \hline
		$L$ &  number of links in the network\\ \hline
		$e$ &  an arbitrary network link\\ \hline
		$F^{e}$ &  the number of fibers on link $e$,  $e{\in}\mathcal{E}$\\ \hline
		$R$ &  number of routes in the networks\\ \hline
		$r$ &  an arbitrary route\\ \hline
		$W_r$ & the  load for route $r$; depends on the traffic pattern; $W_r = 1, \quad \forall r$ for uniform traffic\\ \hline
		$K_r$ &  number of candidate paths for route $r$\\ \hline
		$k$ &  an arbitrary candidate path\\ \hline
		$A_{r, k}^{e}$ & = 1 if link $e$ is on the $k^{th}$ candidate path of route $r$; = 0, otherwise\\ \hline
	\end{tabular}
\vspace{10pt}
\end{table}

\subsection{Mixed Integer Linear Programming for Path Selection}
\indent We have derived Mixed Integer Linear Programming formulations to compute the optimal probabilities of candidate paths for each route (source-destination node pair). The objective is to minimize average and maximum traffic load over all fibers (as opposed to links, since different links may have different numbers of fibers) to balance the loads. 

\indent Objective: Minimize
\begin{center}
	$\frac{1}{L}\sum_{e = 1}^{L} \frac{\sum_{r = 1}^{R}W_{r}y^{e}_{r}}{F^{e}} + 
	\max_{e} \frac{\sum_{r = 1}^{R}W_{r}y^{e}_{r}}{F^{e}}$
\end{center}

\indent Variables:

a) The probability of selecting the $k^th$ candidate path of route $r$: 
\begin{equation*}
0 \leq p^{k}_{r} \leq 1.
\end{equation*}

b) The percentage of load for route $r$ that traverses link $e$: 
\begin{equation*}
0 \leq y^{e}_{r} \leq 1.
\end{equation*}

\indent Constraints:

a) The total probability of selecting candidate paths for a route should be $1.0$. For all $r$,
\begin{equation*}
	\sum_{k=1}^{K_r} p^{k}_{r} = 1,
\end{equation*}

b) The percentage of load for a route $r$ on a link depends on the candidate path selection probabilities. For all $r, e$,
\begin{equation*}
	y^{e}_{r} = \sum_{k=1}^{K_r} A_{r, k}^{e} p^{k}_{r}
\end{equation*}

\subsection{Evaluation}
\indent Figure \ref{fig:ps} shows the comparison between the performance (demand blocking ratio) of First-Fit slot assignment algorithm with a fixed single shortest path routing (SSP) and that of FF with our multi-path selection scheme (MPS) for 1 million dynamic requests. The detailed simulation settings are described in Section \ref{sec.label74}. We can see that our routing scheme works much better than the fixed single shortest path scheme. Since the selection probabilities can be precomputed via the ILP model, there will be no increment in computation overhead.

\begin{figure}
	\includegraphics[draft = false, scale = 0.45]{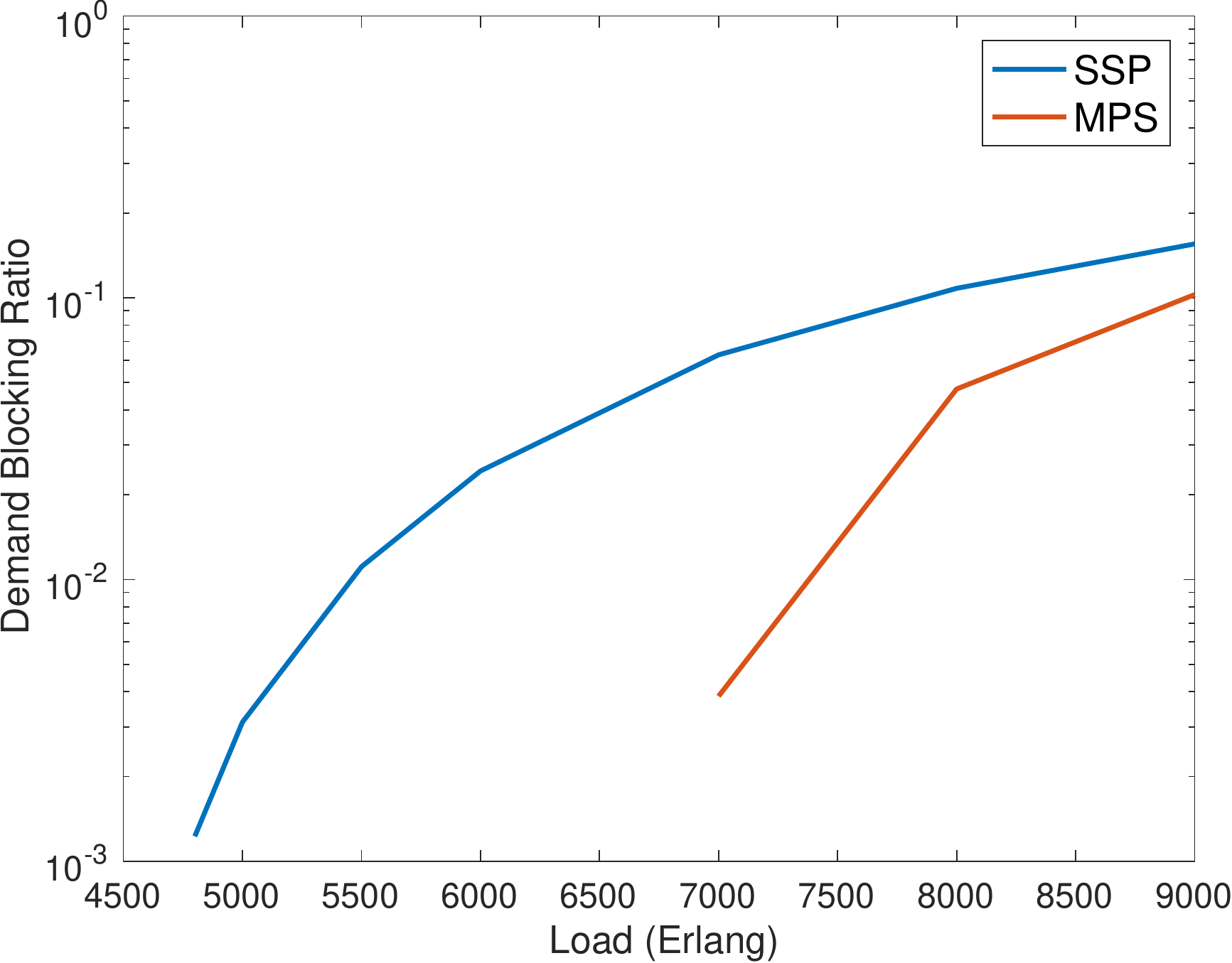}
	\centering
	\caption{\label{fig:ps} Effectiveness of the path selection scheme.}
\end{figure}

\section{Spectrum Assignment with Partitioning}\label{spectrum}
\indent We look at the spectrum assignment problem in this section. When a request arrives, a path is selected to accommodate the request by simply using the probability distribution computed above. After a path is chosen for the arriving request, SA is performed to assign contiguous FSs to that request. If there are no available slots on the selected path, the request will be blocked.

\subsection{Spectrum Management}

\indent To eliminate the fragmentation caused by the mismatch of heterogeneous request bandwidths during the dynamic set-up and tear-down of traffic, a dedicated partition is utilized. The spectrum is partitioned into different segments, each dedicated to traffic requests with the same bandwidth. 

\indent We use the number of contiguous slots to represent the size of a request. Suppose there are $M$ sets of connections with different request sizes $b_j,\ j = 1\cdots M$. Then the whole spectrum can be partitioned into $M$ segments. The number of contiguous FSs dedicated to each segment is denoted by $P_j,\ j = 1\cdots M$. Let $S$ be the size of the whole spectrum, that is, the total number of FSs in a fiber. Then the sizes of segments should meet the constraint: $\sum_{j = 1}^{M} P_j = S$. In the segment for request set $j$, we have $P_j / b_j$ bins (a bin is a consecutive set of FSs, with bin size of $b_j$) in each fiber.

\indent Assume that the traffic distribution is known, with $\rho_j$ being the probability for connection requests of size $b_j$. The total probability of the traffic distribution should be 1.0, $\sum_{j = 1}^{M} \rho_j = 1.0$. The distribution information is then used to calculate the sizes of different segments. 
\begin{equation}\label{parti}
P_j = S\cdot \frac{\rho_j \cdot b_j}{\sum_{j = 1}^{M} \rho_j \cdot b_j}
\end{equation}

\indent Through this partition, spectrum fairness can be achieved to some extent. For example, suppose there are 352 FSs on each fiber, there are 3 sets of requests with sizes $(b_1, b_2, b_3) = (3, 4, 7)$ and densities $(\rho_1, \rho_2, \rho_3) = (0.2, 0.5, 0.3)$. Then according to Equation \ref{parti}, we get the segment sizes $(P_1, P_2, P_3) = (45, 152, 154)$, which correspond to 15, 38, and 22 bins for requests of size 3, 4, and 7, respectively, and 75 bins in total. The distribution of the number of bins for each request size are almost the same with the traffic distribution.

\indent Figure \ref{fig:ffp} shows the spectrum efficiency improvement caused by dedicated spectrum partition. The baseline -- FF without partitioning is compared with the result of dedicated partition (PD) with FF spectrum assignment. Both use the single shortest path routing (SSP). We can see that even without resource sharing among partitions, there is a steady improvement in spectrum efficiency. The spectrum underutilization for small workloads can be improved by resource sharing among different spectrum segments.

\begin{figure}
	\includegraphics[draft = false, scale = 0.45]{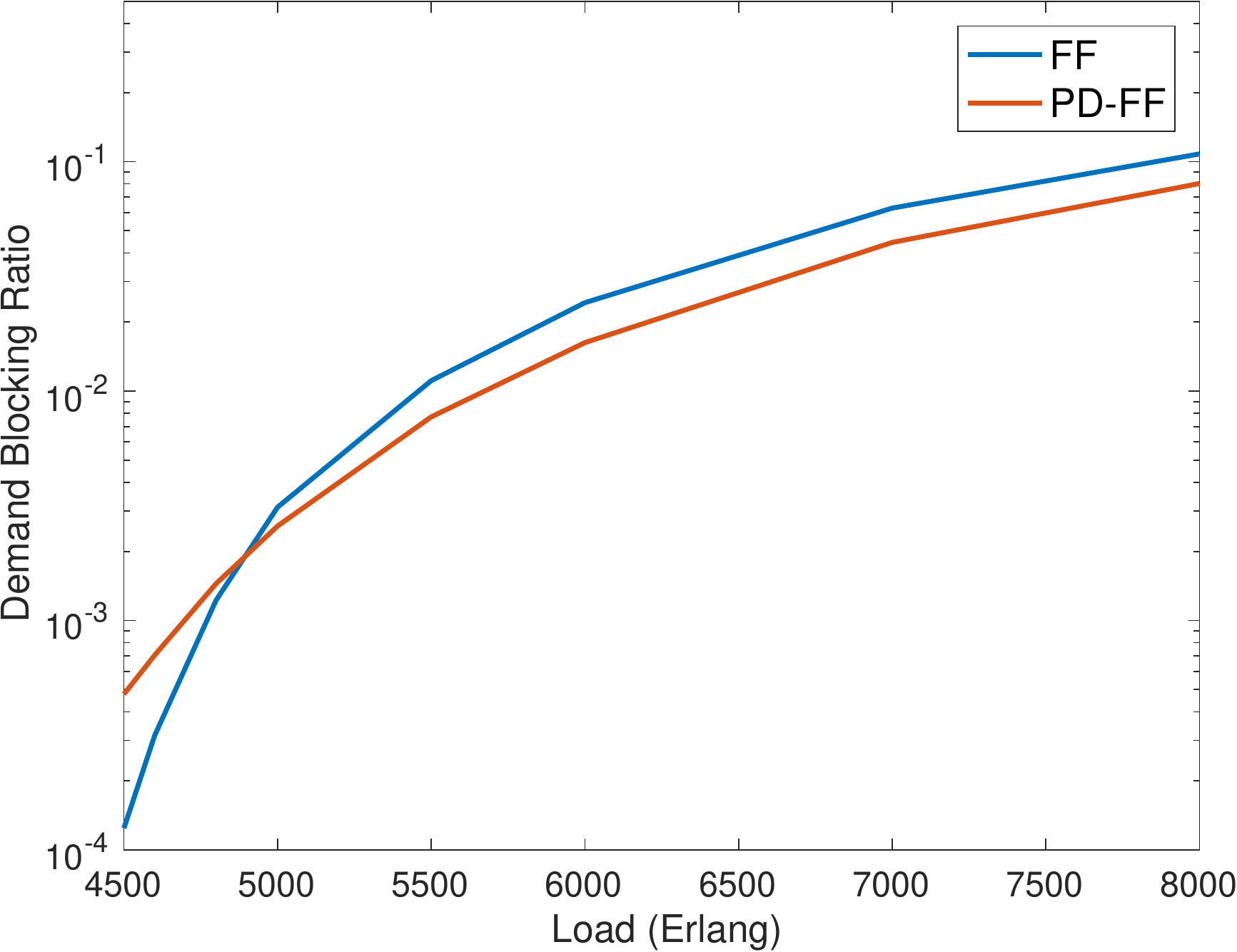}
	\centering
	\caption{\label{fig:ffp} Effectiveness of spectrum partition.}
\end{figure}

\subsection{Next State Aware Spectrum Assignment}

\indent In this section, we propose a Next State Aware (NSA) Spectrum Assignment algorithm based on both the next state (i.e., the network state {\em after} the request is set up) and the path selection probabilities for requests, given the spectrum partitioning. 

\indent Since the whole spectrum has been partitioned into dedicated segments, each providing consecutive slots for a particular request size, the spectrum assignment problem can be addressed by scheduling a bin in the specific segment to a request. In the $j$th dedicated segment of the spectrum, we have $P_j / b_j$ FS bins (a bin is a set of $b_j$ consecutive FSs) in each fiber. Suppose the request has a bandwidth requirement of $b_j$ FSs; the SA problem is to assign a bin of size $b_j$ for the request. The spectrum assignment problem is then transformed to select an available bin in the dedicated segment to accommodate a request without disrupting existing connections so that the blocking probability is minimized.

\indent For each new request, we only need to consider the network state corresponding to the specific spectrum segment which this request size belongs to, based on the partitioning in the previous section. The network state denotes the current status of the network, such as the slot usage in fibers. Since dedicated partitioning has been determined, we only need to look at the status of each bin in that segment. In order to reduce the blocking probability, we should select a bin that can provide good network state after the request is established, e.g. the path capacities are least reduced.     

\indent In the rest of this section, we assume $X$ to be the number of bins in the segment the request belongs to. We first define the link capacity $c_e^x$ of link $e$ on bin $x$ in one network state as the number of fibers on which $x$ is unused on the link. Initially, in an empty network, $c_e^{x} = F^{e}, \forall x$. The path capacity is determined based on capacities of links along the path. For a path $k$, the path capacity on bin $x$ is defined as the least link capacity on bin $x$ along the path (which is also the link capacity of the most congested link along the path),
\begin{equation}
C_k^x = \min_{e\in \Xi(k)} c_e^x. 
\end{equation}

\noindent
where $\Xi(k)$ is the set of links in path $k$. Then the path capacity $C_k$ is the total capacity across all bins,
\begin{equation}
C_k = \sum_{x = 1}^{X} C_k^x.
\end{equation}

\noindent
Since for each route (source-destination node pair) $r$ in the network, we've already determined the selection probabilities $p_r^k$ for each path candidate $k$, the capacity for that route is then defined as
\begin{equation}
C_r = \sum_{k = 1}^{K_r} p_r^k C_k. 
\end{equation}

\noindent
Then the network state can be represented as the total capacity over all routes $\sum_{r=1}^{R} C_r$.

\indent Suppose a request is assigned to a candidate path $\kappa$. After the request is established on a bin in the spectrum segment, there may be capacity loss for other routes in the network. We should choose a bin that causes the least capacity loss among all available bins to accommodate the request. In order to calculate the capacity loss, we utilize the definition of conflict graph $\mathcal{G'}=(\mathcal{V'},\mathcal{E'})$, where each vertex $k \in \mathcal{V'}$ represents a path $k$ in the network, and an undirected edge $(k_1, k_2) \in \mathcal{E'}$ denotes that paths $k_1$ and $k_2$ shares at least one link. The capacity loss will occur only in the paths which contain common links with $\kappa$. So when calculating the total capacity loss, we only need to look at paths that are connected to $\kappa$ in the conflict graph $\mathcal{G'}$. We use $\psi^{\kappa}$ to denote this set of paths.

\indent Let's consider the total capacity loss $\zeta_{x}$ when bin $x$ is chosen for path $\kappa$. If at least one of the common links between $\kappa$ and a path $k' \in \psi^{\kappa}$ have the minimum capacity on bin $x$ along path $k'$, $C_{k'}^{x}$ will be decreased by 1 after the establishment of the request. Otherwise, when all common links between $\kappa$ and $k'$ have capacity larger than the current path capacity $C_{k'}^{x}$, the link capacity decrement of the common links caused by the request establishment will not affect $C_{k'}^{x}$ and there will be no capacity loss for path $k'$. We use $\upsilon_{k'}$ to denote the capacity loss in path $k'$. Taking the path selection probability into account, if path $k'$ belongs to route $r$ and the probability of selecting $k'$ for $r$ is $p_{r}^{k'}$, then $\upsilon_{k'}$ equals to either $p_{r}^{k'}$ or $0$, depending on the network state. The total capacity loss after bin $x$ is chosen will be $\zeta_{x} = \sum_{k' \in \psi^{\kappa}} \upsilon_{k'}$. We should assign the bin $x*$ which causes least capacity loss by the request, e.g., 
\begin{equation}
x* = \arg \min_{x \in \Omega(\kappa)} \zeta_{x}
\end{equation}

\noindent
where $\Omega(\kappa)$ is the set of available bins in path $\kappa$ for the request.

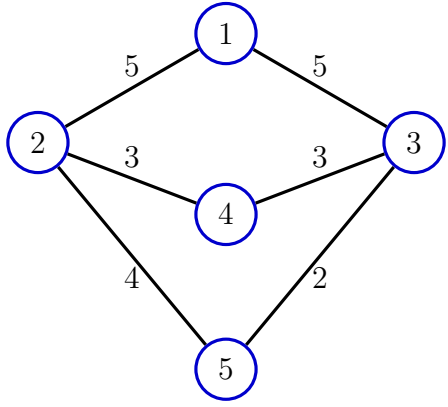
\begin{figure}
	\centering
	\begin{tikzpicture}[yscale=0.58,
	inner/.style={circle, minimum width=8mm,draw=blue!80!black,fill=white!30, very thick,inner sep=0pt},
	]
	\node[inner] (1) at (0,2.5) {1};
	\node[inner] (2) at (-2.5,0) {2};
	\node[inner] (3) at (2.5,0) {3};
	\node[inner] (4) at (0,-1.65) {4};
	\node[inner] (5) at (0,-5.2) {5};
	\path[draw, very thick] (1) --node[pos=0.5,above] {5} (2) --node[pos=0.5,above] {3} (4) --node[pos=0.5,above] {3} (3) --node[pos=0.5,above] {5} (1);
	\path[draw, very thick] (2) --node[pos=0.5,below] {4} (5) --node[pos=0.5,below] {2} (3);
	\end{tikzpicture}
	\caption{\label{fig:RSA_smallnet}A small 5-node network topology.}
\end{figure}

\indent Let's take a small 5-node network as Figure \ref{fig:RSA_smallnet} as an example. There are 5 physical nodes and 6 physical links in the network. The number of fibers are $5, 5, 3, 4, 3, 2$ on bidirectional links $(1,2), (1,3), (2,4), (2,5),(3,4),(3,5)$, respectively. Given the path candidates for each route, based on the multi-path selection formulations, we got the probabilities as in Table \ref{psexample}.

\begin{table}
	\centering
	\footnotesize
	\caption{Path selection probabilies for the small network.}
	\label{psexample}
	\begin{tabular}{|c|c|c|}
		\hline
		\rule{0pt}{12pt}\bf{Route} & \bf{Candidate Paths} & \bf{Selection Probability} \\
		\hline
		1 $\sim$ 2 & 1-2 & 1.0 \\ \hline
		1 $\sim$ 3 & 1-3 & 1.0 \\ \hline
		\multirow{2}{*}{1 $\sim$ 4} & 1-2-4 & 1/3 \\ \cline{2-3} 
		& 1-3-4 & 2/3 \\ \hline
		\multirow{2}{*}{1 $\sim$ 5} & 1-2-5 & 1.0 \\ \cline{2-3} 
		& 1-3-5 & 0.0 \\ \hline
		\multirow{3}{*}{2 $\sim$ 3} & 2-1-3 & 1.0 \\ \cline{2-3} 
		& 2-4-3 & 0.0 \\ \cline{2-3} 
		& 2-5-3 & 0.0 \\ \hline
		2 $\sim$ 4 & 2-4 & 1.0 \\ \hline
		2 $\sim$ 5 & 2-5 & 1.0 \\ \hline
		3 $\sim$ 4 & 3-4 & 1.0 \\ \hline
		3 $\sim$ 5 & 3-5 & 1.0 \\ \hline
		\multirow{2}{*}{4 $\sim$ 5} & 4-2-5 & 2/3 \\ \cline{2-3} 
		& 4-3-5 & 1/3 \\ \hline
	\end{tabular}
\vspace{10pt}
\end{table}

\indent When a new request from source node 2 to destination node 5 arrives, from Table \ref{psexample}, path $2-5$ is selected. Then we should check the network state of the spectrum segment this request belongs to. Based on the conflict graph, we can see that only paths $1-2-5$, $4-2-5$, $2-5-3$ have common links with path $2-5$. So we only need to check the state of partial network. Assume there are $4$ bins $x_1, x_2, x_3, x_4$ in the segment. Current network states are shown as Figure \ref{netstate}. The number attached to each link denotes current link capacity on a bin, which is the number of fibers the bin is unused on.

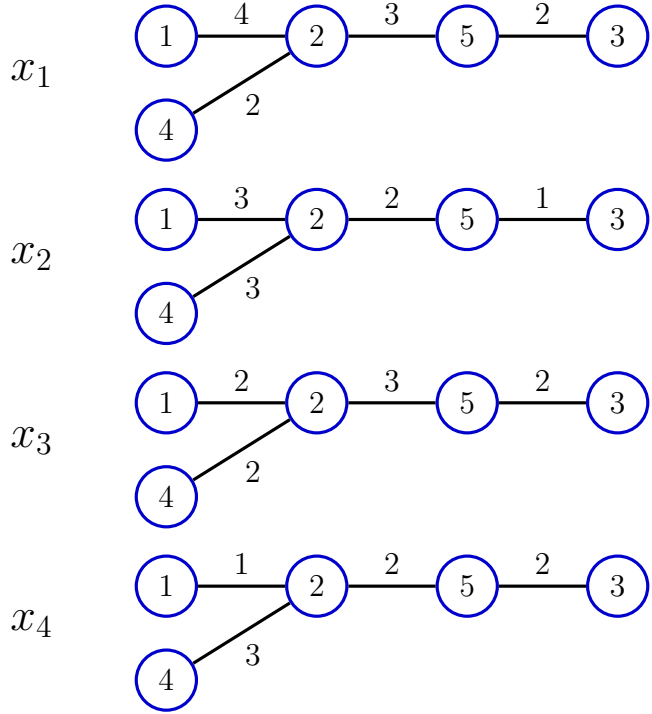
\begin{figure}
	\centering
	\begin{tikzpicture}[yscale=0.5,inner/.style={circle, minimum width=8mm,draw=blue!80!black,fill=white!30, very thick,inner sep=0pt}]
	\node[inner sep=0pt] at (-0.8, -1) {\parbox{\linewidth/6}{\Large $x_{1}$}};
	\node[inner] (1) at (0, 0) {1};
	\node[inner] (2) at (2,0) {2};
	\node[inner] (5) at (4,0) {5};
	\node[inner] (3) at (6,0) {3};
	\node[inner] (4) at (0,-2.5) {4};
	\path[draw, very thick] (1) --node[pos=0.5,above] {4} (2) --node[pos=0.5,above] {3} (5) --node[pos=0.5,above] {2} (3);
	\path[draw, very thick] (4) -- node[pos=0.5,below, xshift = 1.5mm]{2} (2);
	\node[inner sep=0pt, below, yshift = -18mm]{};
	\end{tikzpicture}
	\begin{tikzpicture}[yscale=0.5,inner/.style={circle, minimum width=8mm,draw=blue!80!black,fill=white!30, very thick,inner sep=0pt}]
	\node[inner sep=0pt] at (-0.8, -1) {\parbox{\linewidth/6}{\Large $x_{2}$}};
	\node[inner] (1) at (0,0) {1};
	\node[inner] (2) at (2,0) {2};
	\node[inner] (5) at (4,0) {5};
	\node[inner] (3) at (6,0) {3};
	\node[inner] (4) at (0,-2.5) {4};
	\path[draw, very thick] (1) --node[pos=0.5,above] {3} (2) --node[pos=0.5,above] {2} (5) --node[pos=0.5,above] {1} (3);
	\path[draw, very thick] (4) -- node[pos=0.5,below, xshift = 1.5mm]{3} (2);
	\node[inner sep=0pt, below, yshift = -18mm]{};
	\end{tikzpicture}
	\begin{tikzpicture}[yscale=0.5,inner/.style={circle, minimum width=8mm,draw=blue!80!black,fill=white!30, very thick,inner sep=0pt}]
	\node[inner sep=0pt] at (-0.8, -1) {\parbox{\linewidth/6}{\Large $x_{3}$}};
	\node[inner] (1) at (0,0) {1};
	\node[inner] (2) at (2,0) {2};
	\node[inner] (5) at (4,0) {5};
	\node[inner] (3) at (6,0) {3};
	\node[inner] (4) at (0,-2.5) {4};
	\path[draw, very thick] (1) --node[pos=0.5,above] {2} (2) --node[pos=0.5,above] {3} (5) --node[pos=0.5,above] {2} (3);
	\path[draw, very thick] (4) -- node[pos=0.5,below, xshift = 1.5mm]{2} (2);
	\node[inner sep=0pt, below, yshift = -18mm]{};
	\end{tikzpicture}
	\begin{tikzpicture}[yscale=0.5,inner/.style={circle, minimum width=8mm,draw=blue!80!black,fill=white!30, very thick,inner sep=0pt}]
	\node[inner sep=0pt] at (-0.8, -1){\parbox{\linewidth/6}{\Large $x_{4}$}};
	\node[inner] (1) at (0,0) {1};
	\node[inner] (2) at (2,0) {2};
	\node[inner] (5) at (4,0) {5};
	\node[inner] (3) at (6,0) {3};
	\node[inner] (4) at (0,-2.5) {4};
	\path[draw, very thick] (1) --node[pos=0.5,above] {1} (2) --node[pos=0.5,above] {2} (5) --node[pos=0.5,above] {2} (3);
	\path[draw, very thick] (4) -- node[pos=0.5,below, xshift = 1.5mm]{3} (2);
	\end{tikzpicture}
	\caption{\label{netstate}Network state example -- link capacities of partial network on bin ${\normalsize x_1, x_2, x_3, x_4}$.}
\end{figure}

\indent Let's take bin $x_2$ as an example to see how the total capacity loss is calculated. Since link $(2,5)$ has the least link capacity along both path $1-2-5$ and $4-2-5$, if $x_2$ is allocated to the new request, both path capacities will be decreased. According to the selection probabilities of these two paths, capacity loss is $1.0$ for path $1-2-5$, and $\frac{2}{3}$ for path $4-2-5$. For path $2-5-3$, the selection probability is $0$, so we don't need to take this path capacity into account. In total, the capacity loss is $\frac{5}{3}$ if bin $x_2$ is used. Similarly, capacity losses for $x_1, x_3, x_4$ are $1.0, 0, \frac{2}{3}$, respectively. Bin $x_3$ will be allocated to accommodate this request.

\indent To evaluate the efficiency of our NSA algorithm, we run simulations for traffic requests with a single size. The performance of NSA is compared with that of baseline -- First Fit spectrum assignment. Both use the single shortest path routing (SSP). The first set of comparison is on a 20-node unidirectional ring with 10 fibers per link. Figure \ref{fig:nsa_ring} shows a huge spectrum efficiency improvement by utilizing NSA in the unidirectional ring network. The second set of comparison is on the 14-node NSF network, with different number of fibers on each link. In Figure \ref{fig:nsa_nsf}, the spectrum efficiency improvement in the mesh network is not as obvious as that in the unidirectional ring network. It indicates that when there are more paths overlapping (more common links), our algorithm works better. But in Figure \ref{fig:nsa_nsf}, we can still see $9\%-31\%$ performance improvement for the mesh network.

\begin{figure}
	\includegraphics[draft = false, scale = 0.45]{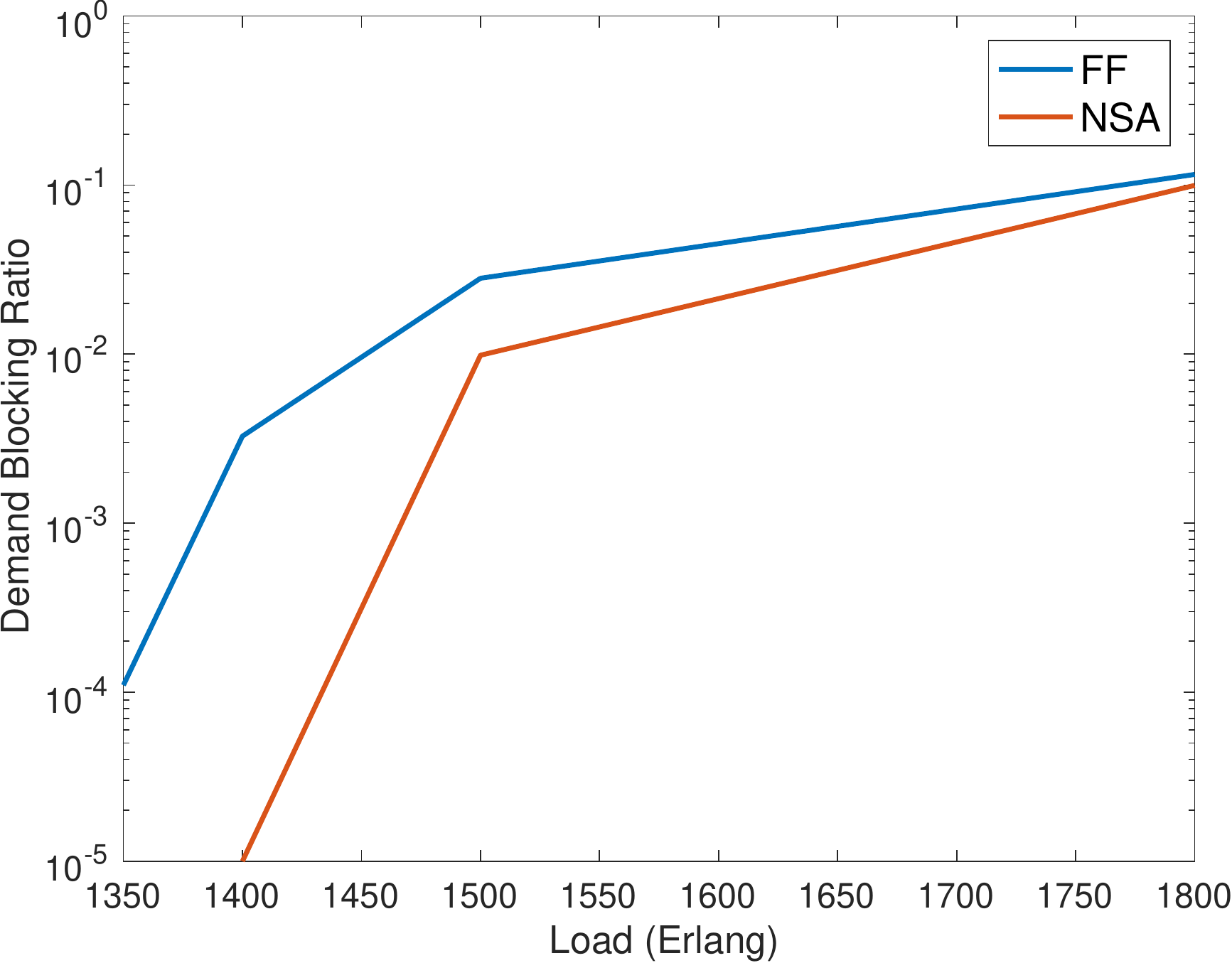}
	\centering
	\caption{\label{fig:nsa_ring} Effectiveness of NSA algorithm in a 20-node unidirectional ring network.}
\end{figure}

\begin{figure}
	\includegraphics[draft = false, scale = 0.45]{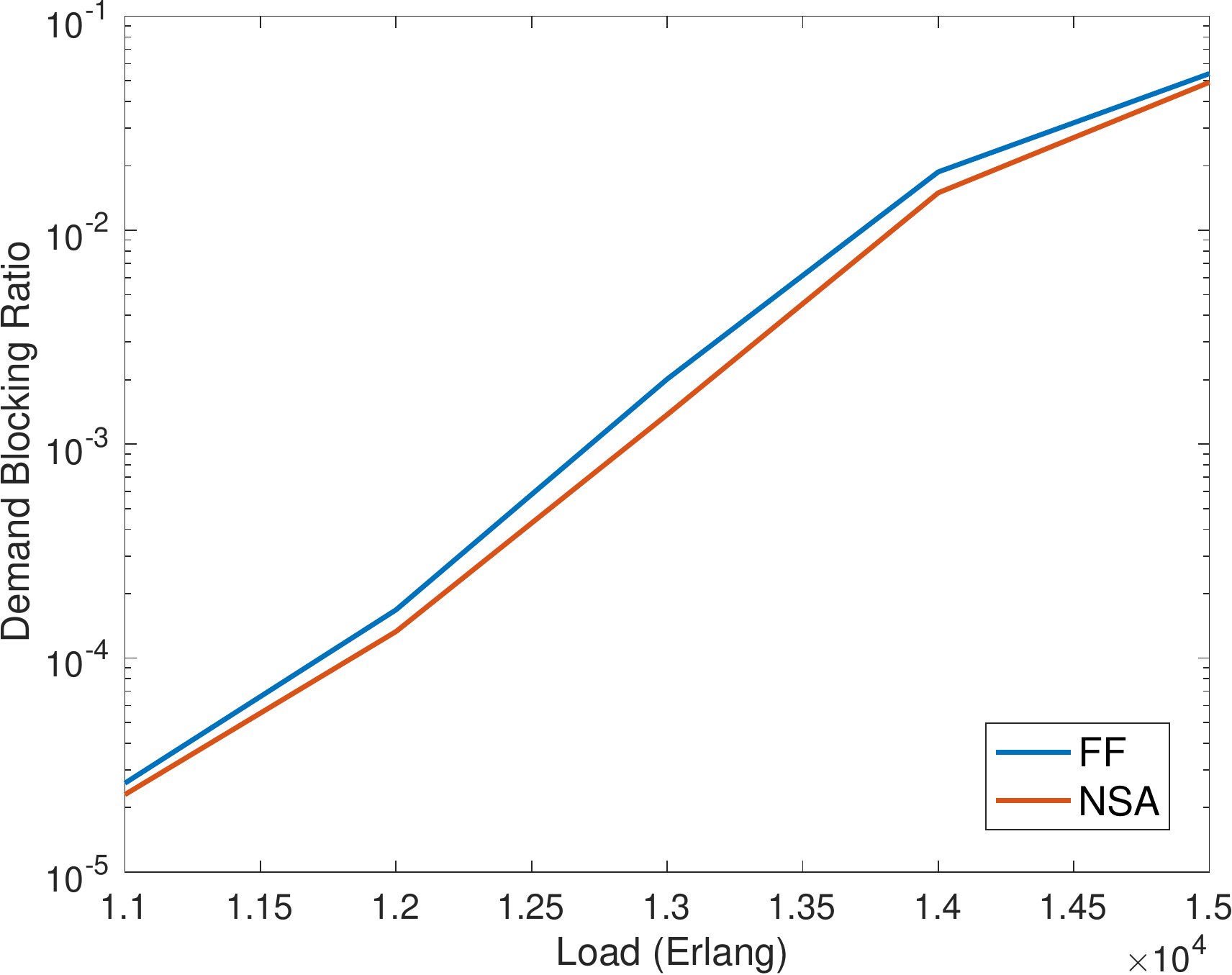}
	\centering
	\caption{\label{fig:nsa_nsf} Effectiveness of NSA algorithm in a 14-node NSF network.}
\end{figure}

\indent Figure \ref{fig:msp} shows the spectrum efficiency improvement caused by both dedicated spectrum partition and our proposed NSA algorithm. The baseline -- FF is compared with the result of dedicated partition (PD) with NSA spectrum assignment. Both use single shortest path routing (SSP). We can see that even without resource sharing among different partitions, there is a steady improvement in spectrum efficiency. The underutilization of spectrum in the small load can be eliminated by resource sharing among spectrum segments.

\begin{figure}
	\includegraphics[draft = false, scale = 0.45]{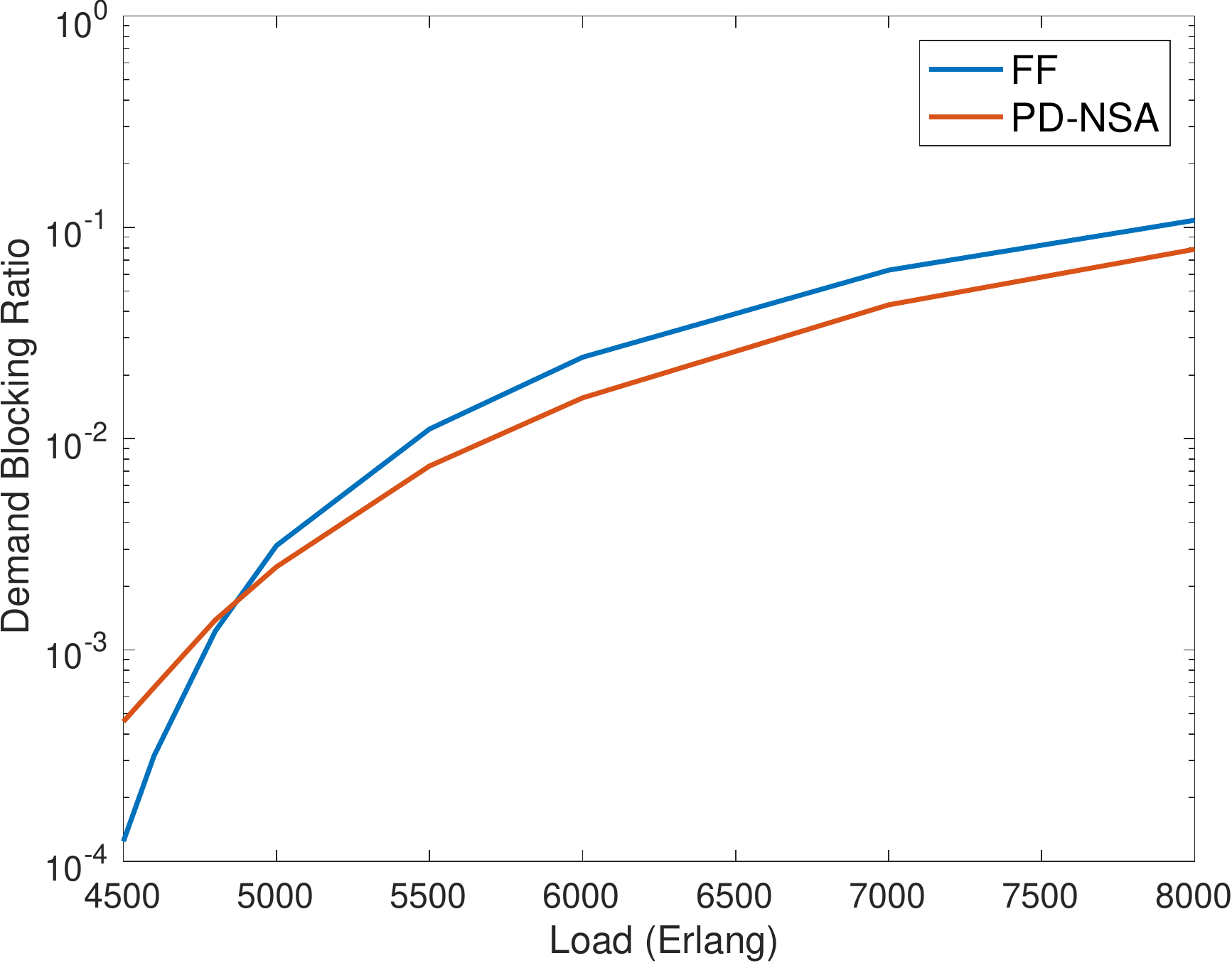}
	\centering
	\caption{\label{fig:msp} Effectiveness of spectrum partition and NSA algorithm.}
\end{figure}

\subsection{Resource sharing among partitions}

\begin{figure}
	\includegraphics[draft = false, scale = 0.45]{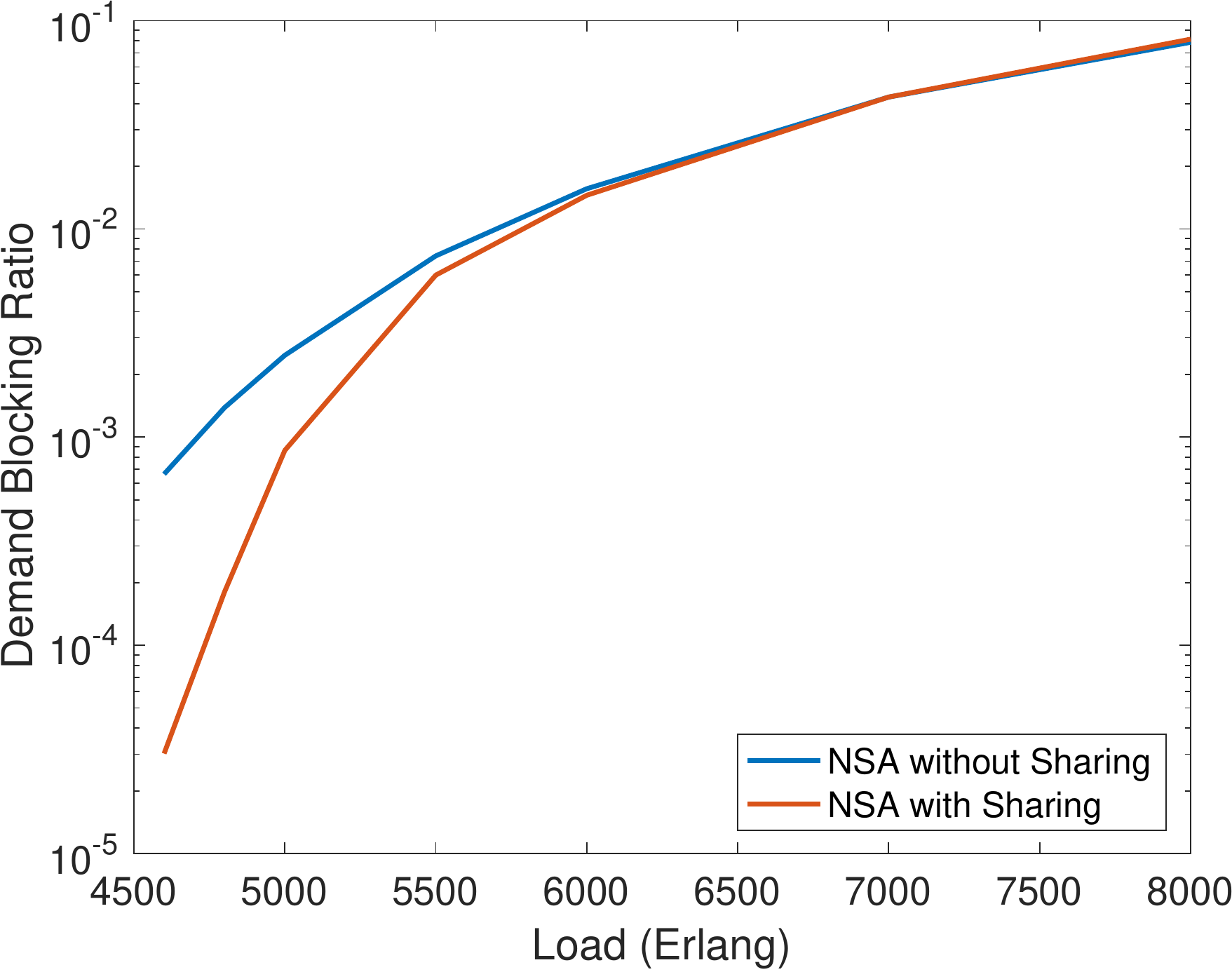}
	\centering
	\caption{\label{fig:prsmsmp} Effectiveness of resource sharing among partitions.}
\end{figure}

\indent According to the observation of previous results, there are spectrum underutilizations for low loads caused by the dedicated partitioning. So we consider to further improve spectrum efficiency by resource sharing among partitions. For each request, we first check whether there is any available bin in the dedicated spectrum segment this request belongs to. If no available bin exists in this segment, FSs in other partitions will be examined. In addition to the bin states in each segment, the occupancy states for each FS also need to be recorded, which are used to determine the candidate FS sets. We use the same assignment scheme as the previous section, with slight modifications in the calculation of capacity loss. Take a candidate FS set $\chi$ as an example, it may be across multiple bins, and only partial of the first or last bins are in $\chi$. For the bins where all FSs are within $\chi$, the capacity loss calculation is the same as that in the previous section. For the first or last bin, if it has already been marked as occupied (which means some slots in the bin while outside $\chi$ are in use) because of preceding requests, the link capacity loss for this bin is set to 0. The total capacity loss will be the sum of capacity loss for each bin in this candidate FS set. Still, the FS set which causes minimum capacity loss will be assigned to the request. If $\chi$ is the selection, all bins $\chi$ crosses will be marked as occupied. Figure \ref{fig:prsmsmp} shows the spectral efficiency improvement by resource sharing among partitions. We can see a huge improvement in performance for low traffic loads in the network.

\section{Numerical Results}\label{sec.label74} 

\indent We present performance evaluation results for two real network topologies, the NSF network (Figure~\ref{fig:NSF}) and the pan-European network (Figure~\ref{fig:pan}). Each link has a random number of fibers which are uniformly distributed between 5 and 10 fibers. For the elastic optical network, we assume the fiber capacity of 352 frequency slots, with each slot having a bandwidth of 12.5 GHz. 

\indent The traffic demand is a set of dynamically arriving connection requests. Each request represents a connection between a pair of nodes in the network. Connection requests arrive to the network according to a Poisson process. Each request has a mean holding time of 1 (arbitrary time unit), and the arrival rate of traffic requests is varied in order to examine the network performance under varying offered loads. Frequency conversion is not considered in this work. There are three types of demands with 40/100/400 Gbps, requesting 3, 4, 7 frequency slots, respectively~\cite{wu2015comparison}. The distribution is $(\rho_1, \rho_2, \rho_3) = (0.2, 0.5, 0.3)$. Based on this information, the spectrum partitioning can be determined. We use the demand blocking ratio of dynamic traffic requests to indicate the performance in multi-fiber elastic optical networks. For each simulation, the results of 1 million dynamic requests excluding 10000 warm-up requests are recorded.

\indent We use "R" to denote random slot set assignment, "FF" as first fit SA, "FLF" as first last fit SA, "MK" as the scheme proposed in the literature \cite{wang2014spectrum} and "NSA" as our network state aware algorithm. In the FLF SA scheme \cite{chatterjee2015routing}, the whole spectrum is divided into several partitions, requests attempt to use the lowest indexed FSs in the odd number partitions and highest indexed FSs in the even number partitions to create chance for more contiguous FSs. We adapt MK to our multi-fiber link model. In MK, the same partitioning with NSA and first fit slot assignment in each partition is used, while the resource sharing in MK is different, each partition is shared by only higher bandwidth requests. 

\indent We first conduct a comparison among several spectrum assignment schemes with the fixed single shortest path routing. Figure \ref{fig:SSPNSF} shows the results for the uniform traffic pattern in NSF network. The source and destination nodes for each connection request are uniformly randomly selected. As expected, in schemes without dedicated partitioning, R always performs worst and FLF performs better than FF by giving more contiguous aligned FSs. In low load cases, MK has higher demand blocking probability than FF and FLF due to its limited resource sharing and therefore resource underutilization. Then the effect of spectrum management (decreased fragmentation) dominates in higher load cases, which leads to better performance of MK than FF and FLF. Our scheme takes all these aspects into account, and performs spectrum allocation based on the global network states, therefore it shows best performances.

\begin{figure}
	\includegraphics[draft = false, scale = 0.5]{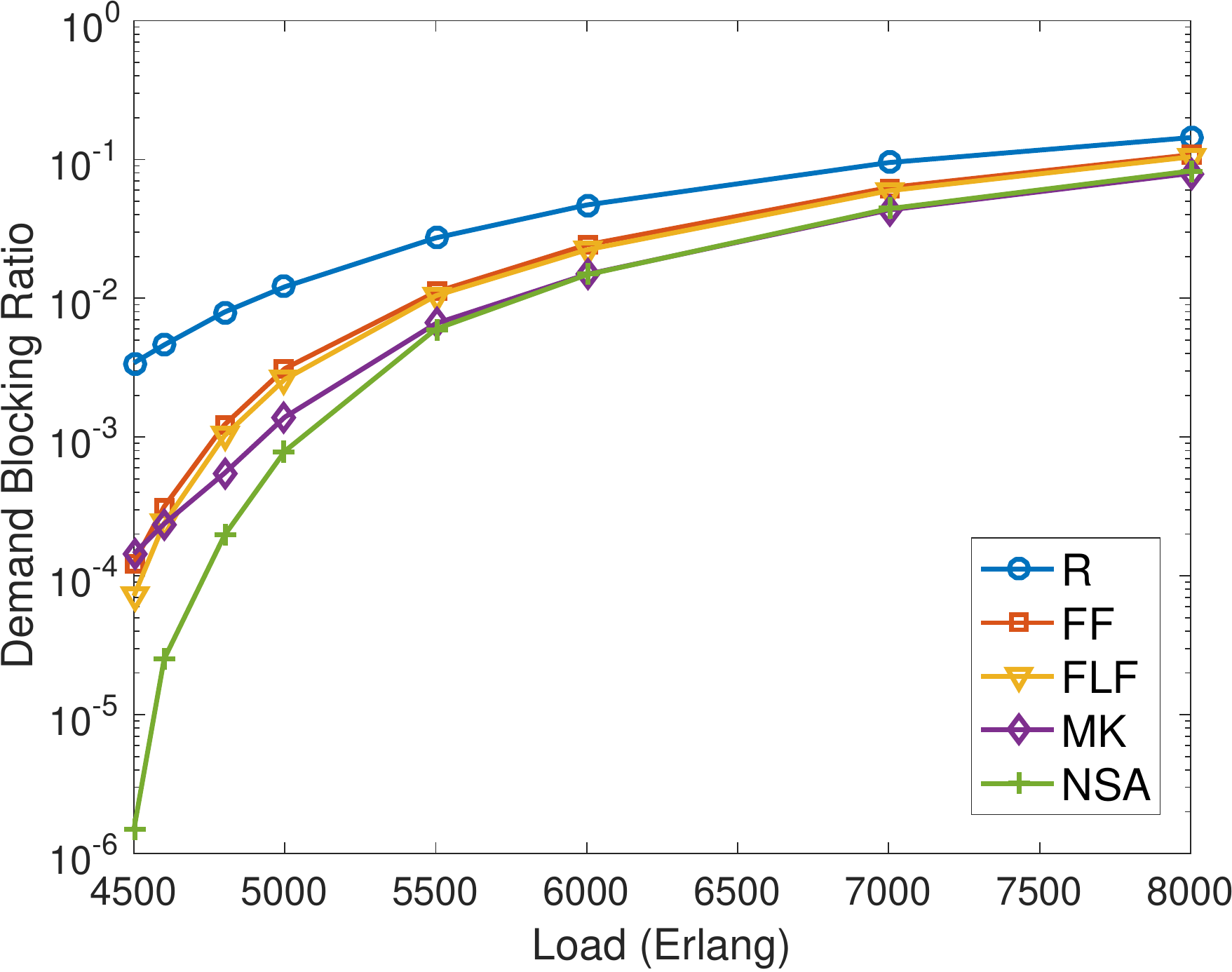}
	\centering
	\caption{\label{fig:SSPNSF} Comparison of spectrum allocation schemes in 14-Node NSF network with uniform traffic pattern.}
\end{figure}

\indent Then we conduct two sets of simulations to evaluate the joint effect of the multi-path selection and the spectrum assignment schemes. According to the traffic pattern, a set of path selection probabilities is precomputed. We apply this multi-path selection jointly with all spectrum assignment schemes in our evaluation. 

\indent We first show a set of simulation results for the uniform traffic pattern. Figure \ref{fig:NSFcomp} and \ref{fig:Eurocomp} show the performance comparison among our scheme and other baselines jointly with multi-path selection for both NSF and Pan-European networks. With the joint scheme, our algorithm still performs best and the same conclusion can be drawn. By comparing results in Figure \ref{fig:SSPNSF} and \ref{fig:NSFcomp}, we can see that by jointly applying multi-path selection with each spectrum assignment scheme, the demand blocking ratio greatly improves and more load can be accommodated in the network. 

\begin{figure}
	\includegraphics[draft = false, scale = 0.5]{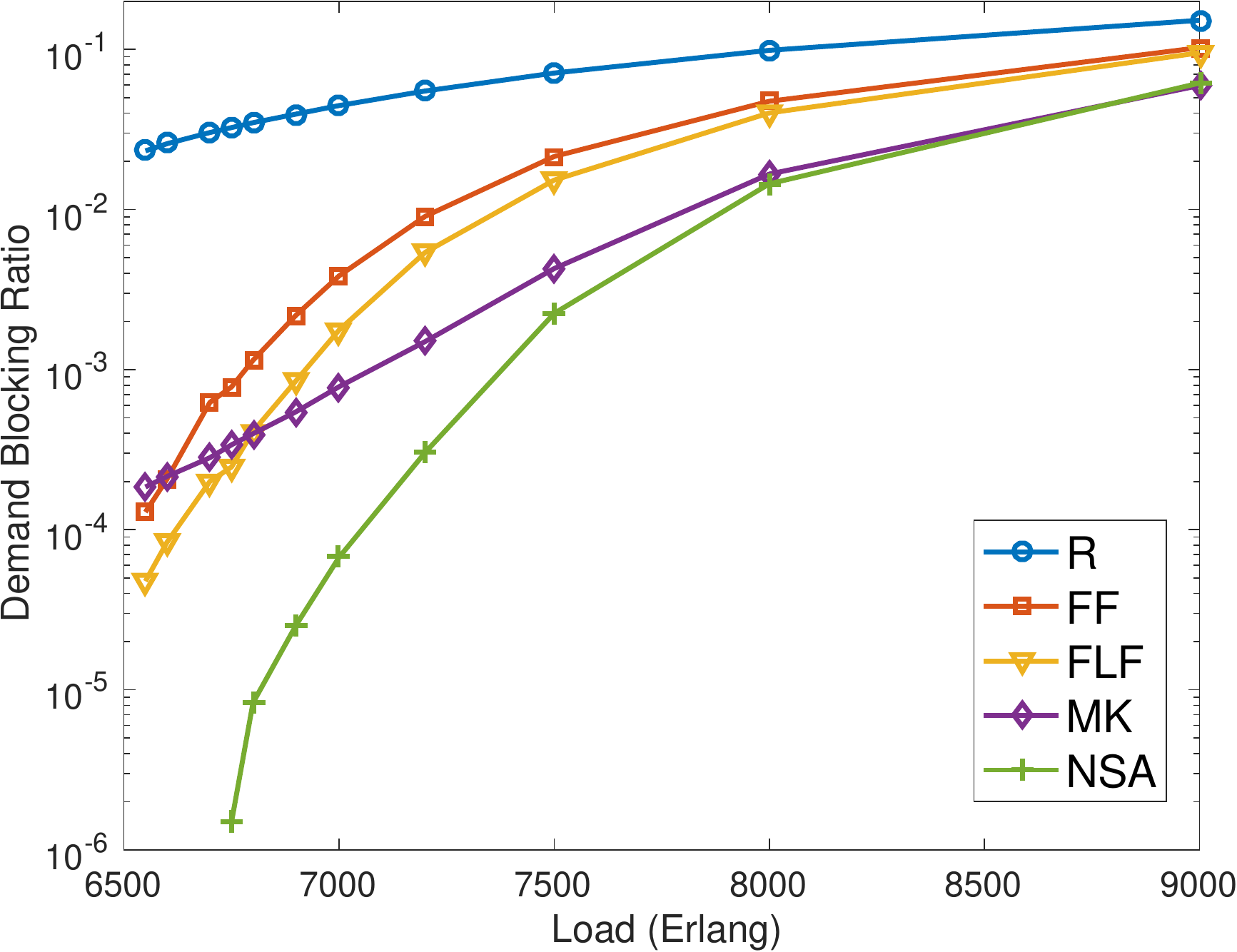}
	\centering
	\caption{\label{fig:NSFcomp} Joint comparison in 14-Node NSF network with uniform traffic pattern.}
\end{figure}

\begin{figure}
	\includegraphics[draft = false, scale = 0.5]{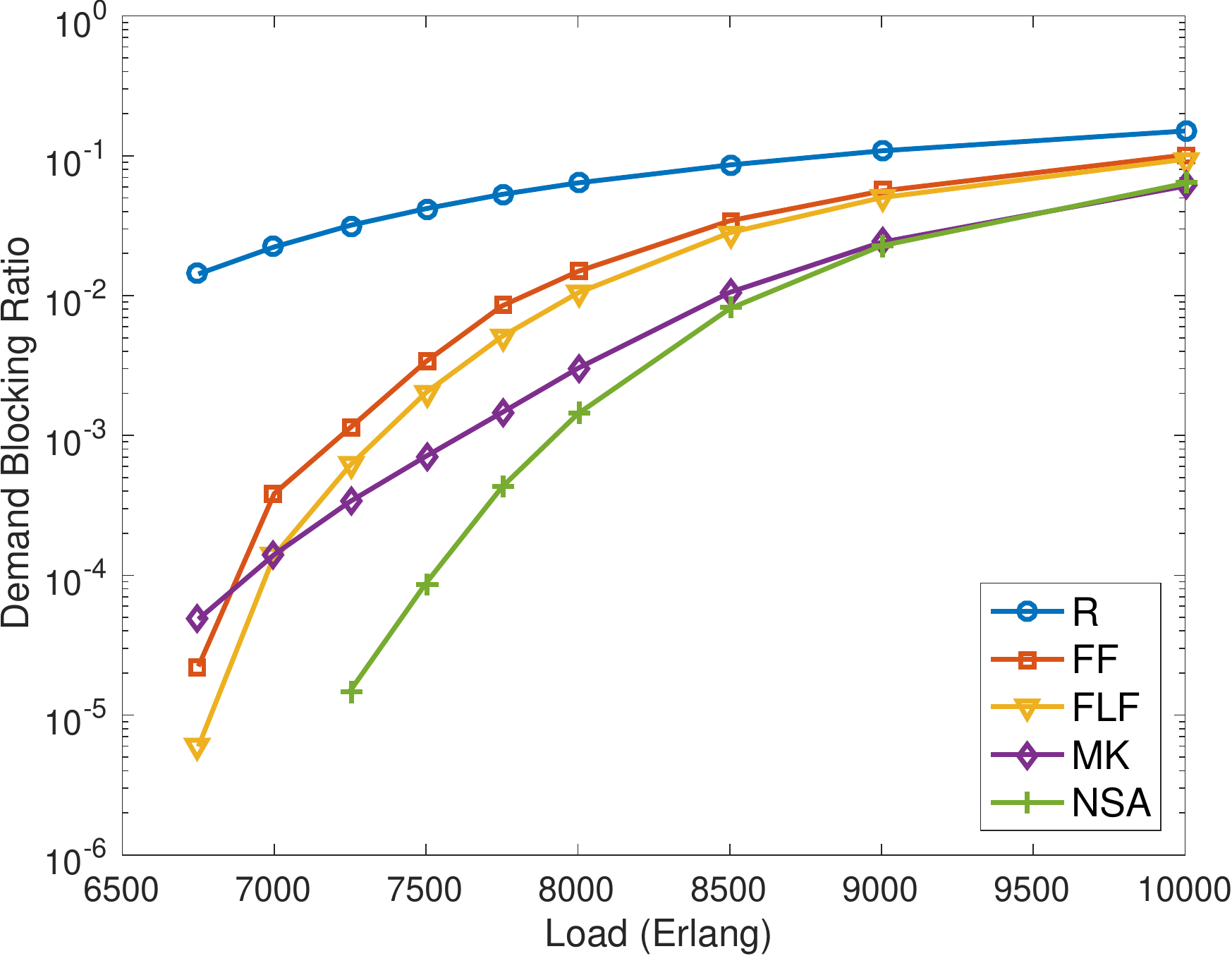}
	\centering
	\caption{\label{fig:Eurocomp} Joint comparison in 28-Node Pan-European network with uniform traffic pattern.}
\end{figure}

\indent The other set of simulations is conducted for a non-uniform traffic pattern, where the probability that each node is selected as source or destination $u_v, v\in\mathcal{V}$ can be different from each other. $\sum_{v\in\mathcal{V}}u_v = 1.0$. We assume nodes with higher connectivities have larger chance to send and receive traffic, and the probabilities are in proportional to the node degrees. The comparison results for both NSF and Pan-European network are shown in Figure \ref{fig:NSFcompnon} and \ref{fig:Eurocompnon} respectively. We can see similar trends to the previous results. This further validates the effectiveness of our algorithm.

\begin{figure}
	\includegraphics[draft = false, scale = 0.5]{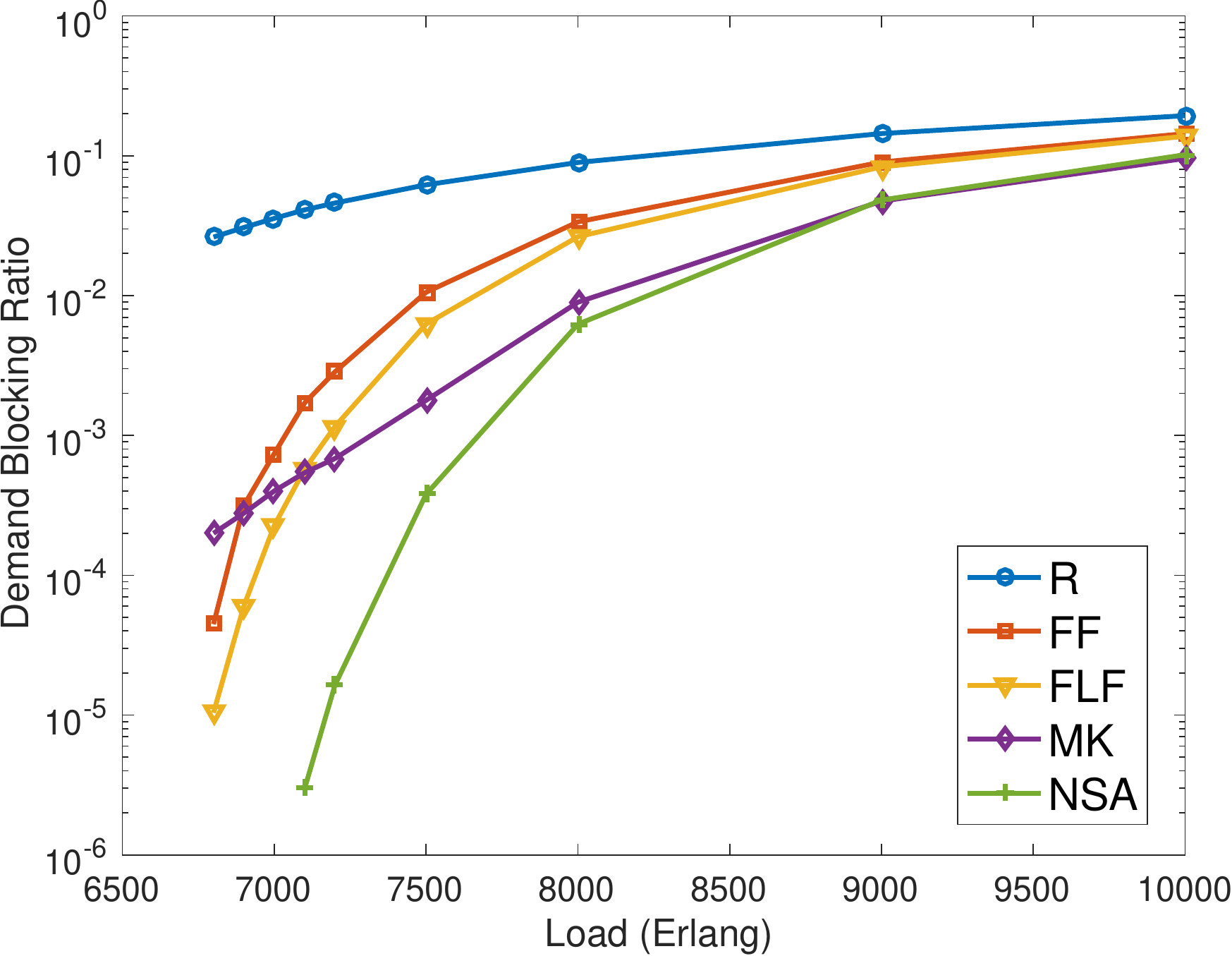}
	\centering
	\caption{\label{fig:NSFcompnon} Joint comparison in 14-Node NSF network with non-uniform traffic pattern.}
\end{figure}

\begin{figure}
	\includegraphics[draft = false, scale = 0.5]{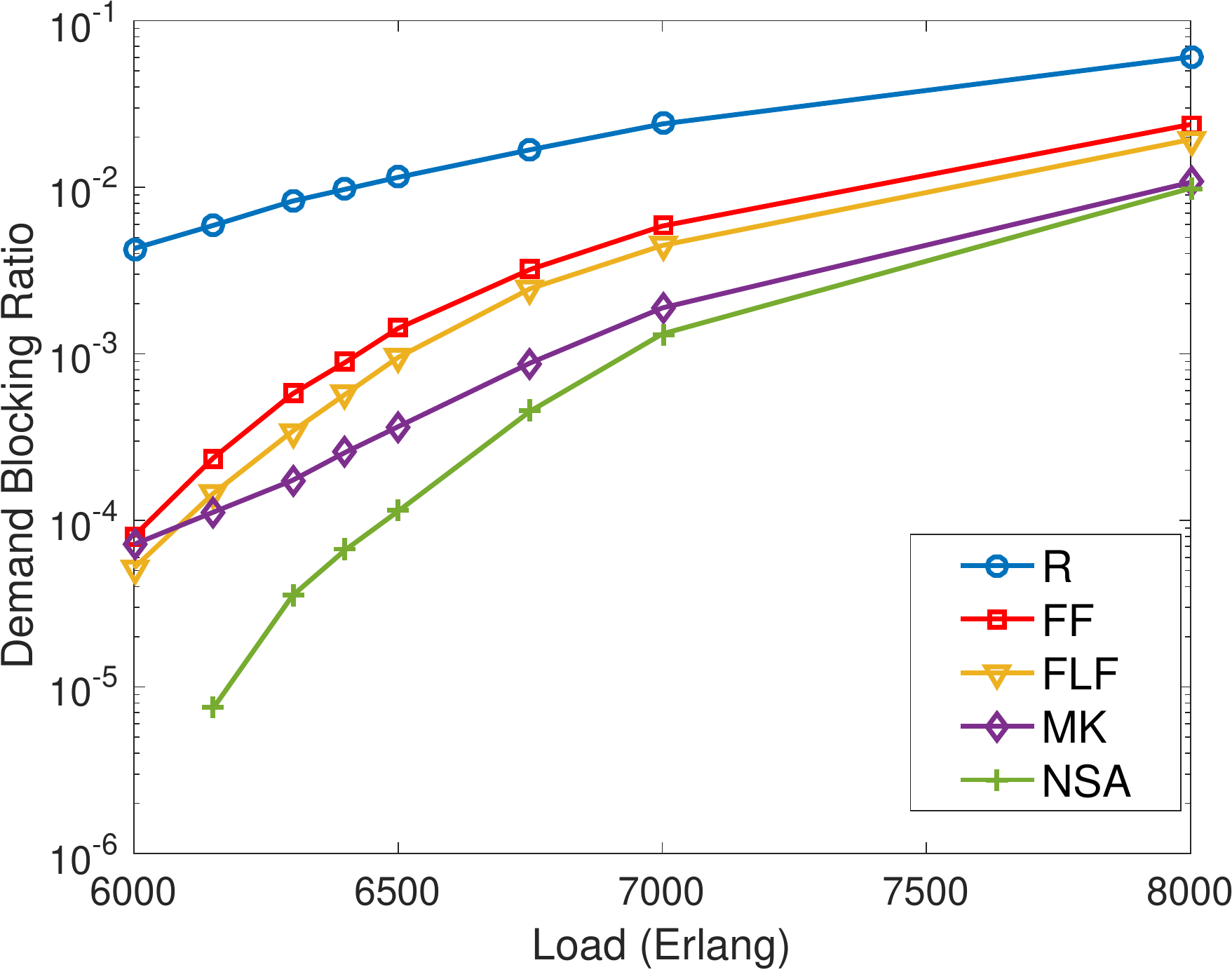}
	\centering
	\caption{\label{fig:Eurocompnon} Joint comparison in 28-Node Pan-European network with non-uniform traffic pattern.}
\end{figure}

\indent We also conduct sensitivity tests to see how well our scheme performs when there is variation in the traffic pattern and demand size distribution. Figure \ref{fig:sensPT} shows the test results regarding to demand size distribution in the NSF network with uniform traffic pattern. The distribution of three demand types is $(\rho_1, \rho_2, \rho_3) = (0.2, 0.5, 0.3)$ originally, and spectrum partitioning is calculated based on this distribution. If there is a small change $\eta$ in the distribution ($\rho_1 + \eta, \rho_2 + \eta, \rho_3 - 2\eta$), we compare the performance of applying the original spectrum partition to the changed traffic and that of applying adjusted partition calculated based on exact distribution. The results are represented by "Exact" (with adjustment according to exact information) and "Proposed" (without exact information). The two sets ($\eta = -0.05$ and $\eta = 0.05$) both show little performance degradation. 

\begin{figure}
	\includegraphics[draft = false, scale = 0.5]{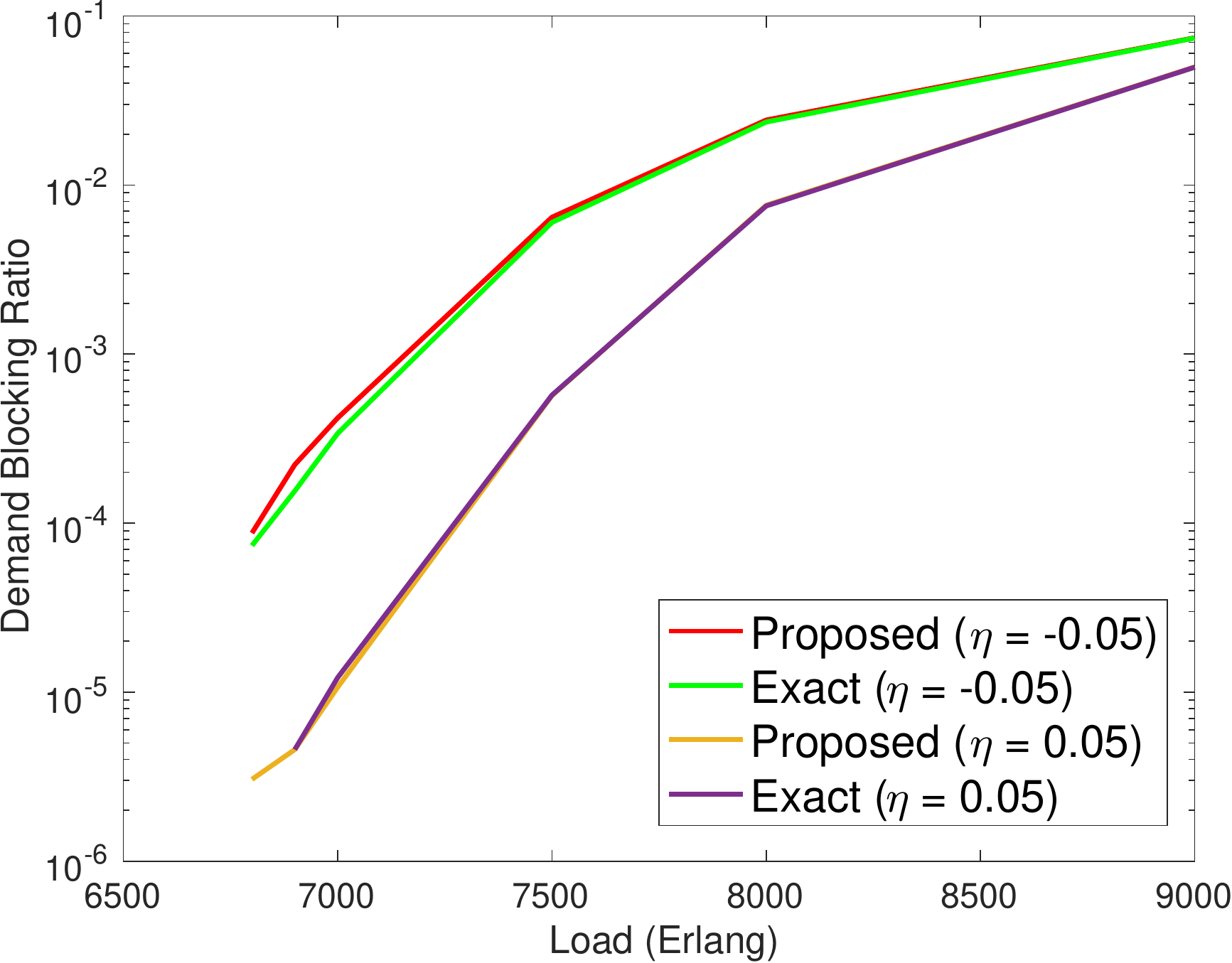}
	\centering
	\caption{\label{fig:sensPT} Sensitivity test for demand size distribution in NSF network.}
\end{figure}

\indent Figure \ref{fig:sensNU} is the test result regarding to traffic pattern in the NSF network. The demand size distribution does not change. The original probabilities that a node is selected as source or destination are in proportional to the node degrees and the multi-path selection has been calculated according to this non-uniform traffic pattern. If there is a small change $\eta$ in the probabilities ($u_1 + \eta, u_2 + \eta, \cdots, u_N - (N - 1)\eta$),  we compare the performance of applying the original path selection to the changed traffic and that of applying adjusted path selection calculated based on exact traffic pattern. The two sets ($\eta = -0.01$ and $\eta = 0.01$) both show little performance variation.

\begin{figure}
	\includegraphics[draft = false, scale = 0.5]{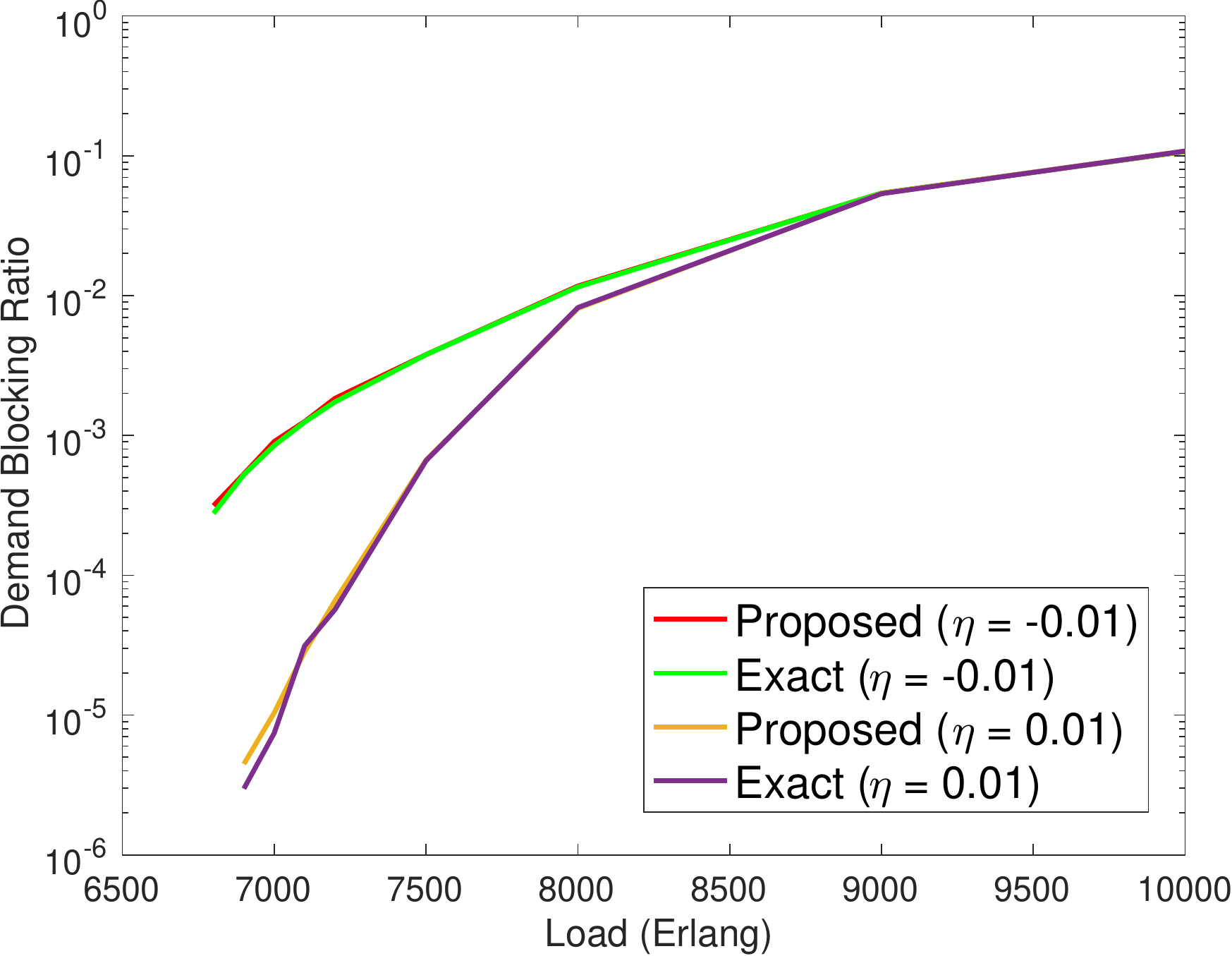}
	\centering
	\caption{\label{fig:sensNU} Sensitivity test for traffic pattern in NSF network.}
\end{figure}

\section{Conclusions}\label{sec.label75}
\indent We study the Routing and Spectrum Assignment problem for dynamic heterogeneous traffic requests in elastic optical networks with multiple fibers per link \cite{wu2018dynamic}. The routing problem is solved by a multi-path selection scheme, which is formulated mathematically with parameters given by the topology information and traffic pattern. Once the selection probabilities of each candidate path between a source and a destination are determined, any arriving traffic request for that source-destination pair will choose one path according to the probabilities. Then the slot usage along the path is examined to make the spectrum assignment. To avoid the fragmentation caused by the heterogeneity of demand sizes, we partition the whole spectrum into segments dedicated to different demand sizes. A spectrum assignment algorithm based on the network state within a specific partition is proposed first, and further improved by considering the sharing among all partitions. We validate the effectiveness of each scheme for subproblems. Multi-path selection, appropriate spectrum partitioning, and network state aware spectrum assignment all can improve the spectrum efficiency. The results of combining these schemes show much better performance than the baselines.

\chapter{Conclusions and Future Directions}
\label{chap_8}

\section{Conclusions}\label{sec.label81}

This dissertation investigated several problems of resource allocation in multi-granular optical networks. The research objective was to increase the resource utilization efficiency under various constraints. 

In Chapter \ref{chap_2}, we investigated the co-scheduling problem in elastic optical networks for both static and dynamic multi-task jobs. We proposed an Integer Linear Programming model and two heuristics to jointly allocate computational and networking resources to jobs. The simulation results show that our proposed Children-Aware algorithm outperformed the baseline algorithm in terms of makespan for static jobs and blocking probability for dynamic jobs.  

In Chapters \ref{chap_3} and \ref{chap_4}, we studied resource allocation problems in multi-granular WDM networks. A non-uniform wavebanding framework from the entire network's view was proposed to minimize the total waveband required while satisfying a set of demands so as to reduce hardware requirements. The simulation results suggested that wavebanding is an effective strategy to reduce switching elements. Then we investigated the OXC architectures for WDM networks with multiple fibers per link to accommodate increasing traffic demands. Two architectures were evaluated and compared in terms of hardware requirements, power consumption and costs. Heuristics for resource (i.e., fiber and wavelength) assignment were presented and analytical models were developed to predict the blocking performance.

In Chapter \ref{chap_5} and \ref{chap_6}, we studied the routing, fiber, band and spectrum allocation problem in multi-granular EONs with OXC nodes which can accommodate non-contiguous and non-uniform wavebands and multi-fiber links. An auxiliary layered-graph framework with pluggable cost functions was developed to solve the joint RFBSA problem. Cost functions to minimize the maximum spectral usage were proposed for a set of traffic requests. By comparison with other baseline heuristics, our RFBSA framework could achieve less spectral usage. The RFBSA framework was further improved in Chapter \ref{chap_6} to reduce time complexity. To achieve good network performance, while saving considerable hardware costs, a joint wavebanding node placement and RFBSA problem given the network budget (in terms of total number of available WSSs) was investigated. Heuristics based on the above framework was proposed to solve the problem.

In Chapter \ref{chap_7}, we proposed schemes for dynamic Routing and Spectrum Assignment problem in elastic optical networks. We investigate both path selection and spectrum management for heterogeneous bandwidth requests in multi-fiber elastic optical networks (EONs) to improve the spectrum efficiency. An ILP model based on topology and traffic pattern information was proposed to precomputed the selection probabilities of candidate paths for each source-destination node pair. A dedicated spectrum partition scheme and a spectrum assignment algorithm based on network states with resource sharing among partitions were proposed. The simulation results showed each scheme performs well in improving the spectrum efficiency. The joint allocation scheme significantly outperformed other heuristics in the literature.  

\section{Future Directions}\label{sec.label82}

Space Division Multiplexing (SDM) is expected to be the next frontier in optical communication by increasing capacity through multiplicity of space channels, such as utilizing multi-core fiber (MCF) or multi-mode fiber (MMF). One direction of my future research is to extend my previous work to resource allocation problem in space division multiplexing elastic optical networks by taking the physical impairment and crosstalk into account. The other problem is the survivable virtual infrastructure embedding over software-defined SDM optical networks. The problem is more complicated than the general virtual network embedding problem given the additional constraints related to the underlying optical networks. 
	
	\bibliographystyle{abbrv}
	\bibliography{global}
	
\end{document}